\documentclass[11pt,reqno,hidelinks]{amsart}
\usepackage{graphics,graphicx,amsmath,amssymb,epsfig,stmaryrd,color,flafter,xspace,lscape,threeparttable,pdfpages}
\usepackage{csquotes,amsfonts,amsthm,geometry,array,booktabs,latexsym,caption,multirow,tabularx,rotating}
\usepackage[font=small,labelfont=bf]{caption}
\usepackage{changepage,enumitem}
\usepackage{amsmath}
\usepackage{amsfonts}
\usepackage{amssymb}
\usepackage{graphicx}
\usepackage{natbib}
\usepackage{rotating}
\usepackage{bm}
\usepackage{dsfont}
\usepackage{subcaption}
\usepackage{comment}
\usepackage{hyperref}
\newcommand{\plim}[1]{\underset{{#1\to\infty}}{\textnormal{plim}}\,}

\DeclareCaptionFont{tiny}{\tiny}
\everymath{\displaystyle}
\RequirePackage{setspace,indentfirst,epsfig,psfrag,ifthen,ifpdf}

\newtheorem{theorem}{Theorem}
\theoremstyle{plain}

\newtheorem{algorithm}{Algorithm}

\newtheorem{assumption}{Assumption}
\newtheorem{subassumption}{Assumption\normalfont}[assumption]

\newenvironment{assumption*}
{\ifnum\value{subassumption}=0 \stepcounter{assumption}\fi\subassumption}
{\endsubassumption}
\newenvironment{assumption+}[1]
{\renewcommand{\thesubassumption}{#1}\subassumption}
{\endsubassumption}
\theoremstyle{plain} 
\newtheorem{assump}{Assumption}

\newtheorem*{definition*}{Definition}

\newtheorem{lemma}{Lemma}

\newtheorem{proposition}{Proposition}

\numberwithin{equation}{section}
\newcommand{\ind}{\perp\!\!\!\!\perp} 
\newcommand{\norm}[1]{\left\lVert#1\right\rVert}
\newcommand{\vect}{\text{vec}}
\newcommand{\R}{\mathbb{R}}

\newcommand{\sumnt}{\sum_{i=1}^N \sum_{t=1}^T}
\newcommand{\Ex}[1]{\mathbb{E}\left[#1\right]}
\newcommand{\gt}{\widetilde{g}}

\DeclareMathOperator*{\argmin}{argmin}
\graphicspath{ {images/} }

\begin{document}
\title[Unobserved Grouped Heteroskedasticity]{Unobserved Grouped Heteroskedasticity\\ and Fixed Effects}
\author{Jorge A. Rivero}
\thanks{This version: \today, first version: November 9, 2022, \\\emph{All errors are my own. I am grateful for the inspiration, support and encouragement from my committee chair Yanqin Fan.} Correspondence to \href{mailto:jrivero@uw.edu}{jrivero@uw.edu}. Department of Economics, University of Washington,  Box 353330,
Seattle, WA 98195.}

\maketitle

\begin{abstract}

This paper extends the linear grouped fixed effects (GFE) panel model to allow for heteroskedasticity from a discrete latent group variable. Key features of GFE are preserved, such as individuals belonging to one of a finite number of groups and group membership is unrestricted and estimated. Ignoring group heteroskedasticity may lead to poor classification, which is detrimental to finite sample bias and standard errors of estimators. I introduce the ``weighted grouped fixed-effects'' (WGFE) estimator that minimizes a weighted average of group sum of squared residuals. I establish $\sqrt{NT}$-consistency and normality under a concept of group separation based on second moments. A test of group homoskedasticity is discussed. A fast computation procedure is provided. Simulations show that WGFE outperforms alternatives that exclude second moment information. I demonstrate this approach by considering the link between income and democracy and the effect of unionization on earnings.

\vskip20pt


\vskip10pt

\end{abstract}

\section{Introduction}

Agents that exhibit unobserved differences pose a significant challenge in econometric models by confounding the true relationships between variables of interest. Panel data provides a solution where it is possible to make intuitive, though restrictive, assumptions on the form of heterogeneity such as with two way fixed effects or interactive fixed effects (see \cite{wooldridge:2010} and \cite{bai:2009}, respectively). Another approach is group fixed effects where the unobservable is a discrete variable that labels individuals into groups and allows within-group parameters that vary across time (\cite{bm:2015}). These methods require estimating the heterogeneity, which could be harmed by ignoring important features of the marginal distribution of confounding latent variables. I focus on the case of a discrete latent variable that induces group heterogeneous fixed effects and variances in the model, which is written as a linear grouped fixed effects model with unobserved groupwise heteroskedasticity. The heterogeneity that fits this model should generally be closer to the true groupings, rather than one estimated solely focused on the group fixed effects.

A linear model with group fixed effects is written as
\begin{gather}\label{themodel}
	y_{it} =  x_{it}'\theta + \alpha_{g_i t} + u_{i t},
\end{gather}
for $i = 1,\dots, N$ and $t = 1, \dots, T$ where the variable $g_i = 1, \dots, G$ are discrete and take a known or estimable $G$ number of values, are unobserved and may be arbitrarily correlated with the exogenous variables $x_{it}$. The group membership variables $g_i$ and group-specific time effects $\alpha_{g_i t}$ are exogenous and estimated from the data. It is assumed that there are a relatively small number of groups and group membership of individuals does not change over time and that individuals in the same group $g$ follow the same time path of effects $\alpha_{gt}$. This model and a least squares estimation procedure were introduced in \cite{bm:2015} and several alternative estimators have since been proposed (\cite{chetverikov:2022,mugnierPWD:2022}). 

The key identification assumption shared by all of these proposals is that the groups are separated, specifically that groups can be distinguished by the group fixed effects $\alpha_{g} = (\alpha_{gt})_t$ for all $g = 1,\dots,G$. From an error component perspective: $v_{it} = y_{it} - x_{it}'\theta = \alpha_{g_i t} + u_{it}$ has the group fixed effect term $\alpha_{g t}$ as the conditional mean given $g_i = g$ so that identification requires the conditional means of the error components to be separated. Therefore, in order to learn group memberships, one must consider partitions or groupings of $v_{i} = y_{i} - x_{i}\theta = (y_{it} - x_{it}'\theta)_t$ into $G$ groups that uniquely satisfy the moment conditions implied by this model. Mathematically, if the group means $\alpha_{g}$ and parameters $\theta$ were known, this amounts to sorting individuals $i=1,\dots,N$ into a group where their $v_{i}$ is closest to the corresponding mean $\alpha_{g}$.  However, the error component may also exhibit heteroskedasticity with respect to the unobserved groups, which is not a feature that has received attention in the literature. 

Unobserved group heteroskedasticity can be expressed as
\begin{equation}
	\mathbb{E}[ {u}_{it}^2 \vert g_i = g] = \sigma_{g}^2
\end{equation}
where $\sigma_{g} > 0$ for any $g=1,\dots,G$. Using worker's earnings as an example, this expression implies different variability in wages across groups, which  could arise when workers belong to different ability tiers. Workers that are part of a highly skilled group may display more years of schooling, experience, and consistently high earnings, with little variation. On the other hand, workers in a lower skilled, but improving ability group may show consistently lower earnings at first, but may start to demand for compensation as they improve with a larger variance of success, perhaps by being promoted or remaining as a subordinate. 

This paper extends the group fixed effects model to incorporate group heteroskedasticity. Simulation evidence suggests that ignoring this feature may lead to poor estimation of group memberships and severe finite-sample bias and large standard errors of parameter estimators. Intuitively, assigning according to just sample group means ignores the possibility that outlier individuals from a high variance group appears closer to another group, which may result in an erroneous assignment. I introduce the ``weighted grouped fixed-effects'' (WGFE) estimator, where the name serves as an analogy to the estimators connection to the GFE estimator and resemblance to the weighted least squares estimator. The WGFE estimator simultaneously estimates the parameters and finds the optimal grouping by organizing units with outcomes net the effect of covariates into groups that they are most similar to, which is determined by nearness to the mean and variance of each group using a modified Mahalanobis distance. As with GFE estimation, without covariates this reduces to a clustering problem, but usual $k$-means algorithms do not apply since they are known to perform poorly when clusters have heterogeneous variances. Using the proposed distance function, I make use of recent advances in latent factor discovery by \cite{yang:2022} to develop an algorithm adapted to handle heterogeneous group variance. In fact, the $k$-means algorithm itself is a special case of their algorithm when the clusters are uniformly distributed and have the homogeneous variance. The WGFE estimator is equivalent to the GSR estimator of \cite{boot:2022}, where the WGFE estimator arises from concentrating out the group variance parameters in their setup. These estimators for the GFE model with group heteroskedasticity were found independently. 

I provide an asymptotic theory for the WGFE estimator where I allow the number of individuals $N$ and time periods $T$ to approach infinity simultaneously. This estimator shares the property with the GFE estimator that $N$ can grow much faster than $T$ as opposed to the individual fixed-effects estimator which suffers from bias of order $1/T$ as $N/T$ converges to a constant, also known as the incidental parameter problem (\cite{nickell:1981,arellano-hahn:2007}). The WGFE and GFE estimators are both consistent and asymptotically normal as $N/T^\nu$ tends to zero for some $\nu >0$ provided the groups are well separated and errors $u_{it}$ satisfy tail and dependence conditions. 

For identification WGFE requires a stronger notion of separability between groups in order to ensure consistency of group assignments. The inclusion of second moment information must exclude population groupings where differences in group variances are relatively larger than differences in group means in a way that there is significant overlapping support. In other words, groups must be mean-separated as a function of relative differences in the group variances. If this is achieved, the quality of group assignments improves faster than GFE as $T$ approaches infinity and the WGFE estimator is expected to outperform estimators based solely on first moment classification when $T$ is small, which is a common property across panel data sets. GFE remains consistent for large $N$ and $T$ so the performance gap should close for large and long panels. It is also shown that the WGFE estimator is asymptotically equivalent to GFE when there is no group heteroskedasticity and the groups are equally weighted and that the WGFE objective function is bounded by the GFE objective, suggesting a test of group homoskedasticity. Future work involving an arbitrary number of features may involve these separability type assumptions and connections to GFE so this may be a starting point to provide tools to study discrete heterogeneity in depth. 

I showcase the WGFE estimator in two illustrations. The first is an application to study the effect of national income on democracy, adding to the insights found by \cite{bm:2015}. Following \cite{gleditsch-ward:2006} and \cite{ahlquist-wibbels:2012}, it is empirically observed that there are grouped patterns of democratization that transition in both time and space.  I find evidence where groupings determined by the WGFE show contiguous regions on similar time paths using a subset of a balanced panel of 90 countries from 1970-2000. Our findings do not contradict that of \cite{bm:2015} or \cite{acemoglu:2008} that income has little or no effect on democratization, however it strengthens the evidence and equivalently the correlation between the group variable, income and democracy. This supports the assertion that shared historical factors may have set countries on different time paths such as the formation of religious institutions and economic and defensive pacts; see \cite{huntington:1991}.

The second empirical example revisits the question of the effects of unionization on wages using data from the biennial Panel Study of Income Dynamics between 2001 and 2019. Unobserved ability is widely accepted as an important omitted variable that would explain employers selecting more able workers while facing higher labor costs stemming from a union contract. Many studies find the effect of unionization is overestimated without accounting for individual-specific unobserved ability (\cite{freeman:1984,robinson:1989,card:1996}). Using WGFE, workers across various unobserved ability groups experience different rates of expansion and deterioration of worker bargaining power and real earnings across the sample period. Using traditional fixed effects simply shows an aggregate decline, which keeps the different unobserved segments hidden.

Section \ref{sec:ggfe} presents the issues of identification with different notions of separability, the estimator and computational approach. Section \ref{sec:asym}  develops the asymptotic theory of the infeasible WGFE estimator and the WGFE estimator including consistency of parameter estimators and sample group assignments, and asymptotic normality. In Section \ref{sec:sims}, I compare the WGFE and GFE estimators in a simulation study and use our approach to study the two empirical applications on unionization-wages and income-democracy.

\subsection{Related Literature}\label{sec:rellit}

Modeling with latent group heterogeneity in panel models has received attention as a useful alternative to standard fixed effects approaches. \cite{sun:2005} estimate a multinomial logistic regression with unobserved groupings, but known number of groups $G$ via maximum likelihood. \cite{lin:2012} in one of their proposed methods develop a $k$-means algorithm to group individuals based on the regression estimate of one of $G$ groups. When (multiple) groupings are known, \cite{bester:2016} propose an estimator for a random effects model that assumes individuals share a fixed-effect at some (undetermined) level of grouping. \cite{hahn:2010} consider a class of game theoretic models and show that the incidental parameter problem vanishes quickly as $N$ approaches infinity since the number of equilibria are predicted to be finite, hence the support of the fixed-effects is finite.  In the case of multidimensional heterogeneity, \cite{cheng:2019} consider individuals who may belong to different groups modeled by several latent variables. None of these papers consider time-varying heterogeneity. \cite{mugnier:2022} introduce the single index nonlinear GFE model and estimation that does allow for time-varying heterogeneity, but not the possibility of group heteroskedasticity.

When the number of groups are unknown it must be inferred from the data. \cite{bm:2015} follow \cite{bai:2002} and \cite{bai:2009} by using a Bayesian information criteria (BIC) and setting a maximum number of groups. Recently there have been a number of papers that use penalization to classify and estimate parameters. \cite{su:2016} propose the C-Lasso related to group Lasso that classifies and shrinks individual coefficients to the unobserbed group coefficients. \cite{lu:2017} continue with a testing procedure to determine the number of groups based on C-Lasso. \cite{su:2018} extend the C-Lasso to dynamic panel data models with interactive fixed-effects and cross sectional dependence and propose a BIC to estimate the number of factors and groups. \cite{ando:2016} consider grouped factor structure and penalization for coefficient estimates and their $C_p$-type criteria to estimate the number of groups and number of factors in each group. \cite{mehrabani:2022} take these penalization approaches further by allowing the number of groups to diverge along with the sample dimensions. \cite{mugnierPWD:2022} propose a simple nuclear-norm regularized estimator that estimates the number of groups along with parameters.

An equivalent estimator to WGFE was found independently by \cite{boot:2022}, where their identification conditions are similar in requiring strong separability. 
\cite{zhang:2020} propose a Bayesian approach that treats the number of groups as a parameter to be estimated and use a Dirichlet process prior that allows for infinitely many groups and allows for group heteroskedasticty, but require strictly normal errors while WGFE imposes no such structure. Along these lines, \cite{kim:2019} give evidence that in the normal errors case that accounting for group heteroskedasticity in the income and democracy application shows the correlation between the estimate group variable and income is larger than predicted by GFE estimation and that the total number of groups estimated might depend on the specification of group heteroskedasticity, which coincides with the prediction using WGFE.

Finite mixture models are closely related to our approach, see the textbook by \cite{mclachlan:2004} for a comprehensive review. In these models the data is not required to label individuals into groups and group membership is instead estimated. It is typical to assume that the group distributions belong to the same family e.g. Gaussian mixture models so group heteroskedasticity is modeled. For identifiability of models with covariates there needs to be restrictions the values covariates can take on and on the interaction between covariates and the latent variable; see \cite{kasahara:2009} for a result on identification of discrete choice models with panel data where both the group distribution and choice probabilities are nonparametrically specified. I allow covariates to be arbitrarily correlated with the latent variable as in typical fixed-effects and leave group membership probabilities unrestricted. Mixture models can incorporate fixed-effects, see \cite{deb:2013} who provide a solution to the incidental parameter problem for Gaussian and Poisson families. \cite{heckman:1984} apply them to duration models where they take the distribution of the unobservable nonparametrically, while imposing a parametric assumption on the group distributions. A related setting is Markov-Switching models for time series, which can be seen as a mixture model, see \cite{fruhwirth:2006}. Related to switching models, $k$-means clustering is similar to greedy approximations of step function signals, which can be thought of as a estimating a switching model of intercepts with no covariates or parametric assumptions on the error, see \cite{rivero:2023}. In this setting, the user does not need to specify the number of groups and instead could be determined via penalty methods.

The Expectation Maximization (EM) algorithm commonly used in the maximum likelihood estimation of mixture models can be seen as a clustering algorithm, see \cite{redner:1984}. In fact, in the case of a Gaussian mixture, the Mahalanobis distance is a key part of the algorithm, which suggests a connection to our clustering algorithm which also incorporates second moment information. Indeed GFE is also the maximizer of the pseudo-likelihood of a Gaussian mixture model, where the mixing probabilities are individual-specific and unrestricted e.g. independent of covariates. 

 An inspiration for considering unobserved group heteroskedasticity is the  the $k$-means problem of assigning a number of individuals into a finite number of groups; see (\cite{maccqueen:1967,lloyd:1982,forgy:1965}). The $k$-means criterion is a least squares criterion, which implicitly assumes that the groups are separated enough in mean and are of identical variance. When the data violates any of these assumptions, then group assignment at the sample level may be compromised. See the recent survey \cite{ahmed:2020} on the performance of the standard $k$-means algorithm in these contexts.

The estimator is interestingly tied to optimal measure transport theory (\cite{villani:2009, santambrogio:2015}), specifically the Wasserstein (Kantorovich) distance between probability distributions; see \cite{kolouri:2017} for more on Wasserstein distance and its applications. The Wasserstein distance function defines a distance function between probability measures and so permits a Fr\'{e}chet mean of distributions that is also itself a distribution known as the Wasserstein barycenter, see \cite{ac:2011}. The WGFE estimator for a model without covariates can be seen as the minimization problem of optimally forming two parameter (location-scale) groups such that these groups have a barycenter with minimum variance over other possible groupings and barycenters. See \cite{galichon:2018} for an introduction to optimal transport for economists.

\section{Weighted Grouped Fixed Effects Estimation}\label{sec:ggfe}

The slope parameters are contained in the vector $\theta\in\Theta \subset \R^p$ and group-specific time-effects $\alpha_{gt}\in \mathcal{A} \subset \R$ for any $g=1,\dots,G$ and $t = 1,\dots,T$ and $\alpha \in \mathcal{A}^{GT}$. The group assignment parameters are categorical: $g_i\in\{1,\dots,G\}$ for every $i = 1,\dots,N$ and the space of all possible partitions using the $g_i$ is denoted as $\gamma = (g_1, \dots, g_N) \in \{1,\dots,G\}^N =\Gamma_G^N$. The covariates $x_{it}$ can contain lagged outcomes along with strictly exogenous regressors. The covariates are also allowed to be arbitrarily correlated to the time-effects $\alpha_{gt}$. Note the absence of time subscript implies a vector of time profiles or time series for the individual. Denote $ x \mapsto \norm{x}$ as the Standard Euclidean norm of finite-dimensional vector $x$.

\subsection{Identification of groups}

Group fixed effects models assume that group membership is unknown and estimated from data. Identification of group memberships in these class of models require that groups are ``separated'' according to some criteria. It is sufficient to assume that group fixed effects for any pair of groups are different in the mean squared sense. This means that groups can be distinguished by the comparison of their group-specific time series of effects. 

Denote the true parameters with a superscript ``0''. Formally for any $(g,\widetilde{g}) \in \{1,\dots,G\}^2$ such that $g\neq \widetilde{g}$:
\begin{equation}
	\plim{T} \frac{1}{T} \sum_{t=1}^T \left( \alpha_{gt}^0 - \alpha_{\widetilde{g}t}^0\right)^2 = \plim{T}\frac{1}{T} \norm{\alpha_{g}^0 - \alpha_{\widetilde{g}}^0}^2  > 0.
\end{equation}
This is a separability assumption and implies a rule for determining the group of any individual using their time series information. For $i\in\{1,\dots,N\}$ and supposing that $g_i^0 = g$, I will show that the mean squared distance between $v_i = y_i - x_i\theta^0$ and $\alpha_{h}^0$ for any $h \in \{1,\dots,G\}$ is minimized when $h = g$. Recall the true model \eqref{themodel} and consider the following difference we want to show is bounded away from zero:
\begin{align}
	\plim{T}\frac{1}{T}\norm{y_{i} - x_{i}\theta^0 - \alpha_{h}^0}^2 - \frac{1}{T}\norm{y_{i} - x_{i}\theta^0 - \alpha_{g}^0}^2 &= \plim{T}\frac{1}{T}\norm{\left(\alpha_{g}^0 - \alpha_{h}^0 \right) + u_{i}}^2 - \frac{1}{T}\norm{u_{i}}^2\\
	&= \plim{T}\frac{1}{T}\norm{\left(\alpha_{g}^0 - \alpha_{h}^0 \right)}^2 + o_p(1) \\
	&> 0
\end{align}
where the $o_p(1)$ results from an exogeneity condition applied on $\alpha_g^0 - \alpha_h^0$ that will be made in more detail in Section \ref{sec:asym}. The converse is also true-- for any $i\in\{1,\dots,N\}$, if we assume the strict inequality holds for some $g$, it uniquely minimizes the distance between $v_i$ and the group fixed effect $\alpha_{g}$ of group $g$. Therefore, since groups are separated, an individual must belong to a group $g$ if their observed $v_i = y_i - x_i\theta$ is closest to $\alpha_{g}$ compared to the other group fixed effects, thus identifying group assignments for each individual provided parameters are known.

Geometrically this requires that clusters or conditional distributions given groupings have different means. This does not restrict overlap between groups, notably how a particular group variance causes large intersections with other groups. Moreover even when clusters have many features that are heterogeneous this condition delivers identification of group memberships. 

The inclusion of group heteroskedasticity in the model brings more information, but when groups must be estimated it necessitates steeper requirements to detect them in the population. A natural criteria can be based on the Mahalanobis distance, which normalizes distances to the cluster mean based on the ``dispersion''  or noise of the group. Letting $\sigma_g^2$ denote the variance of any group $g$, a candidate rule may be written for all $i \in \{1,\dots,N\}$: $g_i^0 = g$ if and only if, for all $h \neq g$,
\begin{equation}\label{rule:wgfe1}
	\plim{T}\frac{1}{T}\frac{\norm{y_i - x_i'\theta^0 - \alpha_g^0}^2}{\sigma_g} < \plim{T}\frac{1}{T}\frac{\norm{y_i - x_i'\theta^0 - \alpha_h^0}^2}{\sigma_h}.
\end{equation}
For this rule to be valid a stronger separability condition is required that involves both group fixed effects and variances. Following similar calculations on this inequality, we arrive at
\begin{align*}
	\plim{T} \frac{1}{\sigma_h T}\norm{\alpha_{g}^0 - \alpha_h^0}^2 - \left(\frac{\sigma_h - \sigma_g}{\sigma_g\sigma_h}\right)\frac{1}{T} \sum_{t=1}^T u_{it}^2 + o_p(1)
\end{align*}
where once again we would like this to be asymptotically bounded away from zero. A sufficient condition is the following ``strong'' separability of groups:
\begin{equation}\label{strongsep}
	\plim{T} \frac{1}{T} \sum_{t=1}^T \left( \alpha_{gt}^0 - \alpha_{ht}^0\right)^2 > \vert\sigma_h - \sigma_g\vert \frac{\max_{f\in \{1,\dots,G\}}\sigma_f^2}{\min_{f\in \{1,\dots,G\}}\sigma_f} 
\end{equation}
which explicitly requires that the model noise is finite, i.e., $\plim{T} \sum_{t=1}^T u_{it}^2  < \infty$ for all $i \in\{1,\dots,N\}$ and that there are no degenerate groups: $\sigma_f >0$ for all $f = 1,\dots,G$.  In words, the difference between variances is scaled by the $F$-statistic between the noisiest group and the least. A model that has a large range between group variances therefore requires a large discrepancy between group fixed effects in order to identify groups. As an example with no covariates, with $T=2$, two normally distributed groups, fixing the mean-squared error of the two group fixed effects to 4 and one of the groups variances to 1, \eqref{strongsep} implies a valid range of variances for the other group in the interval $(0.2,2)$ where the groups would be identifiable. Figure \ref{fig:strongsep1} displays a visual on how these groups may appear when the variable standard error is close to 2 and 0.2. According to this example the constraint of strong separability does not seem to place a large restriction on the types of clusters WGFE may handle. 

\begin{figure}[h!]
	\includegraphics[scale = 0.5]{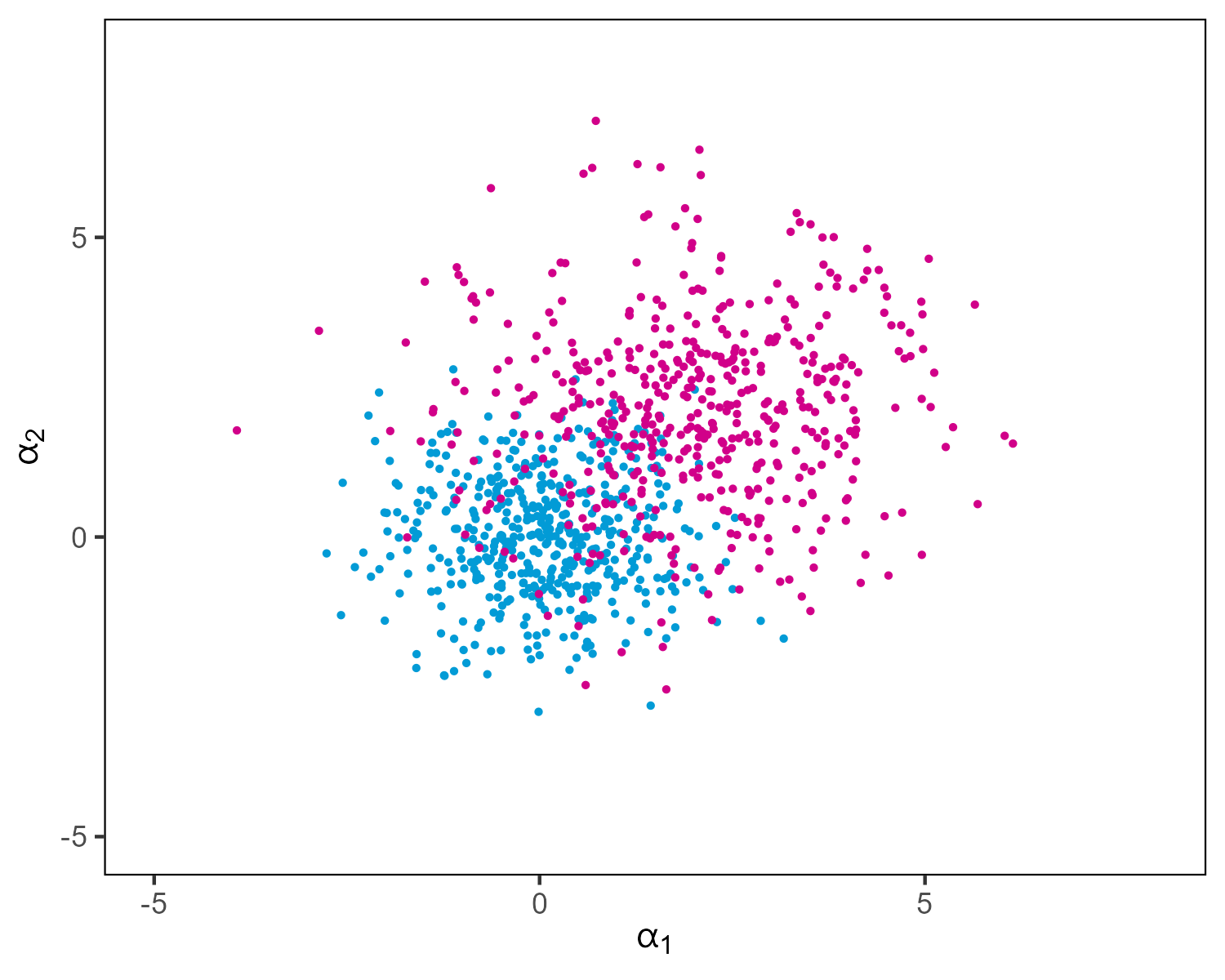}
	\includegraphics[scale = 0.5]{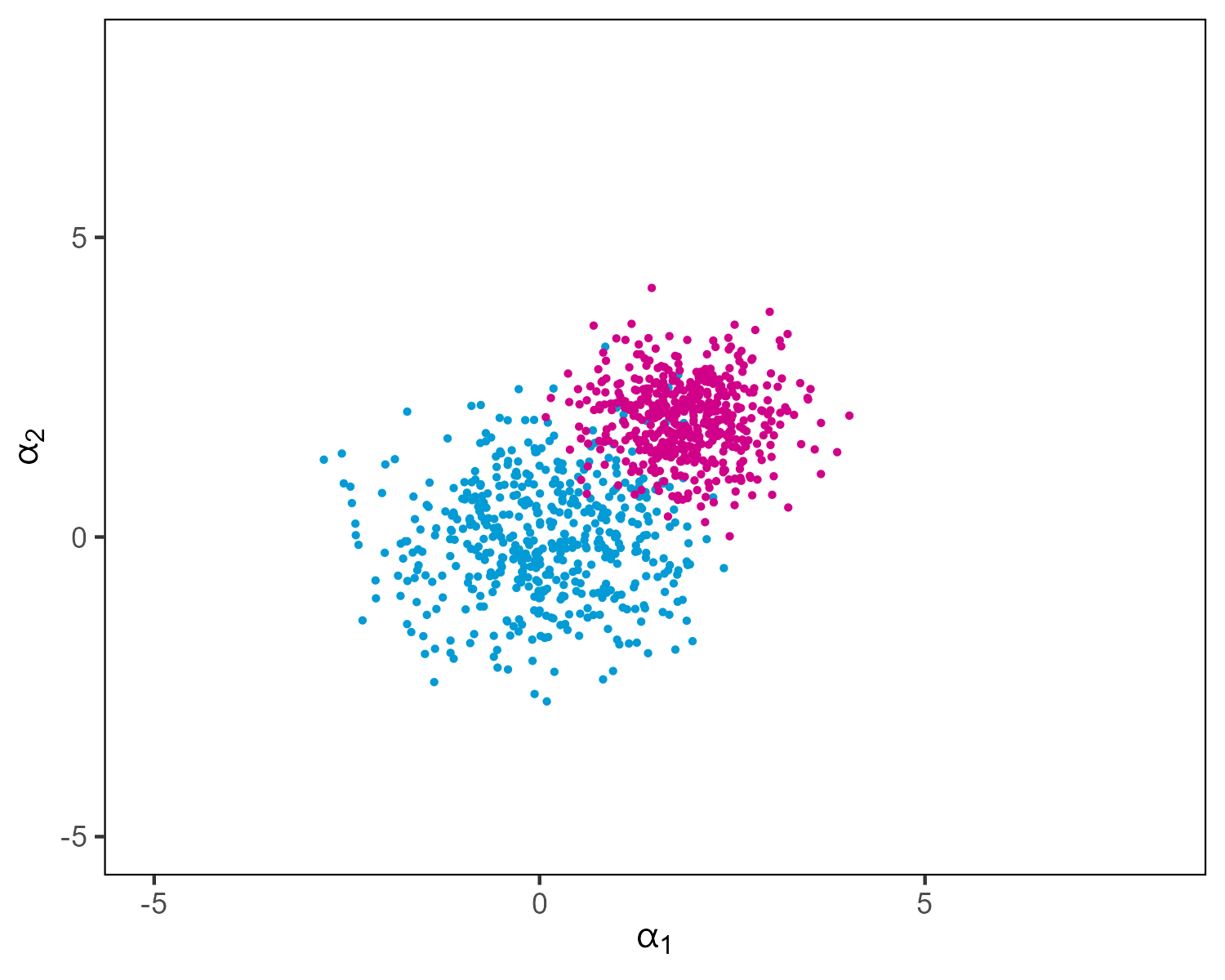}
\caption{Groups that satisfy strong separability. Blue group with standard error 1, Red group with variable $\sigma$.  Left: $\sigma = 1.55$, Right: $\sigma = 0.65$}
\label{fig:strongsep1}
\end{figure}

Taking the rule \eqref{rule:wgfe1} to data may perform poorly by creating and then favoring groupings with enormous variance. To counteract this, a simple additive penalty may be imposed that vanishes asymptotically with larger $T$ so that strong separability is preserved. Following \cite{yang:2022}, I propose the following correction:
\begin{equation}\label{rule:wgfe2}
	\plim{T}\frac{1}{T}\frac{\norm{y_i - x_i'\theta^0 - \alpha_g^0}^2}{\sigma_g} + \frac{1}{T}\sigma_g < \plim{T}\frac{1}{T}\frac{\norm{y_i - x_i'\theta^0 - \alpha_h^0}^2}{\sigma_h} + \frac{1}{T}\sigma_h
\end{equation}
so that noisy groups are penalized and less noisy groups are assigned primarily by closeness to group mean.

This reasoning is easily generalizable to groups with non diagonal  heterogeneous covariances $\mathbf{\Sigma}_{g_i^0}$ (see Appendix \ref{app:ggfe}) and potentially to incorporate other moments of the distributions. An extension up to $K$ moments might involve a distance function that takes arguments as the first $K$ moments and then a separability condition that requires increasingly steeper restrictions on how far the group means must be as a function of the other $K-1$ moments.

\subsection{Estimation} 

 The \emph{weighted grouped fixed-effects} (WGFE) estimator for model \eqref{themodel} is defined as
\begin{equation}\label{wgfe:obj}
	\left(\widehat{\theta},\widehat{\alpha},\widehat{\gamma}\right) = \argmin_{(\theta,\alpha,\gamma) \in \Theta \times \mathcal{A}^{GT} \times \Gamma_G^{N}} 
\sum_{g=1}^G P_g^\gamma \sqrt{\widehat{Q}_g(\theta,\alpha,\gamma)}
\end{equation}
where
\begin{equation}
	\widehat{Q}_g(\theta,\alpha,\gamma) = \frac{1}{T\sum_{j=1}^N P_{g}^{j}(\gamma)} \sum_{i=1}^N \sum_{t=1}^T P_g^i(\gamma)(y_{it} - x_{it}'\theta - \alpha_{g t}^0)^2
\end{equation}
and $P_g^i(\gamma)  = \mathds{1}\{g_i = g\}$ and $P_g^\gamma = N^{-1}\sum_{i=1}^N P_{g}^i$. 

The objective function is a departure from the pooled least squares criteria of \cite{bm:2015} and to $k$-means when covariates are absent. It is instead a minimization of a weighted transformation of within-group least squares that is shown to coincide to choosing $\widehat{g}_i = \widehat{g}_i(\theta,\alpha)$ that satisfies \eqref{rule:wgfe2} for every $i \in \{1,\dots,N\}$ in \cite{yang:2022}. It is a version of the GSR estimator independently developed by \cite{boot:2022} which concentrates out the optimal group variance estimator. Because of this results using these estimators will appear similar.

Let $\sigma_{g_i}(\theta,\alpha,\gamma) = \sqrt{\widehat{Q}_{g_i}(\theta,\alpha,\gamma) }$. With a fixed grouping scheme $\gamma^*$ the first-order conditions with respect to $\theta$ and $\alpha$ are
\begin{gather}
0 = \sumnt {\sigma}_{{g}_i^*}^{-1}(\theta,\alpha,\gamma^*) x_{it} \left(y_{it} - x_{it}'\theta - \alpha_{{g}_i^* t} \right)\label{foc:theta}\\
0 = \frac{\sum_{i=1}^N\mathds{1}\{{g}_i^* = g\} \left(y_{it} - x_{it}'\theta - \alpha_{g t}\right)}{\sum_{j=1}^N \mathds{1}\{{g}_i^* = g\}} \label{foc:alpha}
\end{gather}
for all $ t=1,\dots,T$, $g = 1,\dots, G$. With these groupings, $\theta^*$ is the solution to the nonlinear equation
\begin{gather}\label{wgfe:theta}
	\theta^*
= \left[\sum_{i=1}^N \sum_{t=1}^T {\sigma}_{{g}_i}^{-1}({\theta}^*,{\alpha}^*,\gamma^*) (x_{it} - \overline{x}_{{g}_i t})(x_{it} - \overline{x}_{{g}_i t})'\right]^{-1} \sum_{i=1}^N \sum_{t=1}^T  {\sigma}_{{g}_i}^{-1}({\theta}^*,{\alpha}^*,\gamma^*) (x_{it} - \overline{x}_{{g}_i t})(y_{it} - \overline{y}_{{g}_i t})
\end{gather}
in the system of equations with 
\begin{equation}\label{wgfe:alpha}
	{\alpha}_{{g} t}^* = \overline{y}_{{g} t} - \overline{x}_{g t} '{\theta}^*
\end{equation}
where $\overline{y}_{{g} t}$ and $\overline{x}_{{g} t}$ are the within-group averages of the respective variables according to grouping $\gamma^*$. The form of \eqref{wgfe:theta} displays similarities to weighted least squares and solvable by the Newton-Raphson method (\cite{dennis:1996}) or by fixed-point iteration.

\subsection{Computation} 

The partitioning problem of assigning $N$ units into $G < N$ groups is NP-hard (\cite{aloise:2008,dasgupta:2008}) and, with additional parameters such as $\theta$, rules out exhaustive search \footnote{\cite{clausen:2003}, \cite{brusco:2006} and  \cite{aloise:2012} for some global and exact methods that are prohibitive in my application.}. I follow convention and propose a heuristic algorithm based on Lloyd's ($k$-means) algorithm (\cite{lloyd:1982}) that after initializing will alternate between assignment of individuals into groups and then updating parameters based on those assignments.

\begin{algorithm}\label{wgfe:algo}(Lloyd-type Algorithm for WGFE).
\begin{itemize}
\item[1.] Initialize $\theta^{(0)}$; set the $\alpha^{(0)}$ as $G$ randomly chosen $v_i = y_i - x_i\theta^{(0)}$; create initial assignments $g_i^{(0)}$ by assigning $v_i$ to the closest $\{\alpha_{g_i^0}^{(0)}\}$; and calculate weights $\{P_g^{(0)}\}$ and $\{{{\sigma}}_g^{(0)}\}$ for all $g \in \{1,\dots,G\}$. Set $s = 1$. \newline

\item[2.] (Assignment Step). For all $i$, assign according to
\[	
	{g}_i^{(s+1)} \leftarrow \argmin_g \left\{ \frac{\norm{y_i - x_i\theta^{(s)} - \alpha_{g}^{(s)}}^2}{{\sigma}_g^{(s)}} + {\sigma}_g^{(s)}\right\}
\]
and collect them in $\gamma^{(s+1)} = (g_1^{(s+1)}, \dots, g_N^{(s+1)})$.\\
\item[3.] (Update Step). Update $\alpha^{(s+1)}$ and $\theta^{(s+1)}$ according to \eqref{wgfe:theta} and \eqref{wgfe:alpha}, respectively, 
then update $\{P_g^{(s)}\}$ and $\{{{\sigma}}_g^{(s+1)}\}$ for all $g \in \{1,\dots,G\}$. \newline

\item[4.] If $g_i^{(s)} = g_i^{(s + 1)}$ for all $i$, stop. Otherwise, set $s \leftarrow s + 1$ and go back to step 2.

\end{itemize}

\end{algorithm}

This is a variance augmented $k$-means algorithm similar to gradient descent  where part of the set of parameters take on discrete values and each iteration improves on the function value until a local minimum is reached  (see \cite{forgy:1965}, \cite{maccqueen:1967} and \cite{lloyd:1982}). The algorithm converges quickly, however a complete application of Algorithm \ref{wgfe:algo} involves repeating it many times with a randomly drawn starting value and choosing the result with the smallest local minima. I emphasize that this is a heuristic method sensitive to initialization and has no guarantees of finding the global minimizer. Heuristic algorithms perform well enough; see \cite{brusco:2007} and in particular the class of \emph{variable neighborhood search} (VNS) algorithms (\cite{hansen:2010}). These algorithms combine both stochastic initialization and deterministic search to comb the domain for the smallest local minimum. I adapt a VNS algorithm described in the Supplementary Appendix of \cite{bm:2015} and originally from \cite{pacheco:2003} to incorporate covariates and utilize it in simulations and the empirical applications since it can be more efficient at finding minima versus repeated applications of Algorithm \ref{wgfe:algo}. See Appendix \ref{wgfe:vns} for a full description of the VNS Algorithm \ref{alg:vns}. Despite using the VNS algorithm in this paper, I often refer to Algorithm \ref{wgfe:algo} as it provides intuition for identification and the asymptotic theory and still plays a role in the improved VNS algorithm.

The rule \eqref{rule:wgfe2} appears in the assignment step and serves as a balanced criteria to assign an individual into a group it is closest in group fixed effect relative to the noise in the group. While very noisy groups make the first term smaller, they are penalized by the second term preventing assignment solely based on the standardized distance. This is not the first clustering criteria incorporating second moment information, see \cite{zhao:2015}. The Mahalanobis distance can be used, however too much normalization without the penalty term in \eqref{rule:wgfe2} may result in an assignment rule that favors the noisiest group and the clustering objective function becomes trivial, see \cite{krishnapuram:1999}. 

\subsection{Connection to GFE estimation}\label{sec:gfe}

The GFE estimator solves the pooled least squares problem
\begin{equation}
	\min_{(\theta,\alpha,\gamma) \in \Theta \times \mathcal{A}^{GT} \times \Gamma_G^N} \frac{1}{NT}\sumnt \left(y_{it} - x_{it}'\theta - \alpha_{g_i t} \right)^2
\end{equation}
and the assignment step of GFE estimation is given by
\begin{equation}\label{gfe:assignment}
	g_i \leftarrow \argmin_g \norm{y_i - x_i\theta - \alpha_{g}}^2.
\end{equation}
The estimators have a closed-form for given grouping $\widetilde{g}$:

\begin{gather}\label{gfe:theta}
	\widetilde{\theta}^{GFE} 
= \bigg[\sum_{i=1}^N \sum_{t=1}^T (x_{it} - \overline{x}_{\widetilde{g}_i t})(x_{it} - \overline{x}_{\widetilde{g}_i t})'\bigg]^{-1} \sum_{i=1}^N \sum_{t=1}^T   (x_{it} - \overline{x}_{\widetilde{g}_i t})(y_{it} - \overline{y}_{\widetilde{g}_i t})
\end{gather}
and
\begin{equation}
	\widetilde{\alpha}_{\widetilde{g} t}^{GFE} = \overline{y}_{\widetilde{g} t} - \overline{x}_{\widetilde{g} t}'\widetilde{\theta}^{GFE}.
\end{equation}

The GFE criterion function weighs the residual variance the same across all groups while the WGFE criterion separates them by weighing them according to the size of groups. If groups had homogeneous variance and size, then WGFE and GFE estimation would be equivalent.

\begin{theorem}\label{WC:GFE}
	If $\sigma_g^2 = \sigma^2$ and $\mathbb{P}(g_i^0 =g)= 1/G$ for all $g = 1, \dots, G$,  then the WGFE objective function is asymptotically equivalent to the GFE  objective function.
\end{theorem}

Furthermore, they are connected in a way economists might describe risky behavior. The GFE criterion is to the risk neutral agent while the WGFE is to the risk averse agent. The following proposition is a consequence of Jensen's inequality.
\begin{proposition}\label{WGFE-less-GFE}
	Let $\widehat{Q}^{WGFE}$ and $\widehat{Q}^{GFE}$ denote the sample criterion functions for the WGFE and GFE estimators, respectively. Denote their population counterparts by ${Q}^{WGFE}$ and ${Q}^{GFE}$. Then,
\begin{align*}
\left[\widehat{Q}^{WGFE}(\theta,\alpha,\gamma)\right]^2 \leq&\widehat{Q}^{GFE}(\theta,\alpha,\gamma) \\
\left[{Q}^{WGFE}(\theta,\alpha,\gamma)\right]^2 \leq& {Q}^{GFE}(\theta,\alpha,\gamma)
\end{align*}
 for any $(\theta,\alpha,\gamma)\in \Theta \times \mathcal{A}^{GT} \times \Gamma_G^N$ and with equality for the second inequality if $\sigma_g = \sigma_{\gt}$ for all $(g,\gt) \in \{1,\dots,G\}^2$, $\theta = \theta^0$ and $\alpha = \alpha^0$.
\end{proposition}
 The square of the WGFE objective function is pointwise dominated by the GFE objective function in both the sample and population level. The inequality is sharp at the population level provided that errors are group homoskedastic and groups are uniform. Proposition \ref{WGFE-less-GFE} along with asymptotic equivalence of WGFE and GFE estimation under group homoskedasticity suggests constructing a test of group homoskedasticity and uniformness based on the test statistic
\[
	\tau_{NT} = d_{NT}\left(\widehat{Q}^{GFE}\left(\widehat{\theta},\widehat{\alpha},\widehat{\gamma}\right)  -
\left[\widehat{Q}^{WGFE}\left(\widehat{\theta},\widehat{\alpha},\widehat{\gamma}\right)\right]^2\right) \geq 0
\]
where it is evaluated at the WGFE estimator and $d_{NT} \to \infty$. Under the null that groups are homoskedastic and consistency of criterion functions and estimators, we should have $\tau_{NT} \to_p 0$ as $N,T$ approaches infinity. Indeed, in a sufficiently large sample if $\tau_{NT}$ is exceedingly far from zero we may reject group homoskedasticity and uniformity.

\section{Asymptotic Theory}\label{sec:asym}

Consider the data generating process 
\[
	y_{it} = x_{it}'\theta^0 + \alpha_{g_i^0 t}^0 + u_{it},
\]
where $g_i^0 \in \{1,\dots,G\}$ denotes the true group membership for each individual $i$ and zero superscripts denote true values. Assume that the true number of groups is known so $G= G^0$. I assume throughout that we have a random sample $\{(y_{it}, x_{it})\}$ over $i$ and the unobserved $g_i^0$ are a random sample from a mass function $\mathbb{P}(g_i^0 = g)$ with support on its entire domain $\{1, \dots, G^0\}$. The triple $\{(y_{i}, x_{i},g_i^0)\}_{i =1}^N$ is regarded as a random sample from some joint distribution.

This section provides conditions that enable the WGFE estimators of group membership converge to their true values and use this to show that the WGFE parameter estimators are asymptotically equivalent to their infeasible version where the true group memberships are known. I establish conditions for which this infeasible estimator is consistent and asymptotically normal using standard theory of extremal estimation (\cite{whitney:1986}) and use asymptotic equivalence of the estimators to connect these properties to the WGFE estimator. In this setting both $N$ and $T$ are allowed to approach infinity, but $T$ may grow much slower than $N$ to the effect of $N/T^\nu \to 0$ for some $\nu \gg 1$ as will be shown.  

\subsection{Infeasible WGFE estimation}

 Let $\left(\widetilde{\theta}, \widetilde{\alpha}\right)$ be the infeasible version of the WGFE estimator where the true group memberships are known:
\begin{equation}\label{est:infeasible}
	\left(\widetilde{\theta},\widetilde{\alpha}\right) = \argmin_{(\theta,\alpha) \in \Theta \times \mathcal{A}^{GT}} 
\sum_{g=1}^G P_g \sqrt{\widetilde{Q}_g(\theta,\alpha)}
\end{equation}
where
\begin{equation}
	\widetilde{Q}_g(\theta,\alpha) = \frac{1}{T\sum_{j=1}^N P_{g}^j} \sum_{i=1}^N \sum_{t=1}^T P_g^i(y_{it} - x_{it}'\theta - \alpha_{g t})^2
\end{equation}
and $P_g^i  = \mathds{1}\{g_i^0 = g\}$ and $P_g = N^{-1}\sum_{i=1}^N P_{g}^i$. 

Establishing the asymptotic distribution of the GFE estimator follows a similar strategy since the infeasible version is simply a least squares estimator of a pooled regression with group dummy variables. However,  (\ref{est:infeasible}) is not a pooled least squares estimator so I must show this infeasible version of the estimator has the property of consistency and asymptotic normality under some conditions.

The proposed population criterion function is
\begin{equation}
	Q(\theta,\alpha) = \sum_{g=1}^G P_g^0 \sqrt{Q_g(\theta,\alpha)}
\end{equation}
where $P_g^0 = \mathbb{P}(g_i^0 = g)$ and
\begin{equation}
	Q_g(\theta,\alpha) = \Ex{\frac{1}{T}\sum_{t=1}^T\mathds{1}\{g_i^0 = g\} (y_{it} - x_{it}'\theta - \alpha_{gt})^2}.
\end{equation}

\begin{assump}\label{as:infeas}\textnormal{(Infeasible WGFE).} There exists $M>0$ such that
	\begin{itemize}
		\item[a.] The parameter space $\Theta \times \mathcal{A}^{GT} \subset \R^{p}\times \R^{GT}$ is compact.
		\item[b.] $\Ex{u_{it}^4} <M$, $\Ex{u_{it}} = 0$ and $\Ex{\norm{x_{it}}^4} < M$
		\item[c.] $\sum_{i=1}^N \mathds{1}\{g_i^0 = g\} \to \infty$ as $N \to \infty$.
		\item[d.] The support of $g_i^0$ is $\{1,\dots,G\}$: for every $g = 1,\dots,G$, $\mathbb{P}(g_i^0 = g) >0$.
		\item[e.] The matrix 
		\[
		\plim{T}\sum_{g=1}^G \frac{1}{T}\sum_{t=1}^T \Ex{\left(x_{it} - \Ex{x_{it}\vert g_i^0 = g}\right)\left(x_{it} - \Ex{x_{it}\vert g_i^0 = g}\right)' \vert g_i^0 = g}
		\]
		has a positive minimum eigenvalue. 
	\end{itemize}
\end{assump}

The first two assumptions are standard, while $(c)$ and ($d$) are non degeneracy conditions of each group indexed by $\{1,\dots,G\}$. I require all groups to have many observations in them and thus group variances are bounded away from zero. Assumption \ref{as:infeas}($e$) is a full rank condition analogous to those on design matrices for the basic linear models. A sufficient condition for $(e)$ is that for each $g=1,\dots,G$ and setting $\widetilde{X}_i = X_i - \Ex{X_i\vert g_i^0}$ as the stacked $x_{it} - \Ex{x_{it}\vert g_i}$ for each $i$, the matrix
\[
\Ex{\widetilde{X}_i'\widetilde{X}_i\vert g_i^0 = g}
\]
has full rank. Since these matrices are symmetric, they are positive definite and so is their sum. Hence the minimum eigenvalue of the sum is positive. This is similar to fixed-effects identification where the within-transformation removes the time mean from each individual. Instead, this is at the group level which would require removing the time dependent group mean for which individual $i$ belongs.

Similar to fixed-effects estimation, the following consistency result does not require $T$ to be large, but large $T$ will be needed when groupings are unknown.

\begin{theorem}\label{prop:infeasiblecons}
Under Assumption 1, $\widetilde{\theta} \to_p \theta^0$ and $\widetilde{\alpha} \to_p \alpha^0$, as $N \to \infty$. 
\end{theorem}
The result follows from the consistency of M-estimators due to the definition of $\left(\widetilde{\theta},\widetilde{\alpha}\right)$ as minimizers of $\widetilde{Q}$, uniform weak converge of $\widetilde{Q}$ to $Q$ and the property that the unique minimizer of $Q$ is the true parameter values.

With $\sqrt{N}$-consistency of the infeasible WGFE estimator, I turn to the asymptotic distribution of the estimator. Under a few more standard assumptions, the infeasible WGFE estimator is asymptotically normally distributed as $N,T \to \infty$. This follows from the standard asymptotic normality theorem for M-estimators. First, define the estimator for the group variance when groupings are known as $\widetilde{\sigma}_{g_i^0}^2(\widetilde{\theta},\widetilde{\alpha})$. Since the infeasible estimator is consistent then $\widetilde{\sigma}_{g_i^0}^2(\widetilde{\theta},\widetilde{\alpha}) \to_p \sigma_{g_i^0}^2$ as $N$ approaches infinity using a weak law of large numbers. The following assumptions are sufficient to characterize the asymptotic distribution of $\left(\widetilde{\theta},\widetilde{\alpha}\right)$.\\

\begin{assump}\label{asym:infeas}\textnormal{(Normality of infeasible WGFE).}
\begin{itemize}
	\item[a.] For all $i,j = 1,\dots,N$ and $t = 1,\dots,T$: $\Ex{x_{jt} u_{it}} = 0$.

	\item[b.] There exists positive definite matrices $B_\theta$ and $V_\theta$ such that, as $N,T \to \infty$
\begin{gather*}
	 \frac{1}{NT} \sum_{i=1}^N \sum_{t=1}^T \widetilde{\sigma}_{g_i^0}^{-1}(\widetilde{\theta},\widetilde{\alpha})(x_{it} - \overline{x}_{g_i^0 t})(x_{it} - \overline{x}_{g_i^0 t})'
\longrightarrow_p B_\theta\\
\Ex{\frac{1}{NT}\sum_{i=1}^N \sum_{j=1}^N \sum_{t=1}^T \sum_{s = 1}^T
\left[\widetilde{\sigma}_{g_i^0}(\widetilde{\theta},\widetilde{\alpha}) \widetilde{\sigma}_{g_j^0}(\widetilde{\theta},\widetilde{\alpha})\right]^{-1}(x_{it} - \overline{x}_{g_i^0 t})(x_{js} - \overline{x}_{g_j^0 s})'u_{it}u_{js}} \longrightarrow V_\theta.
\end{gather*}
\item[c.] As $N,T \to\infty$, $
	\frac{1}{\sqrt{NT}} \sum_{i=1}^N \sum_{t=1}^T \sigma_{g_i^0}^{-1}(x_{it} - \overline{x}_{g_i^0 t})u_{it} \to_d N(0,V_\theta).
$
\item[d.] For all $g$ and $t$ and as $N \to \infty$,
$
\frac{1}{N}\sum_{i=1}^N\sum_{j=1}^N \Ex{\mathds{1}\{g_i^0 = g\}\mathds{1}\{g_j^0 = g\} u_{it}u_{jt}} \rightarrow v_{gt} > 0.
$
\item[e.] For all $g$ and $t$ and as $N \to \infty$,
$
\frac{1}{\sqrt{N}}\sum_{i=1}^N \mathds{1}\{g_i^0 = g\}u_{it} \longrightarrow_d N(0,v_{gt}).
$
\end{itemize}
\end{assump}
Assumption \ref{asym:infeas}($a$) allows covariates that are strictly exogenous and also time lagged outcomes. The following establishes the infeasible WGFE estimator is asymptotically normal. 
\begin{theorem}
Suppose Assumptions \ref{as:infeas} and \ref{asym:infeas} hold. As $N$ and $T$ approach infinity,
\[
\sqrt{NT} (\widetilde{\theta} - \theta^0) \longrightarrow_d N(0, B_\theta^{-1} V_\theta B_\theta^{-1})
\]
and, for all $g = 1,\dots,G$ and $t = 1,\dots, T$,
\[
\sqrt{N}(\widetilde{\alpha}_{gt} - \alpha_{gt}^0) \longrightarrow_d N(0,v_{gt}/\left[\mathbb{P}(g_i^0 = g)\right]^2).
\]
\end{theorem}
The form of the asymptotic variance of the infeasible estimator is different than the pooled least squares case in that observations are weighted according to the noise in their group. It appears similar to weighted least squares however the weight is the standard deviation instead of the variance.

\subsection{Consistency of WGFE}
This section establishes conditions for the consistency of the WGFE estimator $\widehat{\theta}$ of $\theta^0$. Remarkably these assumptions are identical to those made in \cite{bm:2015} for consistency of the GFE estimator.\\[1cm]

\begin{assump}\label{as:con}\textnormal{(Consistency of parameters).} There exists $M >0$ as in Assumption \ref{as:infeas} such that
	\begin{itemize}

		\item[$a.$] ${\textstyle \left\vert \frac{1}{NT} \sumnt \sum_{s=1}^T\Ex{u_{it}u_{is} x_{it}'x_{is}}\right\vert \leq M }$.
		
		\item[$b.$] ${\textstyle \frac{1}{N}\sum_{i=1}^N \sum_{j=1}^N \left\vert  \frac{1}{T}\sum_{t=1}^T \Ex{v_{it}v_{jt}} \right\vert \leq M}$.

		\item[$c.$] ${\textstyle \left\vert \frac{1}{N^2 T} \sum_{i=1}^N \sum_{j=1}^N\sum_{t=1}^T \sum_{s=1}^T \textnormal{Cov}(v_{it}v_{jt},v_{is}v_{js}) \right\vert \leq M}$.

		\item[$d.$] Let $\overline{x}_{g\wedge \gt t}$ denote the average of $x_{it}$ in the intersection of groups $g_i^0 =g$ and $g_i = \gt$. For all $g = 1,\dots,G$, and for any $\widetilde{\gamma} =(\widetilde{g}_1, \dots, \widetilde{g}_N) \in \Gamma_G^N$
define $\widehat{\rho}(\widetilde{\gamma})$ as the minimum eigenvalue of the matrix
\[
 \frac{1}{NT}\sumnt  \left( x_{it} -  \overline{x}_{g_i \wedge g^0 t}\right)\left( x_{it} -  \overline{x}_{g_i\wedge g^0 t}\right)'.
\]
Then, $\min_{\gamma\in\Gamma_G}\{\widehat{\rho}(\gamma)\} \to_p \rho > 0$ as $N,T \to \infty$.
	\end{itemize}
\end{assump}

To summarize from previous work, weak dependence conditions in the form of Assumptions \ref{as:con}($b,c,d$) are imposed. These assumptions are related to those made in \cite{stock:2002} and \cite{bai:2002} on large factor models. Assumption \ref{as:con} ($b$) allow for lagged outcomes or any predetermined regressor as covariates e.g. $\Ex{u_{it}\vert x_{it},x_{it-1},u_{it-1}} = 0$. Assumptions \ref{as:con} ($b$) and $(d)$ are conditions on time series dependence of the errors and covariates while Assumption \ref{as:con}($c$) limits cross-sectional dependence. Assumption \ref{as:con} ($e$) is akin to Assumption \ref{as:infeas} ($e$), so it is similar to a full rank condition found in other linear regression models. This sample version is stronger since it requires sufficient variation in the covariates across time and individuals within groups generated by \emph{any} grouping scheme $\widetilde{\gamma}\in \Gamma_G^N$. 
\begin{theorem}\label{prop:conswgfe}
	Suppose Assumption \ref{as:infeas} and \ref{as:con} hold. Then, as $N,T \to \infty$,  
\[\widehat{\theta} \to_p \theta^0 \>\>\>\>\> \text{and} \>\>\>\>\>
\frac{1}{NT} \sumnt \left( \widehat{\alpha}_{\widehat{g}_i t} - \alpha_{g_i^0 t}^0 \right)^2 \longrightarrow_p 0.
\] 
\end{theorem}
The proof largely follows that of \cite{bm:2015} with additional algebraic steps taken given the additional complexity of the WGFE criterion.

\subsection{Consistency of group assignments}

In this section I discuss the properties of the WGFE estimator for group membership in a large $N$ and $T$ setting. I refer to the assignment rules \eqref{rule:wgfe2} and $\eqref{gfe:assignment}$ as the WGFE and GFE assignments, respectively. I will first present the simple case of two groups for intuition and then will move to the general case where I provide conditions for any number of groups $G$ and with covariates.

\subsubsection{Simple Case}\label{gcons:simple} Consider a simplification of the main model where there are no covariates, $G= 2$ and errors are independent, but drawn from different normal distributions depending on $g_i^0$ with variance $\sigma_{g_i^0}^2$. Formally,
\[
	y_{it} = \alpha_{g_i^0}^0 + u_{it}, \>\>\> u_{it} \sim N(0,\sigma_{g_i^0}^2), \>\>\> u_{it}\ind u_{jt} \text{ for }i\neq j.
\]
Assume that $\sigma_1 \geq \sigma_2$. The probability of incorrectly assigning an individual into group 2 while they belong to group 1 using the WGFE estimator is 
\begin{align}
	\mathbb{P}(\widehat{g}_i(\alpha^0)) = 2 \vert g_i^0 =1) &= \mathbb{P} \left( \frac{1}{\sigma_2} \sum_{t=1}^T \left(\alpha_1^0 + u_{it} - \alpha_2^0\right)^2 + \sigma_2 
< \frac{1}{\sigma_1} \sum_{t=1}^T \left(\alpha_1^0 + u_{it} - \alpha_1^0\right)^2 + \sigma_1 \right)\label{wgfe:misclass} \\
&= \mathbb{P} \left(  \sum_{t=1}^T \left(u_{it}^2 + 2\left(\alpha_1^0 - \alpha_2^0\right)u_{it}\right) + T\left(\alpha_1^0 - \alpha_2^0\right) ^2
< \frac{\sigma_2}{\sigma_1} \sum_{t=1}^T u_{it}^2 + \left(\sigma_1 - \sigma_2\right)\sigma_2 \right)\notag\\
&= \mathbb{P} \left(\left(1- \frac{\sigma_2}{\sigma_1}\right) \sum_{t=1}^T u_{it}^2 + 2\left(\alpha_1^0 - \alpha_2^0\right)\sum_{t=1}^Tu_{it} < \left(\sigma_1 - \sigma_2\right)\sigma_2 - T\left(\alpha_1^0 - \alpha_2^0\right) ^2\right)\notag
\end{align}
This expression on the left hand side is the distribution function of some generalized chi-squared distribution and has no closed-form; see \cite{davies:1980}. Note that if $\sigma_1 = \sigma_2$ then this is the example for GFE estimation in \cite{bm:2015}. In the case that $\alpha_{1}^0 = \alpha_{2 }^0$ are equal:
\[
	\mathbb{P}\left(\widehat{g}_i(\alpha^0)\right) = 2 \vert g_i^0 =1) =
	\mathbb{P} \left(\frac{\sum_{t=1}^T u_{it}^2 - T}{\sqrt{2T}} < \frac{\sigma_1\sigma_2 - T}{\sqrt{2T}} \right) 
	\approx \Phi\left(\frac{\sigma_1\sigma_2 - T}{\sqrt{2T}}\right) \longrightarrow 0
\]
as $T\to\infty$ where $\Phi$ is the standard normal cdf. However if $\sigma_2 >\sigma_1$ then this is no longer the case since this assignment rule can't distinguish between the smaller density with $\sigma_1$ and the larger density with $\sigma_2$ and always assigns to the larger density. Essentially the leading term is now negative and the inequality will reverse leading to 
\begin{equation}\label{assign:simplecase}
\mathbb{P}\left(\widehat{g}_i(\alpha^0)\right) = 2 \vert g_i^0 =1)= \mathbb{P} \left(\frac{\sum_{t=1}^T u_{it}^2 - T}{\sqrt{2T}} > \frac{\sigma_1\sigma_2 - T}{\sqrt{2T}} \right)  \approx 1- \Phi\left(\frac{\sigma_1\sigma_2 - T}{\sqrt{2T}}\right) \longrightarrow 1
\end{equation}
so despite the additional information in the assignment rule this demonstrates the need for separation of group fixed effects. 

There is no closed form for the misclassification probability \eqref{wgfe:misclass} so comparisons of asymptotic performance of assignments is given numerically. In Figure \ref{fig:tgroupbig}, group 1 is assumed to have a variance greater than or equal to group 2 in all the examples. In the top left plot we see again that GFE and WGFE assignment are equivalent when variances are identical. In the top right plot, when there is separation in both mean and variance, the probability of misclassification goes to zero at a faster rate for WGFE assignment. In the bottom plot there is weak separation in means and in variances and the WGFE significantly outperforms GFE assignment.
\begin{figure}[h]
\centering
\includegraphics[scale=0.25]{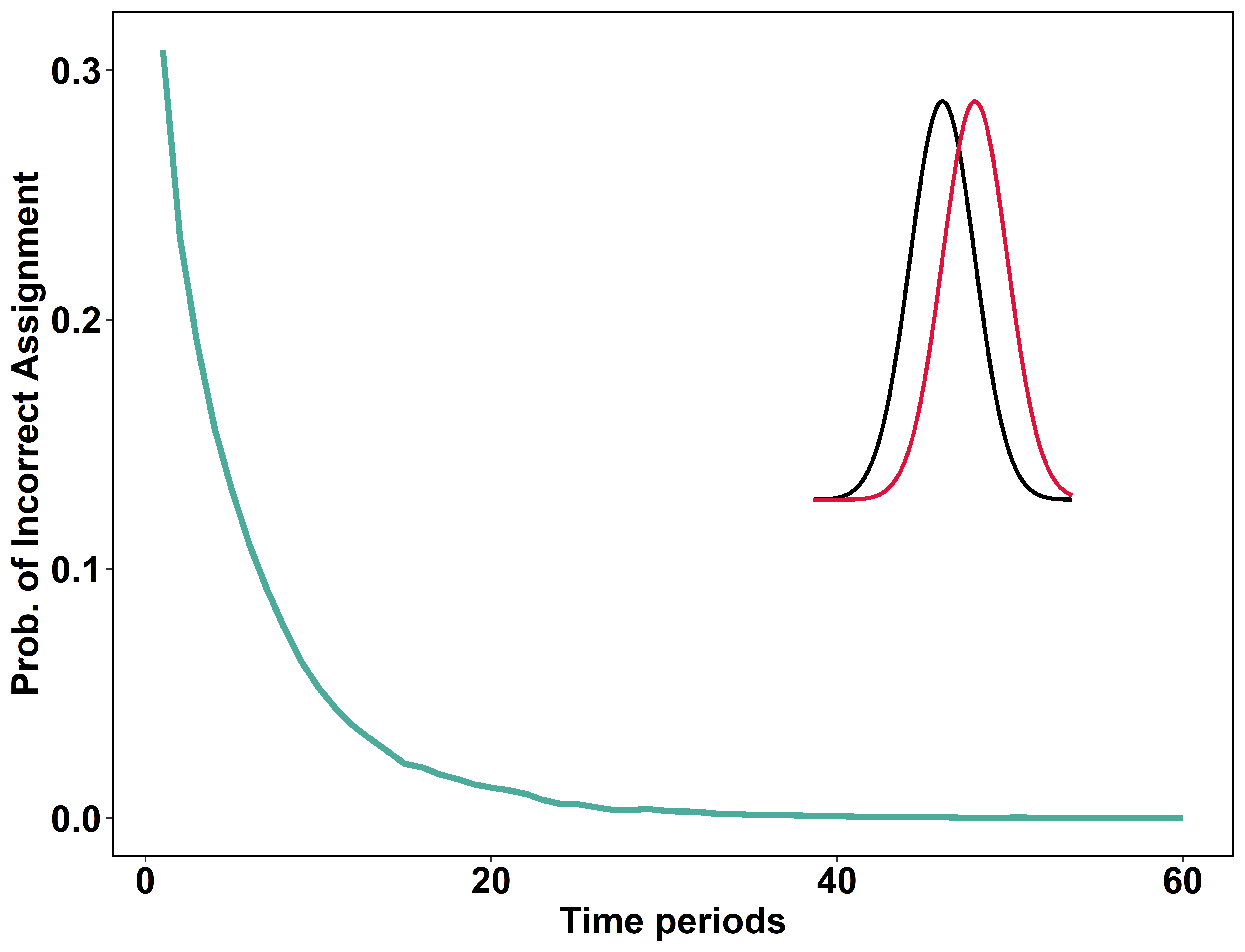}
\includegraphics[scale=0.25]{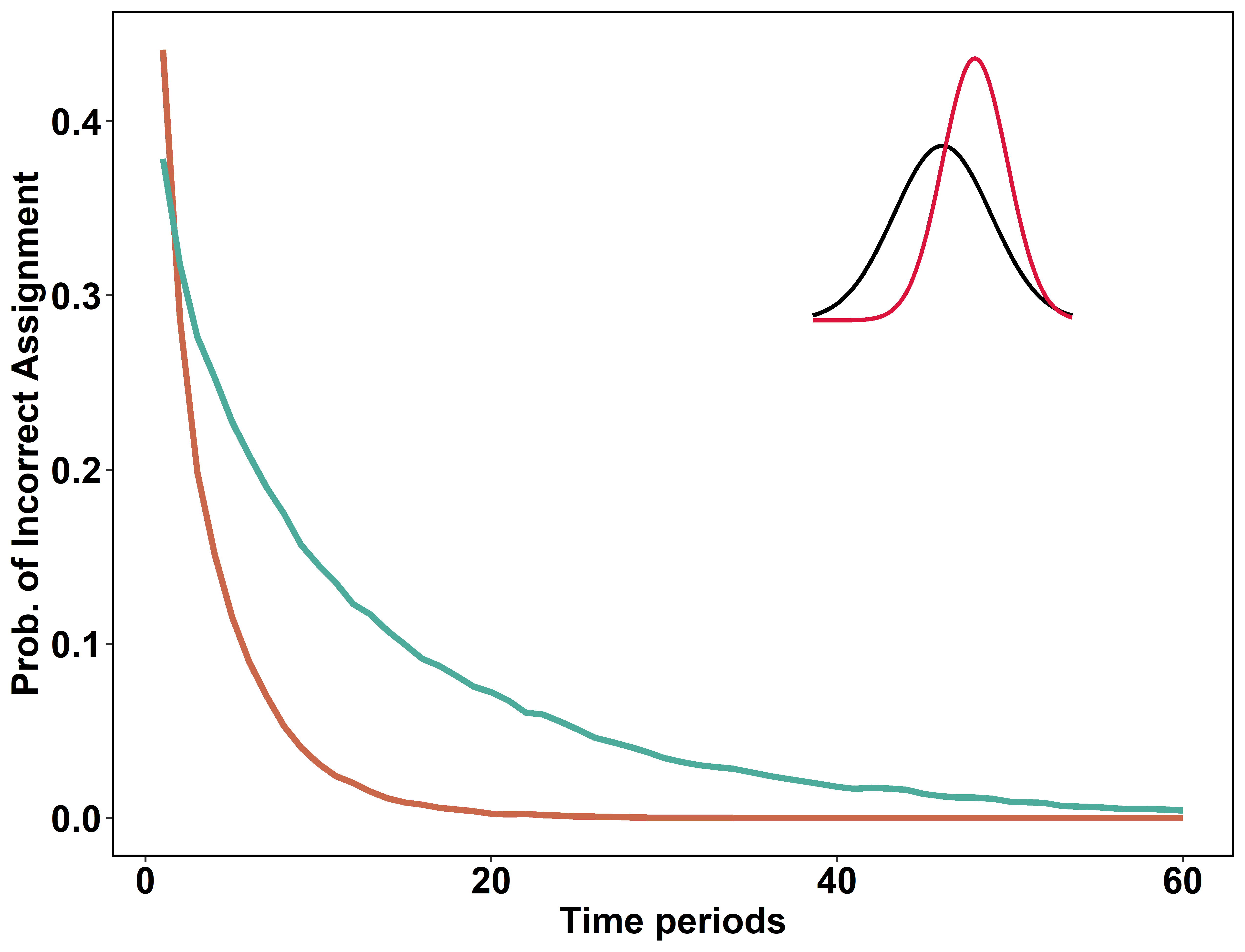}
\includegraphics[scale=0.25]{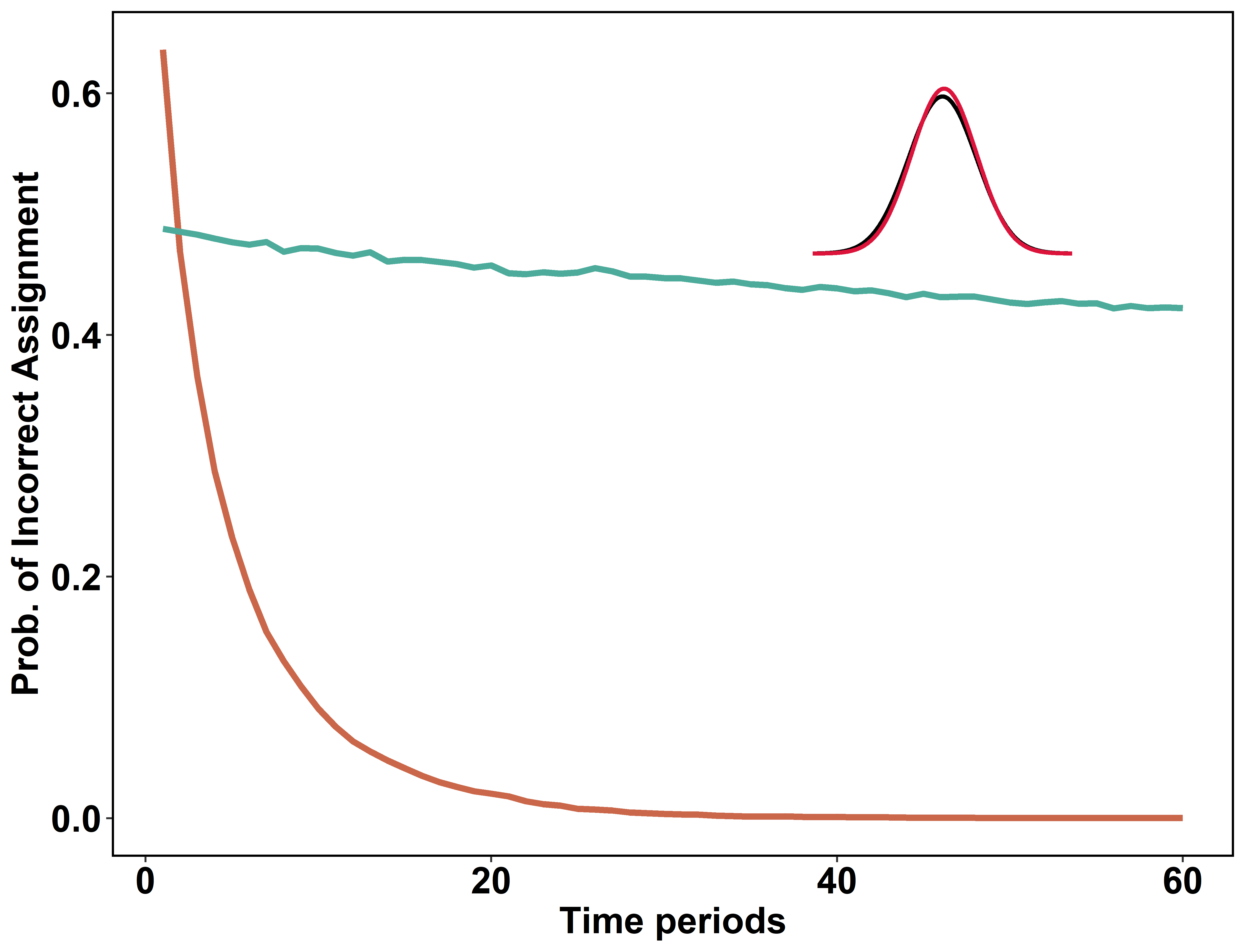}
\caption{WGFE (orange) and GFE (blue) probability of missassignment into group 2 while true group is group 1 as $T$ approaches infinity. Density in black is group 1 and density in red is group 2.}
\label{fig:tgroupbig}
\end{figure}
In the case that the true group variance is strictly less than the missclassified group variance, I run into a similar problem to \eqref{assign:simplecase}, however convergence is still possible. In Figure \ref{fig:tgroupsmall}, I fix $\sigma_1 = 1$, $\sigma_2 = 1.5$ and $\alpha_1^0 = 0$ and I take $\alpha_2^0 = 0.5,1,1.5,2$. When there is weak separation in the $\alpha$'s, misclassification is almost surely going to occur for large $T$ since the WGFE assignment can't distinguish between the two groups due to overlap and every point will be classified into the larger variance group. However, as there is more separation, classification improves and there will be convergence. Note that the GFE assignment also benefits from more separation, but the WGFE assignment benefits more in improving misclassification when the variance of the true group is smaller as it takes into account both first and second moment information. 

In this example of two groups, there will be a group with a smaller variance and so the misclassification probability is susceptible to the issues in Figure \ref{fig:tgroupsmall}. Therefore to achieve desirable simultaneous convergence properties of assignment to both groups, there needs to be additional restrictions on how separated the groups are. In light of the figures and the mathematical expressions, there is a first order importance on mean separation between groups. While variance separation lends to advantages for the WGFE assignment, if the means aren't far enough then this advantage ceases and becomes an obstacle in convergence for one side of group assignments. This is precisely the requirement of strong separability \eqref{strongsep} for identification of groups in WGFE estimation.
\begin{figure}[h]
\centering
\includegraphics[scale=0.25]{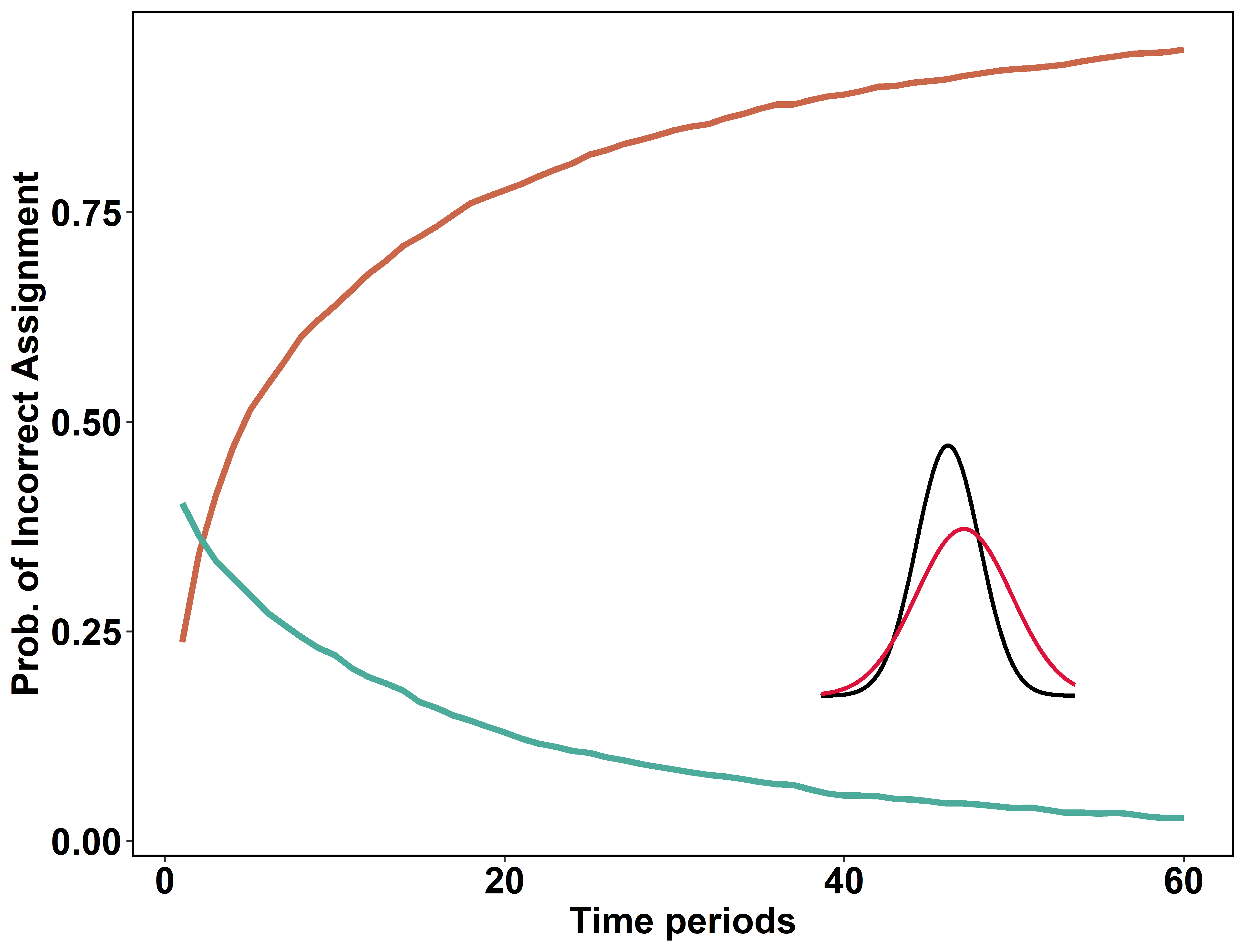}
\includegraphics[scale=0.25]{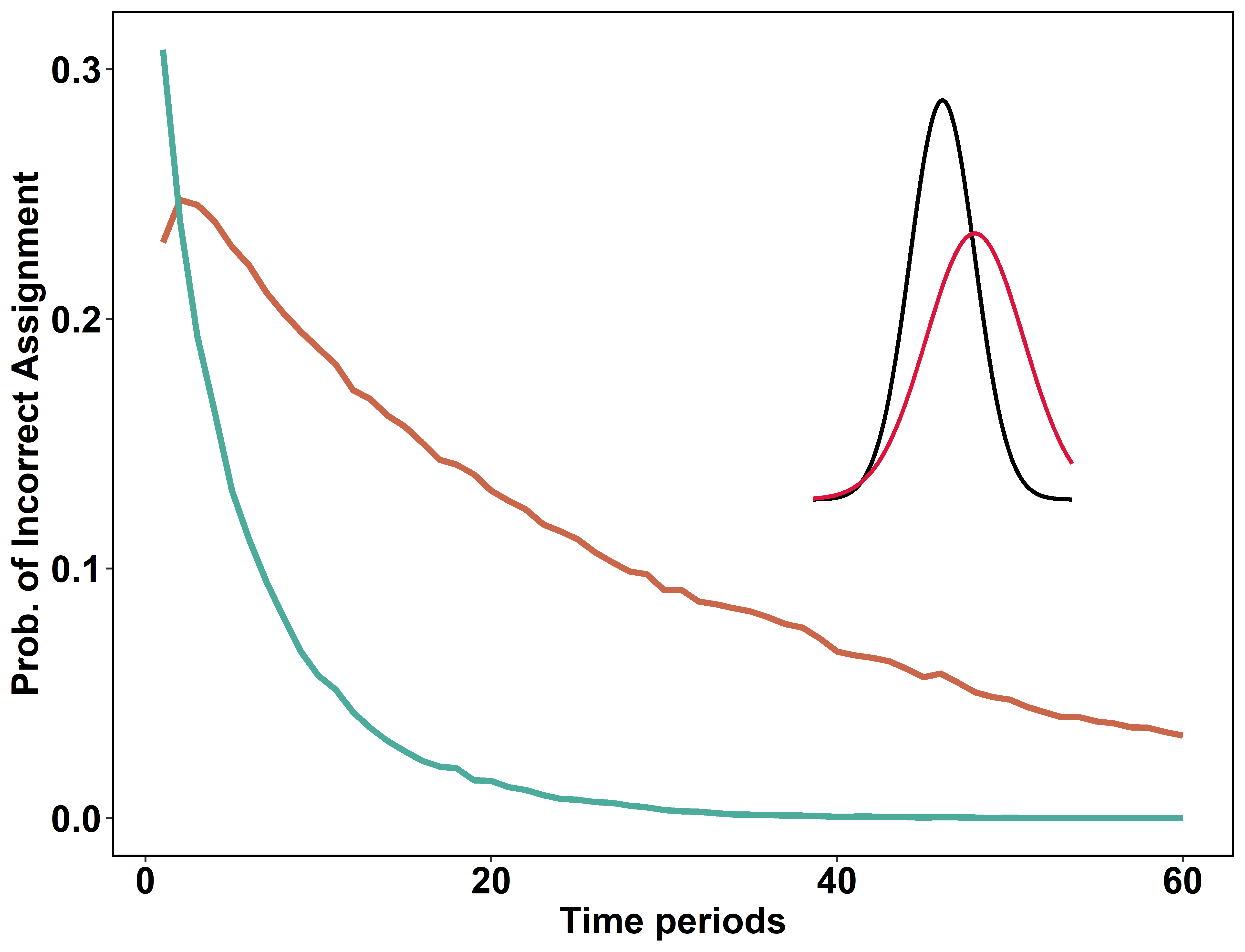}
\includegraphics[scale=0.25]{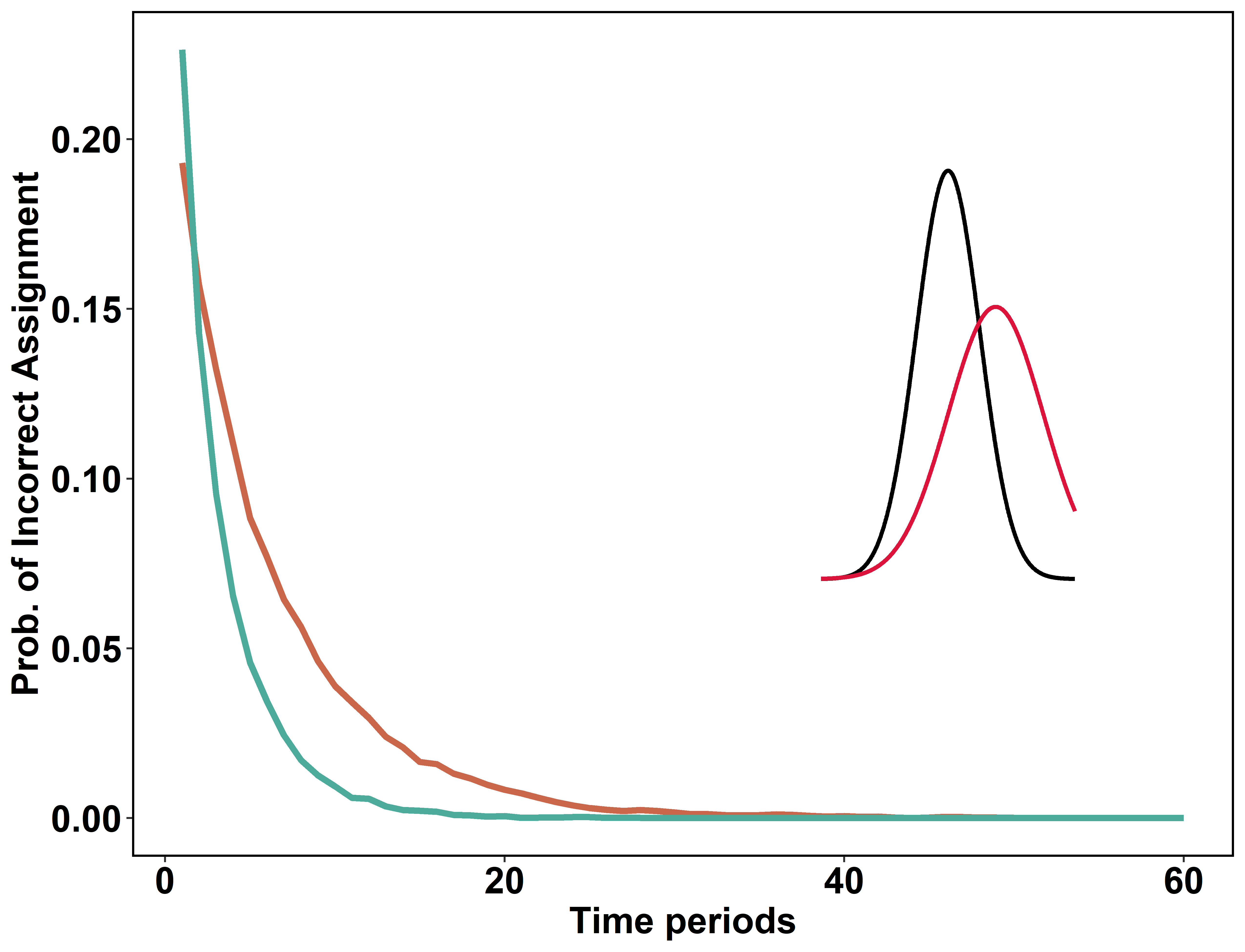}
\includegraphics[scale=0.25]{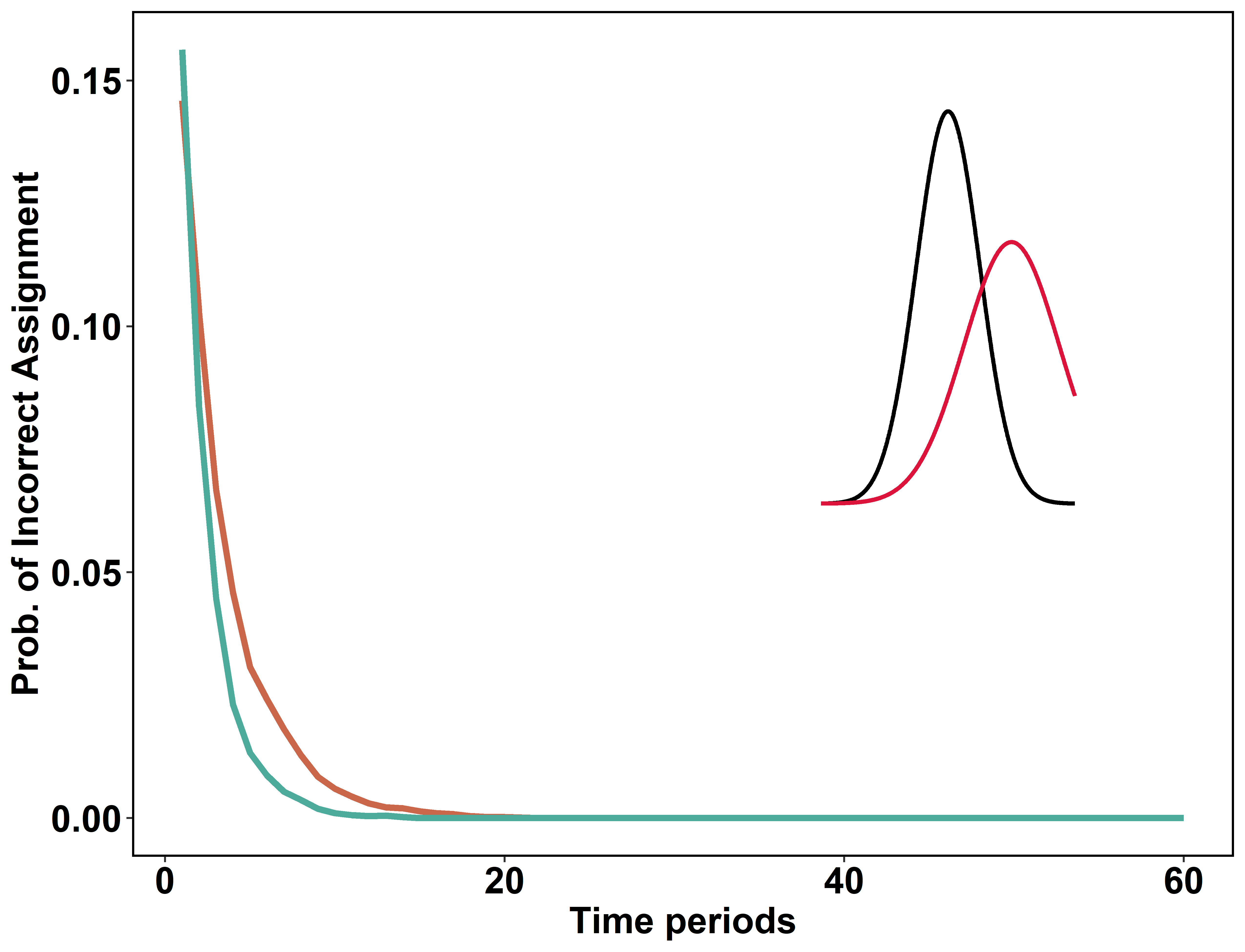}
\caption{WGFE (orange) and GFE (blue) probability of missassignment into group 2 while true group is group 1 as $T$ approaches infinity. Density in black is group 1 and density in red is group 2.}
\label{fig:tgroupsmall}
\end{figure}

\subsubsection{The General Case} 
Let $\eta > 0$ and define a neighborhood $\mathcal{N}_\eta$ around the true parameter values $\theta^0$ and $\alpha^0$ as the subset of $(\theta,\alpha)\subset \Theta \times \mathcal{A}^{GT}$ that satisfy $
\norm{\theta - \theta^0} < \eta$ and $T^{-1}\sum_{t=1}^T (\alpha_{gt}^0 - \alpha_{gt})^2 < \eta$ for all $g = 1,\dots,G$. Define
\begin{align*}
	\sigma_g^2(\theta,\alpha,\gamma) = \frac{1}{T}\sum_{t=1}^T\Ex{\mathds{1}\{g_i = g\} \left(u_{it} +x_{it}'(\theta^0 - \theta) + \alpha_{gt}^0 - \alpha_{gt} \right)^2}
\end{align*}
and denote the discrete probability measure induced by a random partition $\gamma$ as $\lambda(\gamma) \in \mathring{\Delta}^G$ where $\mathring{\Delta}^G$ is the interior of the $G$ dimensional probability simplex. The form of the objective function \eqref{wgfe:obj} prevents empty groups from being optimal hence the optimal grouping must lie in $\mathring{\Delta}^G$.

\phantom{a}
 
\begin{assump}\textnormal{(Consistency of group assignments).}\label{as:ae}
\begin{itemize}
	\item[$a$.] For all $g = 1,\dots,G$, ${\textstyle \frac{1}{N} \sum_{i=1}^N \mathds{1}\{g_i^0 = g\} \rightarrow_p P_g > 0}$; 

	\item[$b$.] There exists constants $a > 0$ and $d_1 > 0$ and a sequence $\alpha[t] \leq e^{-at^{d_1}}$ such that, for all $i = 1, \dots, N$ and $g = 1, \dots, G$, $\{u_{it}\}_t$ and $\{\alpha_{gt}^0\}_t$ are strongly mixing processes with mixing coefficients $\alpha[t]$. Moreover, $\Ex{\alpha_{gt}^0 u_{it}} = 0$ for all $g = 1, \dots, G$.

	\item[$c$.] There exists constants $b > 0$ and $d_2 > 0$ such that $\mathbb{P}(\vert u_{it} \vert > m) \leq e^{1- \left(\frac{m}{b}\right)^{d_2}}$ for all $i$, $t$, and $m>0$.

	\item[$d$.] There exists a constant $M^*>0$ such that, as $N,T$ tend to infinity I have
\[
	\sup_{i \in \{1, \dots, N\}} \mathbb{P}\left(\frac{1}{T}\sum_{t=1}^T \norm{x_{it}} \geq M^* \right) = o\left( T^{-\delta}\right) \>\>\> \text{ for all } \delta >0.
\]
	\item [$e$.] For all $g,\widetilde{g} = 1,\dots,G$ where $g\neq \gt$, 
\begin{align*}
\textnormal{plim}_{T\to\infty}\frac{1}{T} \sum_{t=1}^T \left(\alpha_{g t}^0 - \alpha_{\gt t}^0 \right)^2 &> c_{g,\gt} + \textnormal{plim}_{T\to\infty}\sup_{\lambda(\gamma) \in \mathring{\Delta}^G}\left\vert\frac{\sigma_g(\theta^0,\alpha^0,\gamma)- \sigma_{\gt}(\theta^0,\alpha^0,\gamma)}{\sigma_{\gt}(\theta^0,\alpha^0,\gamma)}\right\vert M\\
&= c_{g,\gt} + C_{g,\gt} M  >0
\end{align*}
where $M$ is the constant from Assumption \ref{as:infeas}.

\end{itemize}
\end{assump}

The assumptions made here do not deviate much from those in \cite{bm:2015} besides a strong separability assumption in part ($e$) that requires a specific bound depending on other features of the model. Note that is a stronger version of \eqref{strongsep} reflecting the fact that the group standard errors must also be estimated and it plays a role in showing the sample probability of misassigning individuals converges to zero. The denominator of $C_{g,\widetilde{g}}$ is bounded away from zero since the groupings are selected from the interior of the probability simplex.
 
In Assumptions \ref{as:ae}($b$) and ($c$) I strengthen restrictions on the dependence and tail properties of the error, respectively. I assume the error is strongly mixing with faster-than-polynomial decay rate and with tails that also decay faster than any polynomial. The mixing property is stronger than ergodicity and says that, eventually, the process will forget its history. Mathematically, for any stochastic process $X_t$ on some probability space $(\Sigma,\mathcal{F},\mathbb{P})$ define the function $\alpha[s]$ as a strongly mixing coefficient:
\[
	\alpha[s] = \sup\{|\mathbb{P}(A\cap B) - \mathbb{P}(A)\mathbb{P}(B)|: t\geq0, A \in X_0^t, B \in X_{t+s}^\infty\}
\]
where $X_a^b$ denote the sub $\sigma$-algebra of $\mathcal{F}$ specified between times $a$ and $b$. If $\alpha[s] \to 0$ as $s\to\infty$, then $X_t$ is said to be a strongly mixing process. I also assume that the group time-effects are strongly mixing and contemporaneously uncorrelated with the error $u_{it}$. These assumptions enable us to use exponential inequalities for weakly dependent processes and allows us to bound misclassification probabilities.\footnote{For more on asymptotic theory of weakly dependent process see \cite{rio:2017} and for useful exponential inequalities see \cite{merlevede:2011}. This condition for strongly mixing was established in \cite{rosenblatt:1956} to prove a central limit theorem.} Assumption \ref{as:ae} ($d$) is a condition on the distribution of the covariates. A sufficient condition would be if the covariates have bounded support or, alternatively, satisfy similar dependence and tail conditions as the error term. In their supplementary appendix, \cite{bm:2015} discuss in detail the latter condition since the strong mixing property may not hold with lagged outcomes.

The following proposition establishes an asymptotic equivalence of WGFE estimation to infeasible WGFE estimation and consistency of group membership as $N,T$ approach infinity, but $T$ may grow slower than $N$.

\begin{theorem}\label{prop:congroups}
	Suppose Assumptions \ref{as:infeas}, \ref{as:con} and \ref{as:ae} hold. Then, for all $\delta > 0$,
\begin{equation}
	\mathbb{P} \left( \sup_{i \in \{1,\dots , N} \vert \widehat{g}_{i} - g_i^0\vert >0 \right ) = o(1) + o(NT^{-\delta})
\end{equation}
and for all $g$ and $t$
\begin{gather}
\widehat{\theta} = \widetilde{\theta} + o_p(T^{-\delta})\\
\widehat{\alpha}_{gt}= \widetilde{\alpha}_{gt}+ o_p(T^{-\delta}).
\end{gather}
\end{theorem}

\subsection{Asymptotic Normality of the WGFE Estimator}

With conditions gauranteeing the asymptotic equivalence of WGFE estimator and the infeasible WGFE estimator, the asymptotic distribution of the WGFE estimator for $\theta^0$ and $\alpha^0$ is that of the infeasible version where groups are known.

\begin{proposition}\label{asym:wgfe}
Let Assumptions \ref{as:infeas}, \ref{asym:infeas}, \ref{as:con} and \ref{as:ae} hold. As $N$ and $T$ approach infinity such that for some $\nu>0$, $N/T^\nu \to 0$ I have
\[
	\sqrt{NT}\left( \widehat{\theta} - \theta^0 \right) \longrightarrow_d N\left(0,B_\theta^{-1} V_\theta, B_\theta^{-1}\right)
\]
and, for all $g = 1,\dots,G$ and $t = 1,\dots,T$,
\[
	\sqrt{N}\left( \widehat{\alpha}_{gt} - \alpha_{gt}^0 \right) \longrightarrow_d N\left(0,v_{gt}/\left[P_g\right]^2\right).
\]
\end{proposition}
\begin{proof}
By the asymptotic equivalence of $\widehat{\theta}$ and the infeasible $\widetilde{\theta}$ given by Proposition \ref{prop:congroups}, and the asymptotic normality of $\widetilde{\theta}$ I have the result:
\[
	\sqrt{NT}\left(\widehat{\theta} - \theta^0 \right) = \sqrt{NT}\left(\widetilde{\theta} - \theta^0\right) + \sqrt{NT} \left(\widehat{\theta} - \widetilde{\theta} \right) = \sqrt{NT}\left(\widetilde{\theta} - \theta^0\right) + o_p\left(\frac{\sqrt{N}}{T^\nu}\right).
\]
where I have chosen $\delta$ such that $1/2 - \delta = \nu$. Similarly for $\widehat{\alpha}_{gt}$ and $\widetilde{\alpha}_{gt}$, for any $g$ and $t$. 

\end{proof}

\section{Simulations and Empirical Applications}\label{sec:sims}

\subsection{Simulation Design}

Simulation follows \cite{bm:2015} where they revisit the Acemoglu, Johnson, Robinson, \& Yared (2008) study on the association between income and democracy. The data generating distribution is 
\[
	y_{it} = \widehat{\theta}_1 y_{i t-1} + \widehat{\theta}_2 x_{it} + 
	\widehat{\alpha}_{\widehat{g}_i t} + u_{it}
\]
where $y_{it}$ is an index of democracy and $x_{it}$ is log-income per capita. I split this DGP into two depending on the estimates, groupings and variance of normally distributed errors with zero mean:
\begin{itemize}
	\item[ 1.] $\widehat{\theta}^{GFE}$, $\widehat{\alpha}_{\widehat{g}^{GFE}}^{GFE}$ and $\widehat{g}^{GFE}$ and homoskedastic $\widehat{\sigma}^2$
	\item[2.] WGFE estimates, groupings and group heteroskedastic $\widehat{\sigma}_{\widehat{g}_i^{WGFE}}^2$ 
\end{itemize}
I generate data for different specifications of total number of groups $G$ while holding the income vector $x_{it}$ fixed. I collect a sample of 1,000 GFE and WGFE estimates for each specification. In Figure \ref{fig:dgp2stats}, the values of the heterogeneous variances and sizes are provided for DGP 2. In Figure \ref{fig:sim1}, I compare the root mean-squared error of the parameter estimates using WGFE, GFE, and two way fixed effects. In Figure \ref{fig:sim2}, I compare the average rates of misclassification across different $G$ for each DGP. We see that for DGP 1 there is not much difference between WGFE and GFE estimation as expected since there is no group heteroskedasticity. However, for DGP 2 it appears that WGFE is a better performer when variances are heterogeneous across groups. Unlike groupwise heteroskedasticity in simple regression with observed groups, ignoring it in this context not only affects standard errors, but also impacts parameter estimates since the estimators are functions of group assignments. If groups are different in size and variance, then assignment based solely on mean discrepancies is vulnerable to errors. It can also be interpreted as omitted variable bias if some incorrect assignments are determined since the latent variable is not properly controlled for.

\begin{figure}
	\centering
 \begin{tabular}{cccccc}
        \toprule $G$ & $\sigma_1$ & $\sigma_2$ & $\sigma_3$ & $\sigma_4$ & $\sigma_5$\\
        \midrule  2 & 0.219  & 0.086 & \--- & \--- & \--- \\
			      3 & 0.230 & 0.076  & 0.114 & \--- & \--- \\
			      4 & 0.235  & 0.105  & 0.152 & 0.059& \--- \\
				 5 & 0.141  & 0.206  & 0.064 & 0.096 & 0.209\\
        \bottomrule
\end{tabular}
\hspace{1cm}
 \begin{tabular}{cccccc}
        \toprule $G$ & $P_1$ & $P_2$ & $P_3$ & $P_4$ & $P_5$\\
        \midrule  2 & 0.64  & 0.36 & \--- & \--- & \--- \\
			      3 & 0.43 & 0.33  & 0.24 & \--- & \--- \\
			      4 & 0.33  & 0.23  & 0.14 & 0.29 &  \--- \\
				 5 & 0.14  & 0.21  & 0.30 & 0.20 & 0.14\\
        \bottomrule
\end{tabular}
\caption{DGP\# 2: latent variable features across groups. Left: group variances, Right: group probabilities}
\label{fig:dgp2stats}
\end{figure}

\begin{center}
 \begin{figure}[h!]\centering  \begin{tabular}{@{}cccccccccc@{}}
 \toprule DGP\# 1 &\phantom{abc} && \multicolumn{3}{c}{Lagged Term} &\phantom{abs} & \multicolumn{3}{c}{Income Term}\\ 
 \cmidrule{4-6} \cmidrule{8-10}
 \empty&$G$ && WGFE & GFE & FE && WGFE & GFE & FE  \\ \midrule 
  \empty&2 && 0.0484 & 0.0479 & 0.3146 && 0.0173 & 0.0173 &0.1095  \\ 
  \empty&3 && 0.0415 & 0.0416 & 0.1211&& 0.0133 & 0.0132 &0.1029 \\ 
  \empty&4 && 0.0423 & 0.0422 & 0.0623 && 0.0109 & 0.0109 &0.0708\\ 
  \empty&5 && 0.0458 & 0.0449 & 0.0527 && 0.0105 & 0.0105 &0.0855 \\ 
 \toprule DGP\# 2 \\

 \empty&2 && 0.042 & 0.069 & 0.2538 && 0.012 & 0.026 & 0.1134\\
 \empty&3 && 0.049 & 0.067 & 0.1211 && 0.010 & 0.022 & 0.1029\\
 \empty&4 && 0.053 & 0.066 & 0.1607 && 0.010 & 0.020 & 0.0983 \\
 \empty&5 && 0.047 & 0.061 & 0.2072 && 0.008 & 0.013 & 0.0972\\
 \bottomrule\\
 \end{tabular} 
 \caption{Root Mean Squared Errors for the WGFE, GFE and fixed-effects estimators of $\theta_1$ and $\theta_2$}
\label{fig:sim1}
  \end{figure}  
\end{center}

\begin{center}
 \begin{figure}[h!]\centering  \begin{tabular}{@{}cccccccc@{}}
 \toprule DGP\# 1 &\phantom{abc} && \multicolumn{2}{c}{Missclassified} &\phantom{abs} & \multicolumn{2}{c}{St. Error}\\ 
  \cmidrule{4-5} \cmidrule{7-8}
 \empty& $G$ && WGFE & GFE &&WGFE & GFE  \\ \midrule 
   \empty&2 && 10.38\% & 10.28\% && 3.95\% & 3.90\%  \\
    \empty&3 && 6.55\% & 6.56\% && 3.09\% & 3.21\%  \\
    \empty&4 && 4.74\% & 4.69\% && 2.99\% & 2.99\% \\
    \empty&5 && 4.13\% & 4.19\% && 3.17\% & 3.28\%  \\ \toprule
    DGP\# 2 \\
 
    \empty&2 && 3.71\% & 18.51\% && 2.70\% & 6.38\%  \\
    \empty&3 && 4.85\% & 16.93\% && 2.93\% & 4.84\%  \\
    \empty&4 && 8.86\% & 20.41\% && 3.83\% & 3.89\%  \\
    \empty&5 && 7.78\% & 13\% && 4\% & 4.30\%  \\
 \bottomrule\\
 \end{tabular} 
 
 \caption{Average rate of misclassification across 1,000 simulations}
\label{fig:sim2}
  \end{figure}  
\end{center}

\subsection{The Effects of Unionization on Wages}

In studies on the effects of unionization on earnings it is often argued that controlling for unobserved ability/skill is essential since employers who face a union contract may select more able candidates than those employers who do not (\cite{abowd:1982}). A solution for this selection problem would be a fixed-effects approach, but measurement error resulting from misclassification of unionized jobs is more pronounced in panel analysis, which has been shown to harm estimates. \cite{freeman:1984} show that estimates from panel studies are generally smaller than those from cross sectional studies across different changing union status groups. For the Panel Study of Income Dynamics (PSID), estimates range from 8\%--23\% and they might be affected by union misclassification. \cite{card:1996} estimate a discrete proxy for ability that is used to estimate five separate models for each ability level, controlling for misclassification in the process. He finds at ``high'' levels of ability there is negative bias when ignoring high skill heterogeneity versus the positive bias found with low ability workers. In this section, I will revisit the effect of unionization using data from the PSID on $N = 1,158$ workers who are the heads of their household between 2001 and 2019 ($T= 10$, odd numbered years). I draw comparisons to both of these studies by showing the WGFE estimate lies between the pooled and fixed effects estimates as postulated by \cite{freeman:1984} and the estimated groups follow a similar correlation pattern to skill groups found in \cite{card:1996}.

It is natural to assume that employers form impressions of ability of workers based on comparisons to other candidates for the position and so it is reasonable that this impression is part of a discrete set. For example, the employer may classify candidates as ``avoid hire'', ``maybe hire'' or ``strong hire''. For jobs that are not covered by unions, but compensate well, employers may be more selective and there may be a complicated correlation structure between working at a union job and skill. It is also reasonable that this discrete classification is correlated with wages across time, which may be due to the bargaining power of unions and overall economic conditions. For example, jobs with a strong union may be more robust to poor economic conditions such as a recession or worsening wage inequality. In addition the response from each skill group to earnings shocks might be more variable suggesting unobserved group heteroskedasticity.

I control for this discrete unobserved variable by estimating a linear model of unionization and wages with a group fixed-effects term: 
\begin{equation}\label{eq:unionmodel}
	logwage_{it} = \delta union_{it} + x_{it}'\theta + \alpha_{g_i t} + u_{it}
\end{equation}
where $logwage_{it}$ is the log real labor earnings of worker $i$ in year $t$ in 2001 USDs, $union_{it}$ is a dummy variable indicating if worker $i$ had a job at time $t$ that was under a union contract and $x_{it}$ contains time-varying and time-invariant characteristics such as years of schooling, if they are not white, if they are female, number of years of experience at current job and if there job is classified as ``blue collar''. I also include a time-varying dummy variable indicating if the worker resides in the south, along with their marital status and age.

Figure \ref{fig:sumstatsunion} shows summary statistics for the data. About 16\% of the sample are female and 33\% are nonwhite. Most individuals have 12 years of schooling with 26\% of the sample only completing grade school, 47\% completing some college and 20\% with postgraduate work. The sample contains workers who first appear with ages between 19 and 64 years and are followed for 20 years. In this panel, years of schooling is observed to be time-varying as some individuals returned to school or started working during their schooling. Figure \ref{fig:incomes_union} displays median incomes across different groups. Clearly, unions have a positive impact on real earnings despite the erosion of numbers of unionized jobs post 2007 as shown in Figure \ref{fig:union_prop}. Across sex, race, blue collar occupation, this effect persists. 
\begin{center}
 \begin{figure}[h!] \centering  \begin{tabular}{@{}cccccc@{}}
 \toprule  & \empty & Mean & Median & Min & Max\\[1.5mm] \midrule
 Wages & \empty & \$53,272 & \$40,496 & \$694 & \$3,365,385\\[1.5mm]
 Age & \empty & 46.41 & 46 & 19 & 87\\[1.5mm]
 Experience & \empty & 9.88 & 8 & 0 & 57\\[1.5mm]
 \bottomrule
 \end{tabular} 
 \caption{Wages in 2001 USD, experience in years.}
 \label{fig:sumstatsunion}
  \end{figure}  
\end{center}
\begin{figure}[h!]
	\includegraphics[scale = 0.35]{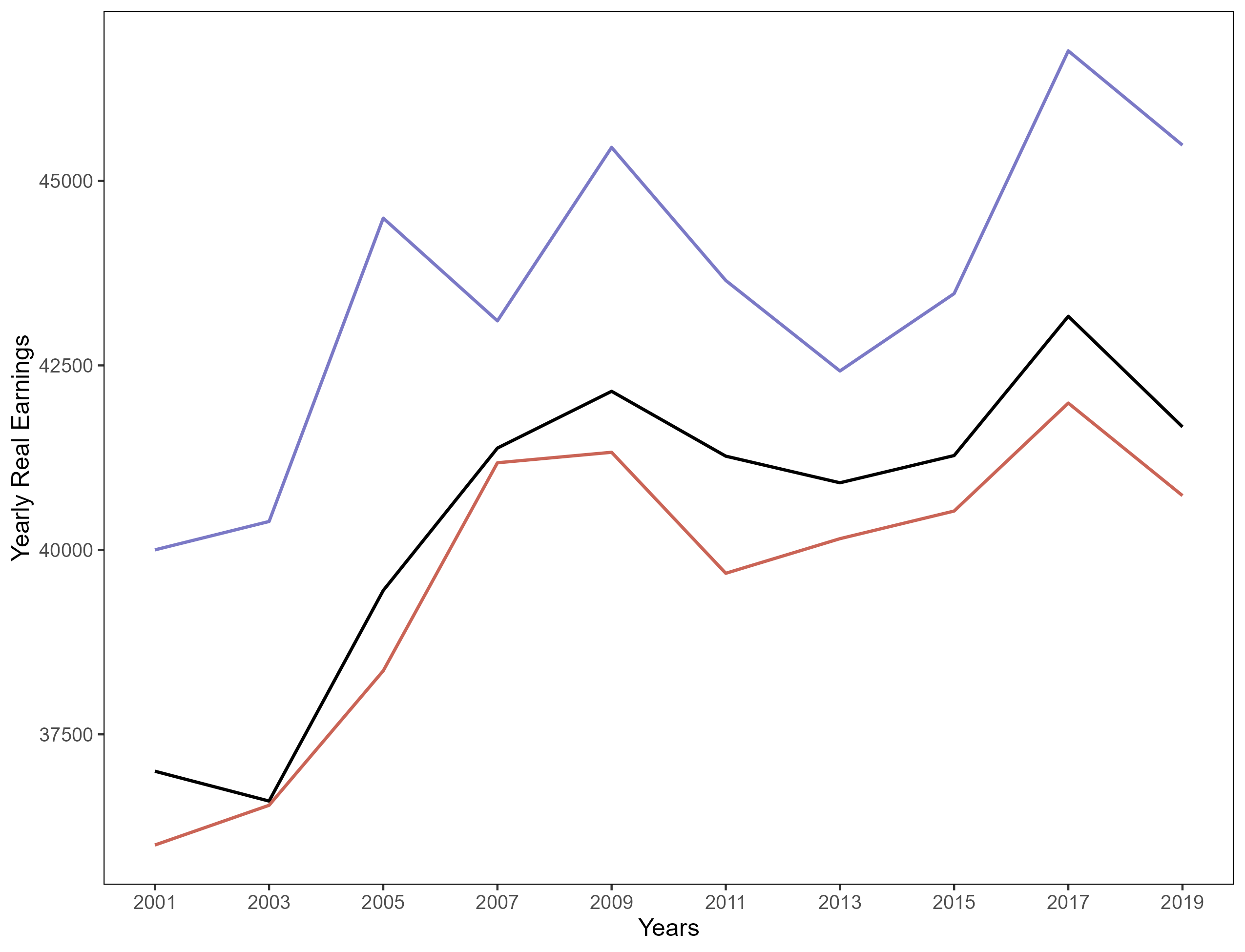}
	\includegraphics[scale = 0.35]{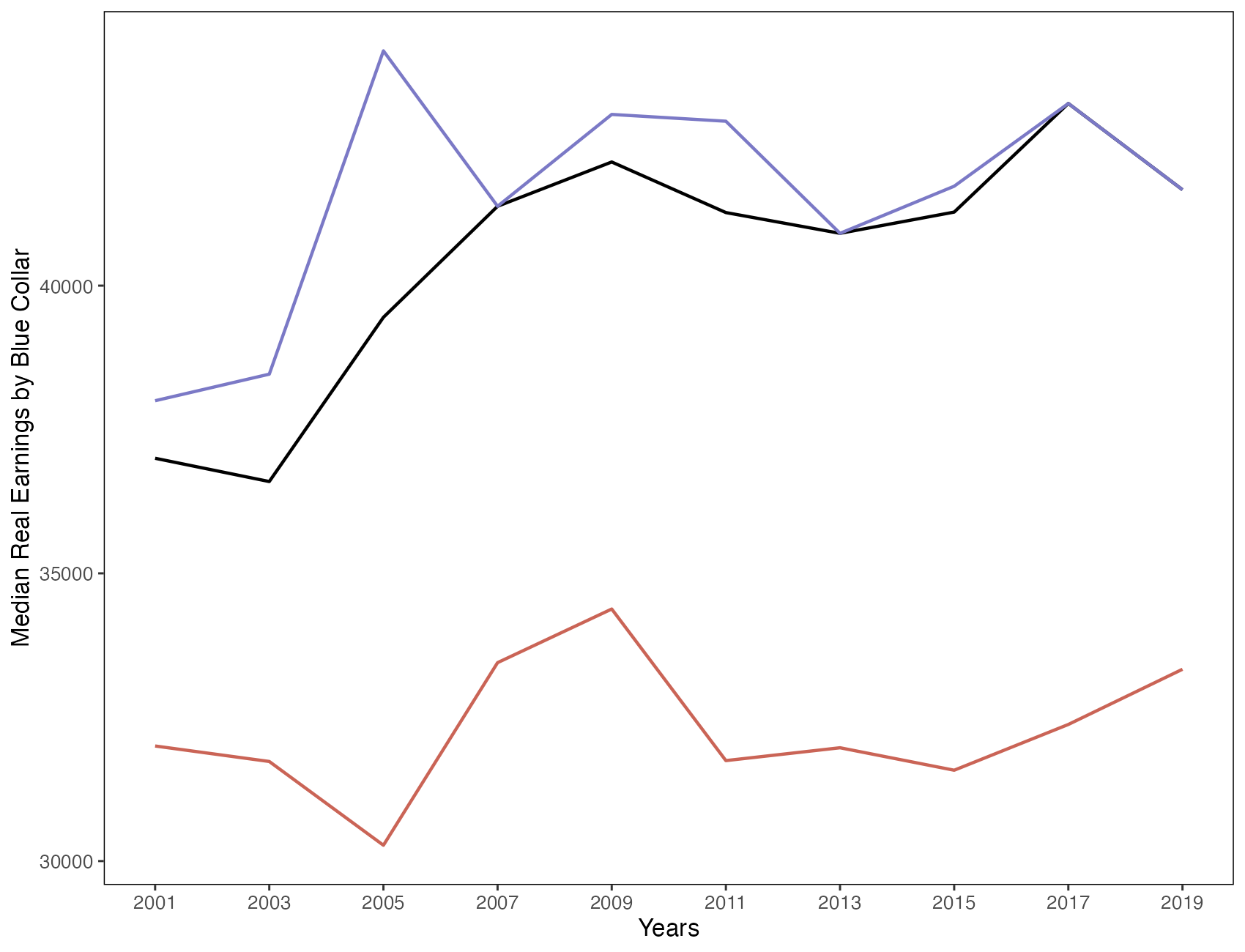}
	\includegraphics[scale = 0.35]{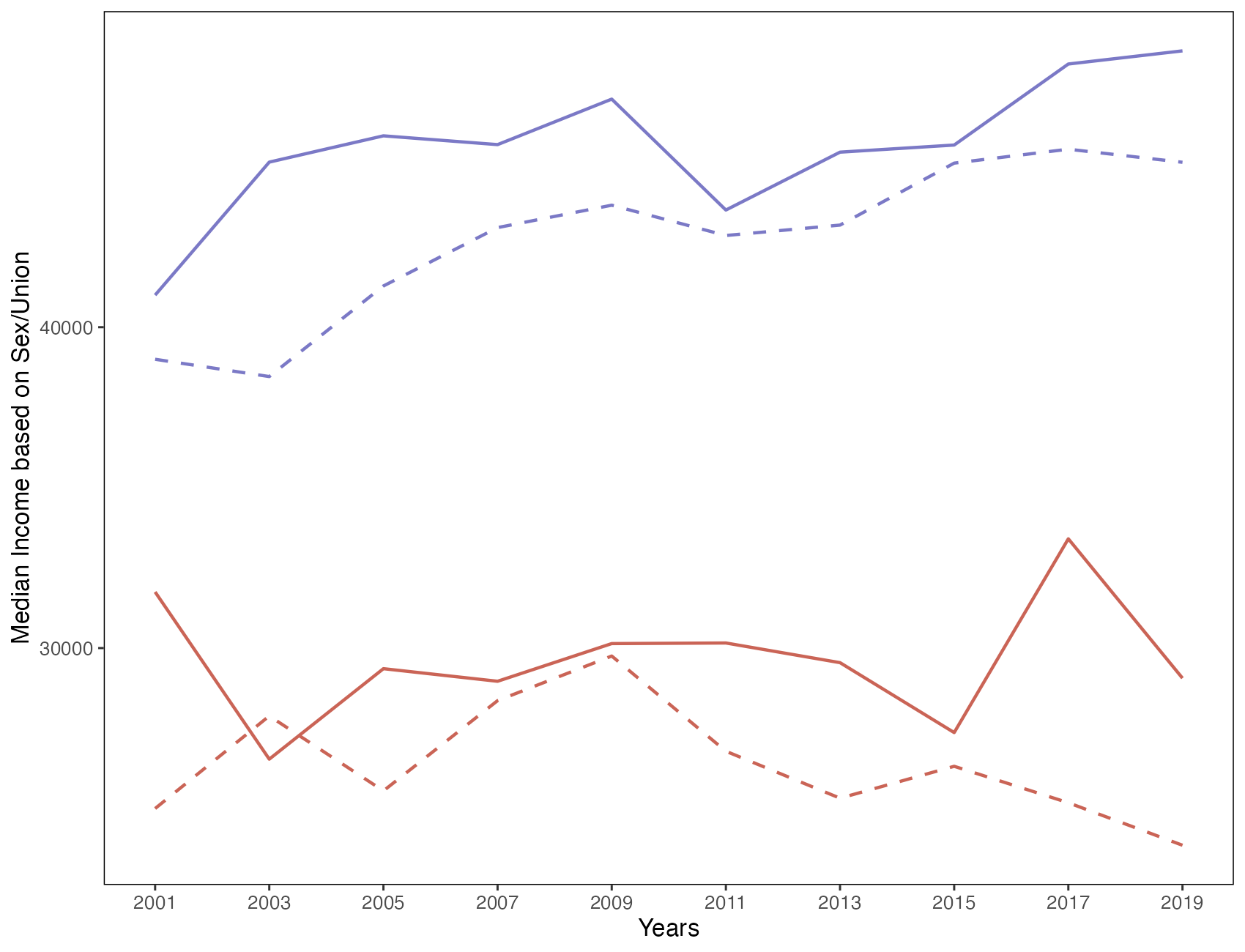}
	\includegraphics[scale = 0.35]{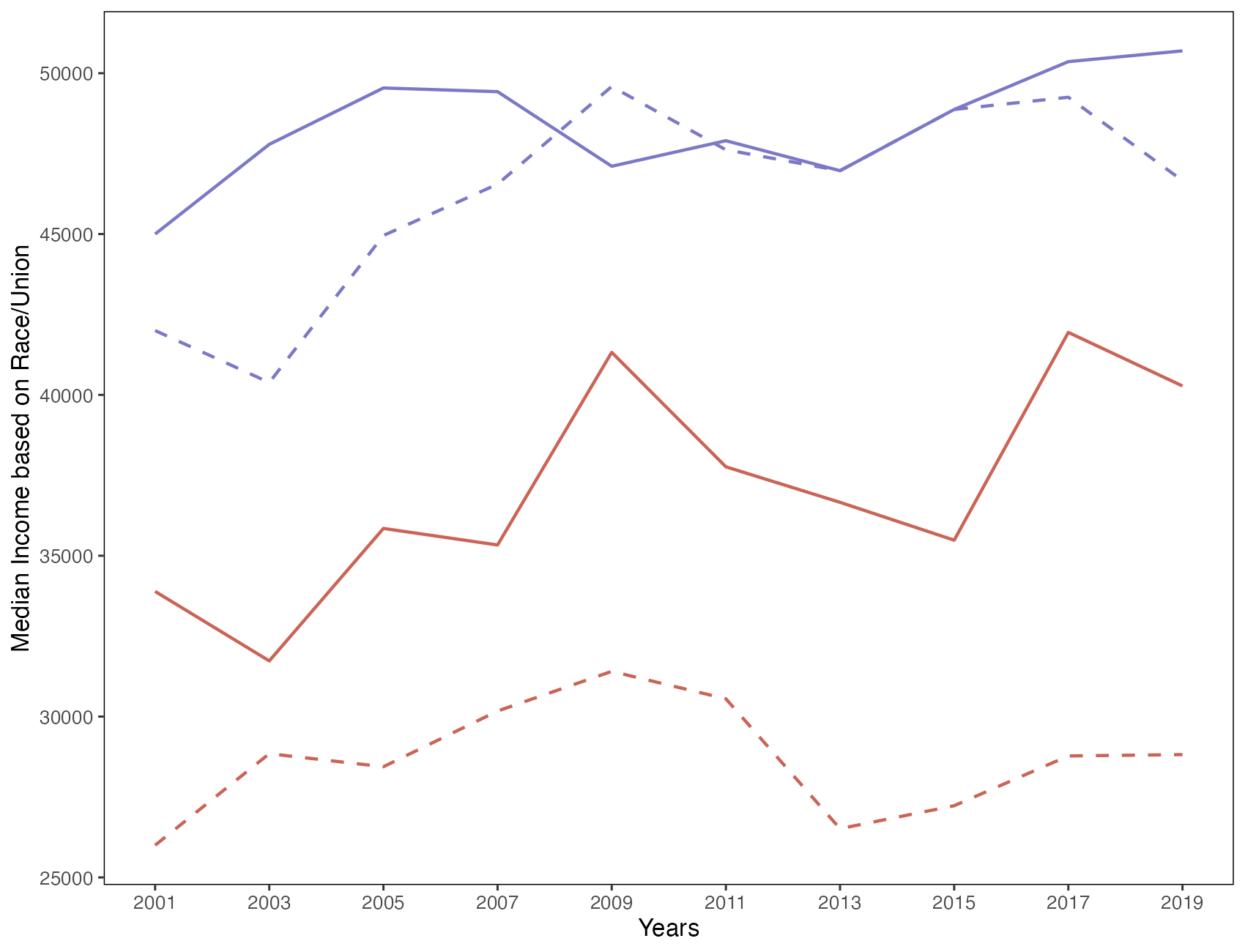}
\caption{Differences in median incomes between the groups. {\it Top Left}: Black: Aggregate, Blue: union group, Red: nonunion group. {\it Top Right}:
Black: Aggregate, Blue: union group, Red: nonunion group. 
{\it Bottom Left}: Blue: Male, Red: Female, Dashed: Nonunion group, Solid: Union group, {\it Bottom Right}: Blue: White, Red: Nonwhite, Dashed: Nonunion group, Solid: Union group}
\label{fig:incomes_union}
\end{figure}

\begin{figure}[h!]
\centering
	\includegraphics[scale = 0.45]{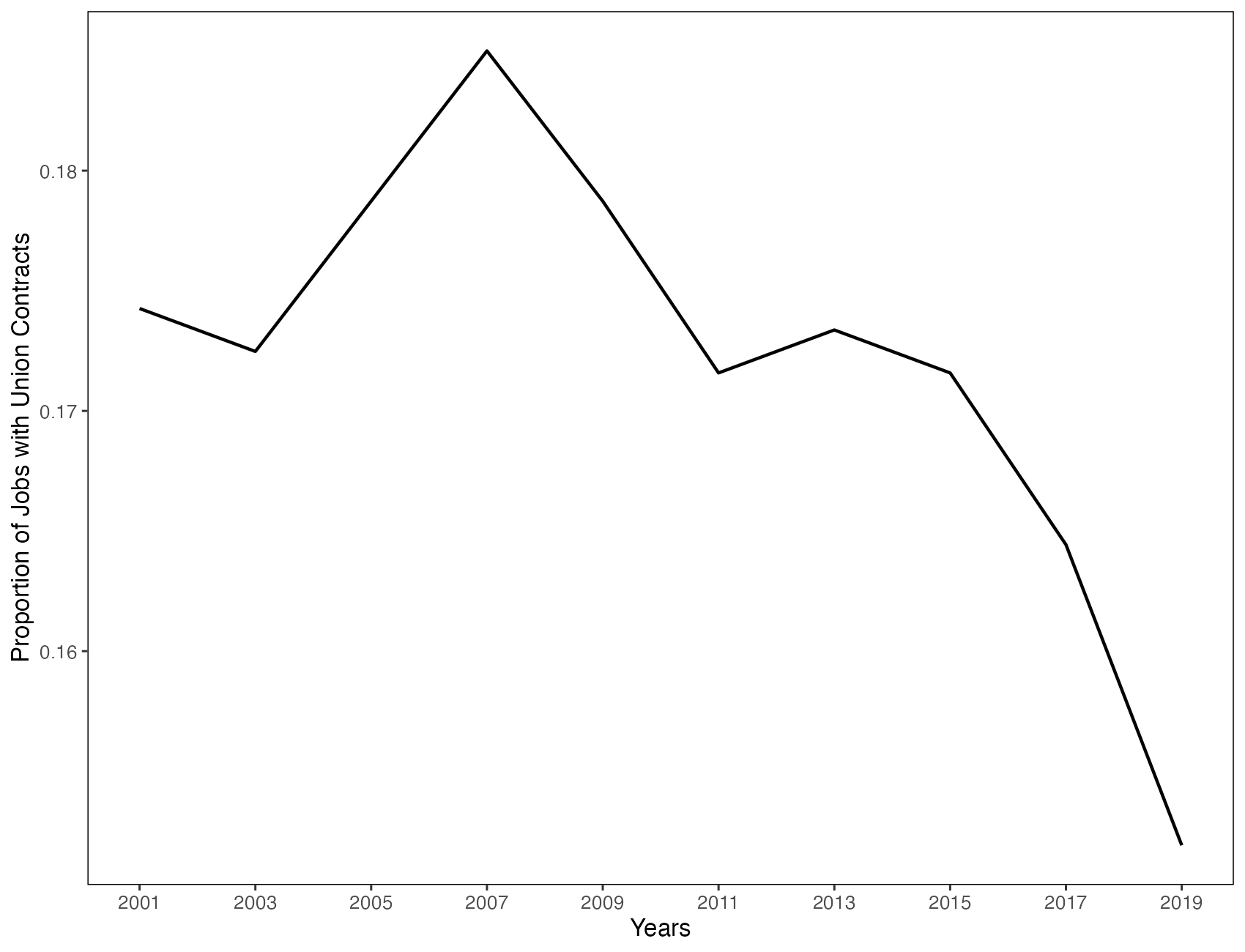}
	\caption{Proportion of workers with jobs with in a union contract}
	\label{fig:union_prop}
\end{figure}

Estimating model \eqref{eq:unionmodel} with pooled least squares and two way fixed effects estimation indicates union membership corresponds to 9.32\% and 22\% increases in real earnings, respectively. This observation is surprising since \cite{freeman:1984} pointed out the opposite ordering and worker bargaining power seems to be waning, especially compared to the proportion of unionized jobs in the older studies. One empirical explanation is that of \cite{card:1996} where unobserved ability has a complicated correlation structure with unionized job selection. This may be due to a modern labor market with more variable job outcomes due to technological advancements and the need for the associated skills. The worker has options beyond the traditional union job and so I observe in the sample consistent high earners who are compensated well outside of a union contract. Therefore, the more able may not value a union as much since employers make more competitive offers in their ability category. Wage inequality may also play as a secondary factor as the dominating wages this group might earn may induce negative correlation with union job employment, explaining why these estimates are much larger than pooled OLS estimates.

With WGFE a Bayesian Information Criteria\footnote{See the supplementary appendix of \cite{bm:2015} and \cite{bai:2002} for the BIC method of estimating the number of groups.} is used to estimate $\widehat{G} = 7$ and find that a union contract raises wages by 14.6\%, which is between the pooled least squares and two way fixed effects. Figure \ref{fig:groupresults} plots the group medians of income and proportion of union jobs. Figure \ref{fig:nonlinearcorr} shows nonlinear correlation where union jobs and earnings are positively correlated with unionization until earnings are much larger when it becomes negatively correlated. Ignoring the green and purple curves for now, focus in order of income by red, orange, blue, teal and gold. The lowest union proportions among these are red, orange and gold, where gold is the lowest income group. Teal and blue have higher earnings and the highest proportions of union jobs in these groups. Therefore at the highest stratum of income, union status is likely to be lower and as we fall in the income tiers the union proportions will increase until the levels of income are low enough where union status falls again. The green and pink group are the least stable groups that experience sharp declines in union status which amounts to a 10\% fall from over 20\% to under 10\% of all jobs. At the same time, pink has a significant fall in real earnings while green is more volatile in their earnings time series. 

\begin{figure}[h!]
	\includegraphics[scale = 0.35]{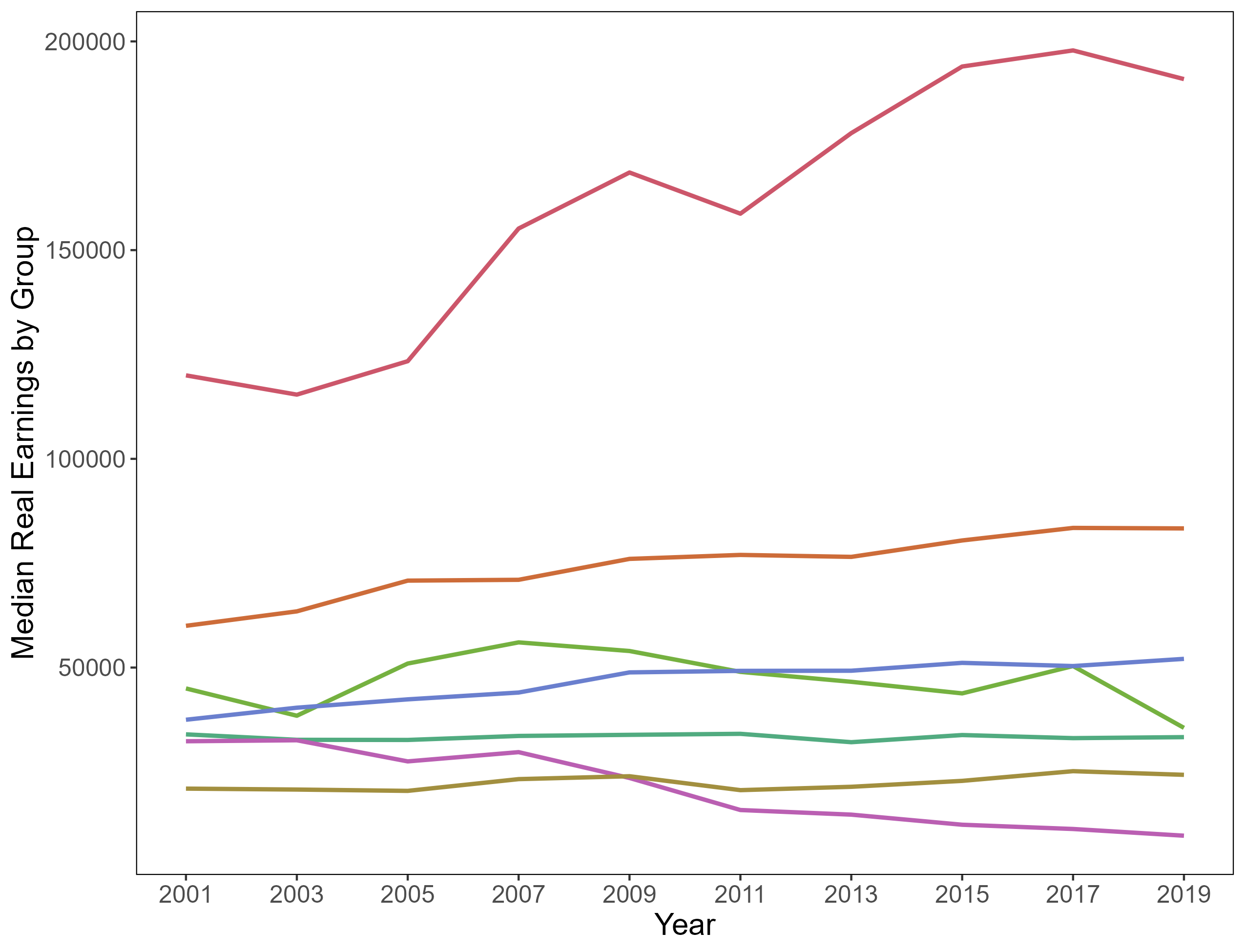}
	\includegraphics[scale = 0.35]{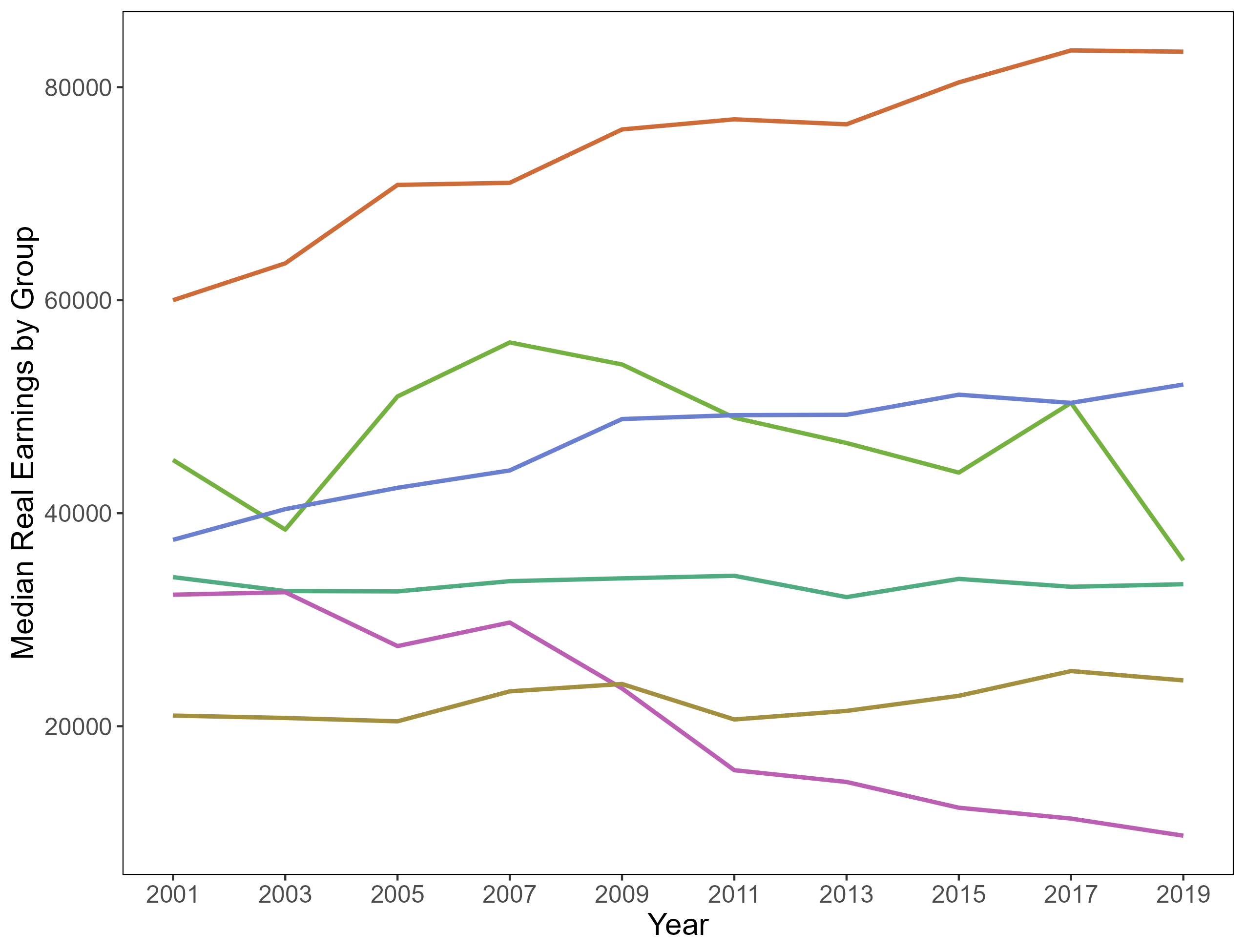}\\
	\includegraphics[scale = 0.4]{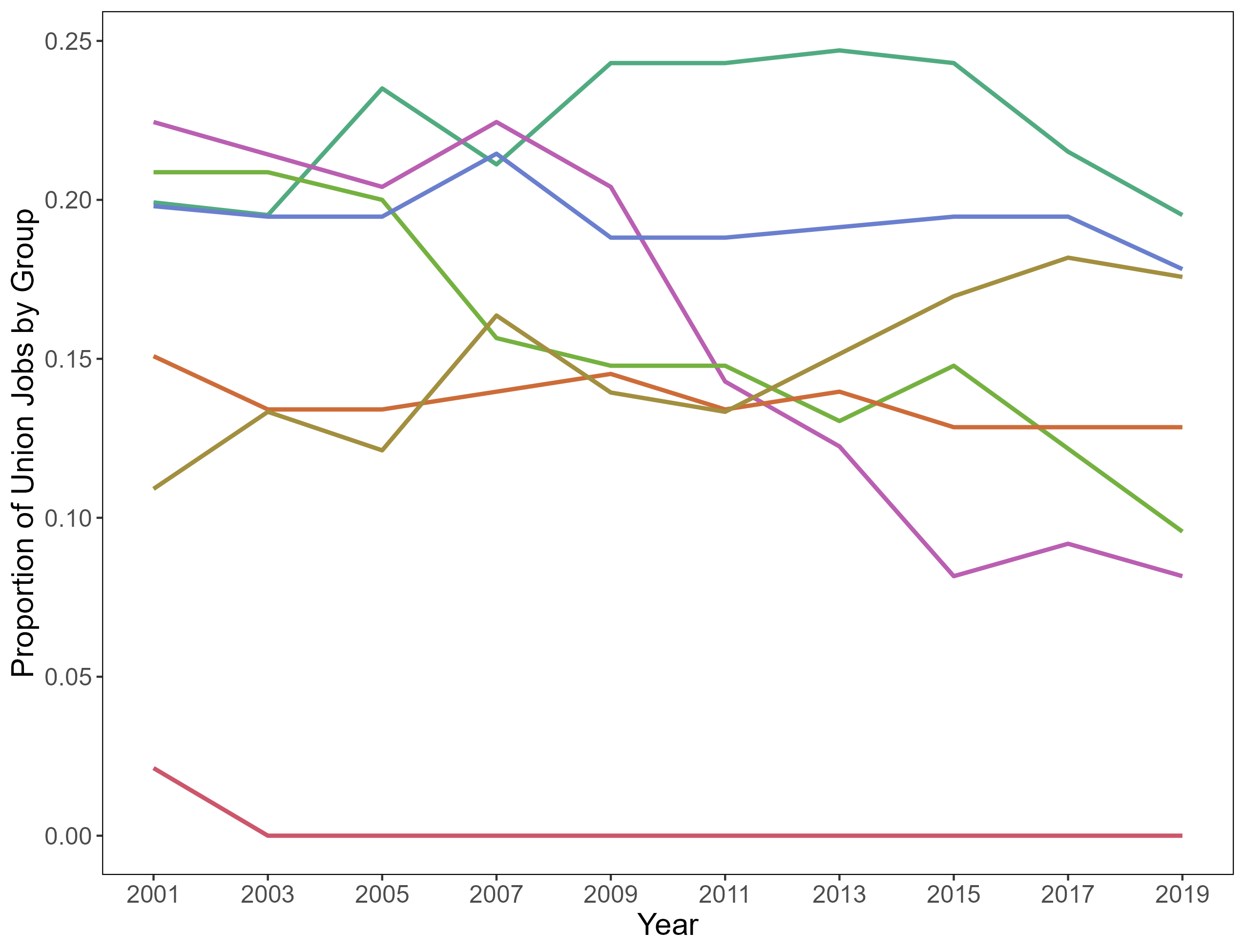}
\caption{Group median incomes where each color is a different group. Top Left: Including elite income earners, Top Right: Excluding elite income earners, Bottom: Proportion of union jobs in each group.}
\label{fig:groupresults}
\end{figure}

\begin{figure}[h!]
	\includegraphics[scale = 0.45]{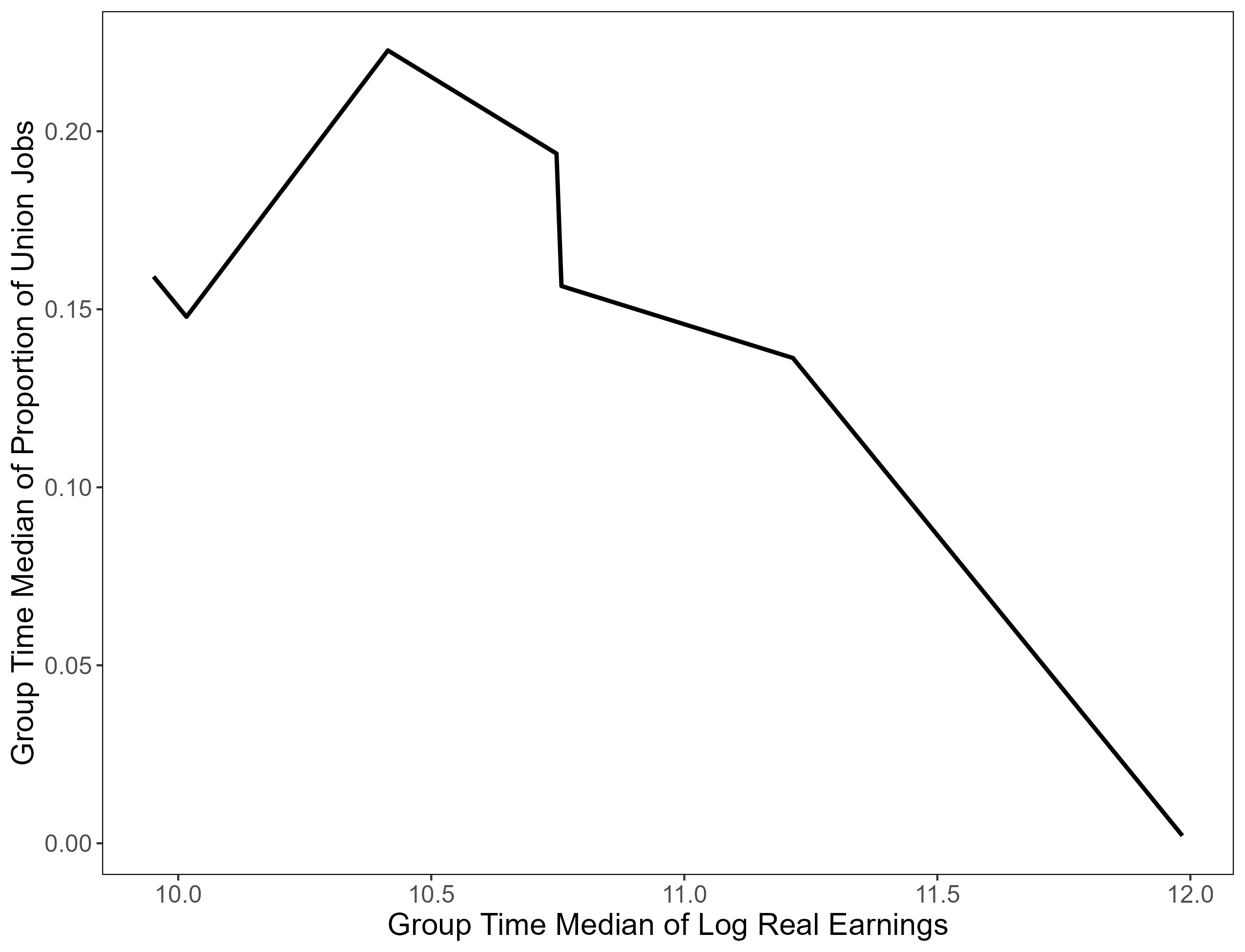}
\caption{Group time series medians of union-log earnings pairs showing concave relationship.}
\label{fig:nonlinearcorr}
\end{figure}

\subsection{Income and Democracy}

In this section I revisit \cite{acemoglu:2008} where they incorporate country fixed-effects to study the link between income and democracy in nations. They find that the positive effects of income on democracy vanish with the inclusion of country fixed effects arguing that this is consistent with unobserved historical factors that initiated divergent economic and political paths across nations. The purpose of this section is to make comparisons and find insights on how to interpret results using these clustering methods since the results are reinforced by GFE or WGFE.

Following this paper and the grouped fixed-effects approach of \cite{bm:2015}, I consider a regression of the Freedom house index of democracy on time-lagged democracy, lagged log-GDP per capita and a grouped fixed-effect term $\alpha_{g_i t}$, which captures group-specific time varying heterogeneity within group $g_i = g$ for some $g = 1,\dots,G$. The model is written as
\[
democracy_{it} = \theta_1 \, democracy_{i t-1} + \theta_2 \, logGDPpc_{i t-1} + \alpha_{g_i t} + u_{it}.
\]
\begin{center}
 \begin{figure}[h!] \centering  \begin{tabular}{@{}cccccccccc@{}}
 \toprule  &\phantom{abc} &\multicolumn{2}{c}{Objective} &\phantom{abs}  &\multicolumn{2}{c}{Lagged Dem. Term ($\theta_1$)} &\phantom{abs} & \multicolumn{2}{c}{Income Term ($\theta_2$)}\\
 \cmidrule{3-4}  \cmidrule{6-7} \cmidrule{9-10}\\[2mm] 
 $G$&& WGFE & GFE && WGFE & GFE && WGFE & GFE  \\ \midrule 
1 && $-$ & $-$ && $-$ & $\underset{(0.049)}{0.665}$ && $-$  & $\underset{(0.014)}{0.083}$\\[5mm] 
  2&& 0.1719& 19.847&& $\underset{(0.045)}{0.554}$ & $\underset{(0.044)}{0.600}$ && $\underset{(0.0095)}{0.062}$ & $\underset{(0.0118)}{0.061}$\\[5mm] 
  3&&  0.1522 & 16.599 && $\underset{(0.048)}{0.403}$ & $\underset{(0.055)}{0.407}$ && $\underset{(0.0089)}{0.070}$ & $\underset{(0.0113)}{0.089}$  \\[5mm] 
  4&&  0.1415 &14.319 && $\underset{(0.051)} {0.425}$ & $\underset{(0.058)}{0.302}$ && $\underset{(0.0091)}{0.065}$ & $\underset{(0.0098)}{0.082}$ \\[5mm]  
  5&&  0.1325 & 12.593 && $\underset{(0.047)}{0.455}$ & $\underset{(0.048)}{0.255}$  && $\underset{(0.0089)}{0.062}$ & $\underset{(0.0086)}{0.079}$  \\[5mm]  
 6&&  0.1252 & 11.132 && $\underset{(0.034)}{0.539}$ & $\underset{(0.046)}{0.465}$ && $\underset{(0.0083)}{0.042}$ & $\underset{(0.0075)}{0.064}$ \\[5mm]  
 7&&  0.1182 & 10.059&& $\underset{(0.038)}{0.483}$  & $\underset{(0.0423)}{0.403}$  && $\underset{(0.0080)}{0.040}$ & $\underset{(0.0079)}{0.065}$ \\[5mm]  
 FE &&  $-$ & $-$ && $-$  & $\underset{(0.058)}{0.284}$  && $-$ & $\underset{(0.069)}{-0.044}$ \\[5mm]  
 \bottomrule\\
 \end{tabular} 
 \caption{WGFE and GFE estimates of $\theta_1$ and $\theta_2$ over number of groups $G=1,\dots,7$, and fixed-effects (FE). Standard errors in parenthesis calculated via formulas in Appendix \ref{sec:ses}.}
 \label{fig:coeff-inc-dem}
  \end{figure}  
\end{center}
Figure \ref{fig:coeff-inc-dem} shows the GFE and WGFE estimates and standard errors of the lagged democracy and income coefficients over different total number of groups $G$ using the 1970-2000 balanced subsample of \cite{acemoglu:2008}. For reference, the results of the pooled regression ($G=1$) is reported under the GFE columns. First note the decrease in both estimates from the pooled estimate, which is consistent with some positive correlation between the unobservable and lagged democracy possibly due to historical/political events. WGFE standard errors are smaller than the GFE standard errors of the estimator of the lag term coefficient with the exception of $G=2$. However, for the income term in larger group numbers the standard errors for GFE are smaller. This can likely be attributed to multicollinearity between the unobservable, that represents additional second moment information and income. In other words the latent factor that is uncovered is more correlated with income, which was also a result found by \cite{kim:2019} in their Bayesian approach to group heteroskedasticity. As for the parameter estimates, the WGFE estimates of the lagged democracy term appear to be more robust to the number of groups $G$ by being more consistent, possibly indicating that the latent variable is being better controlled for in the regression.

\begin{figure}\centering
	\begin{subfigure}{\textwidth}\centering
		\includegraphics[width=0.3\textwidth]{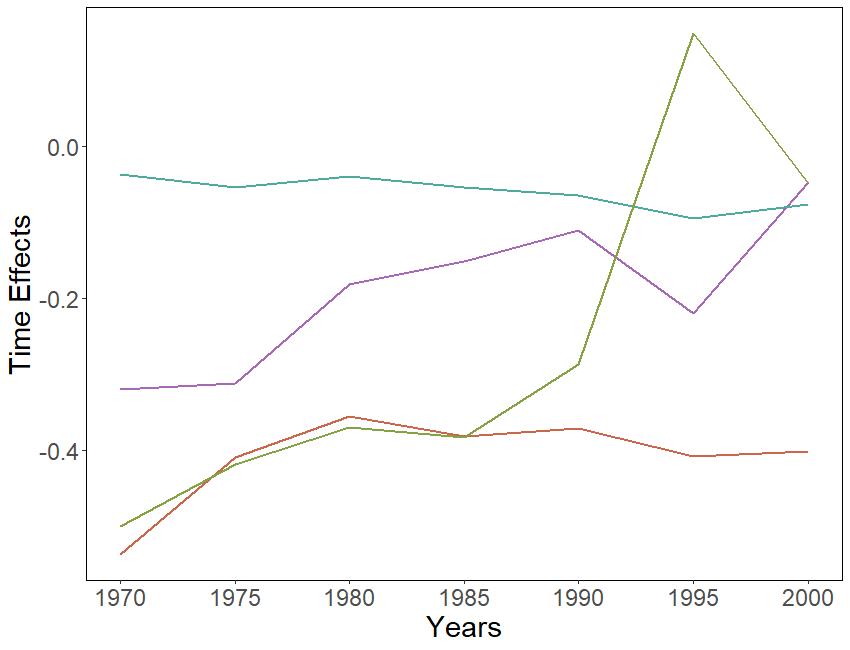}
		\includegraphics[width=0.3\textwidth]{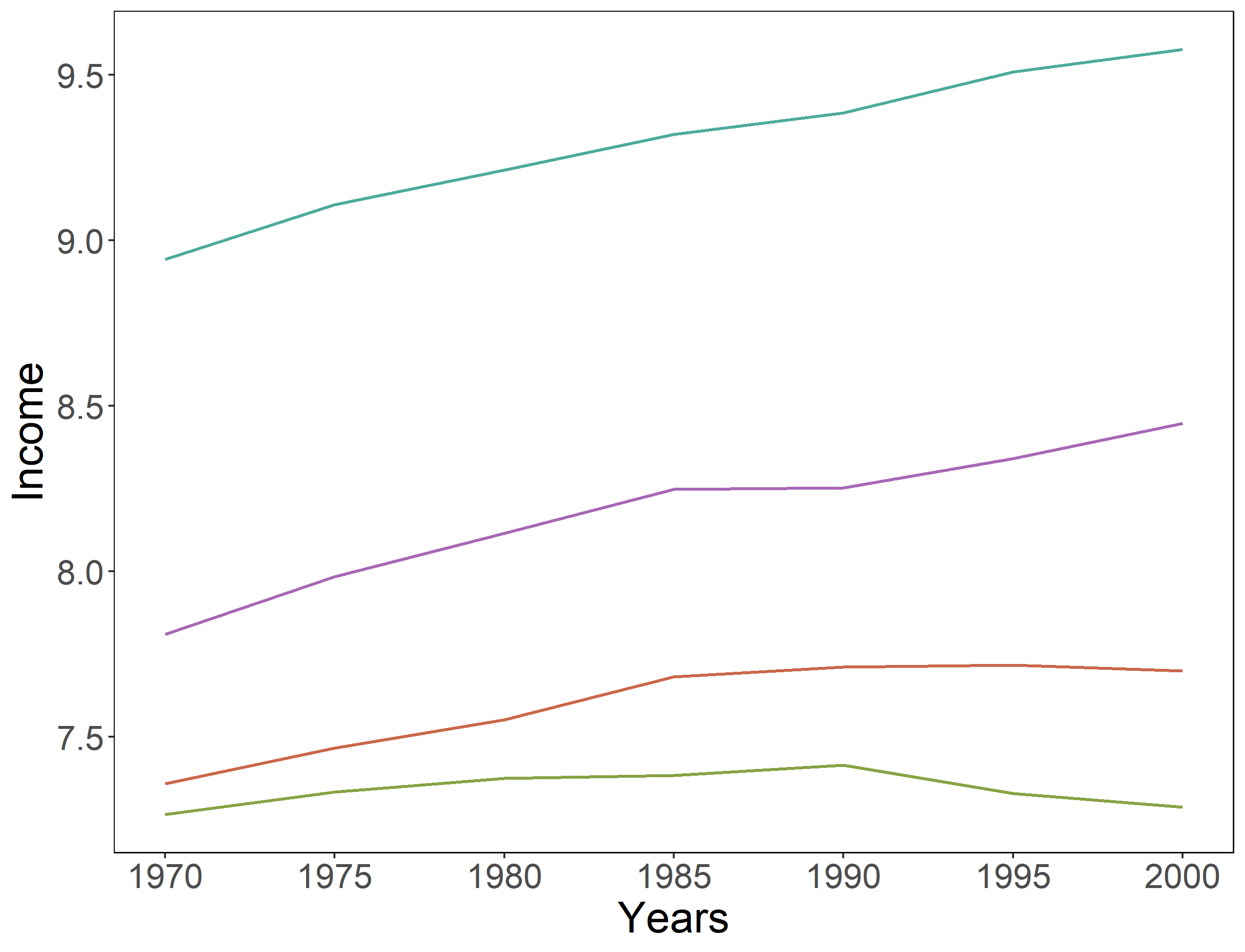}
		\includegraphics[width=0.3\textwidth]{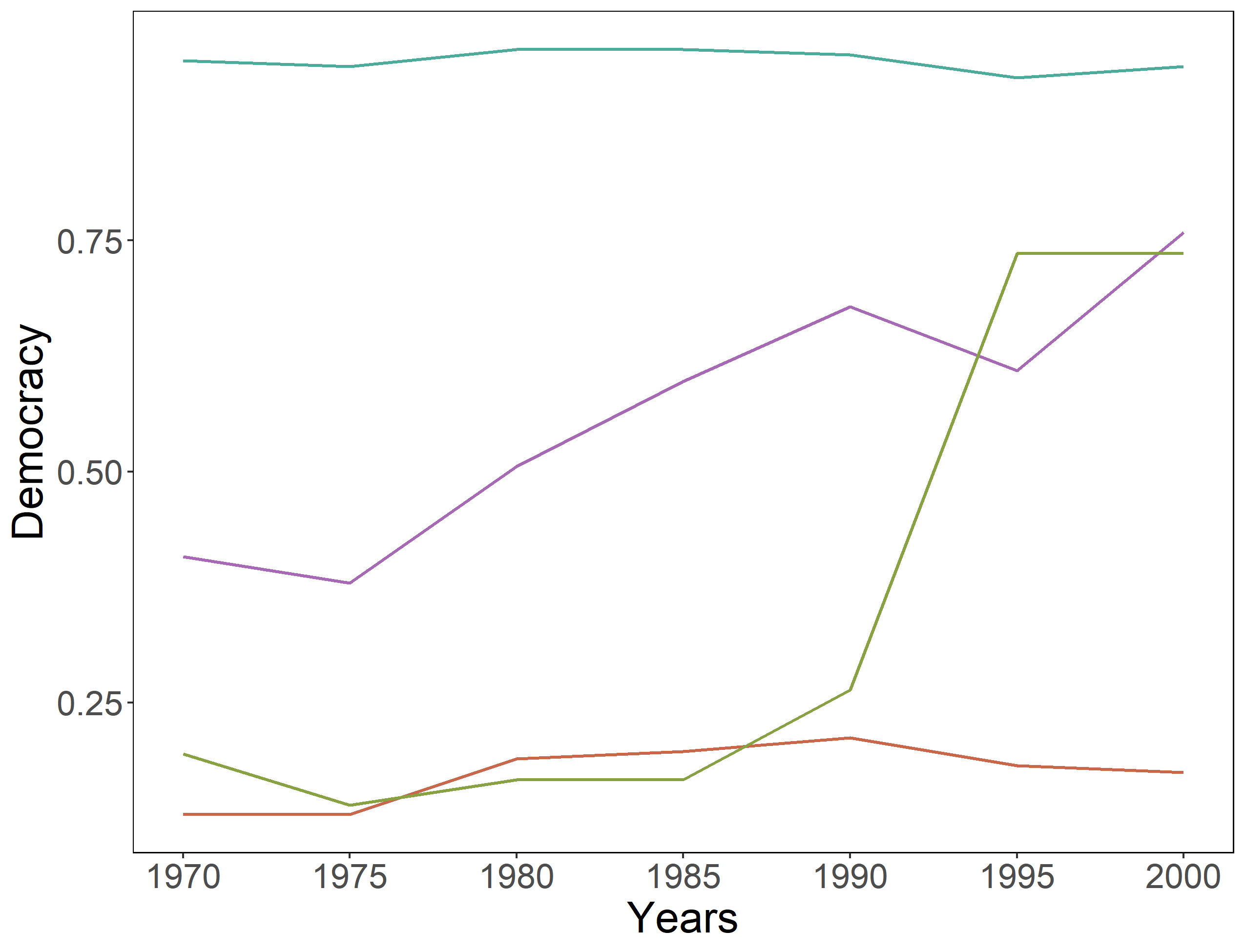}
		\caption{WGFE group assignments}
	\end{subfigure}
	\hfill
	\begin{subfigure}{\textwidth}s\centering
		\includegraphics[width=0.3\textwidth]{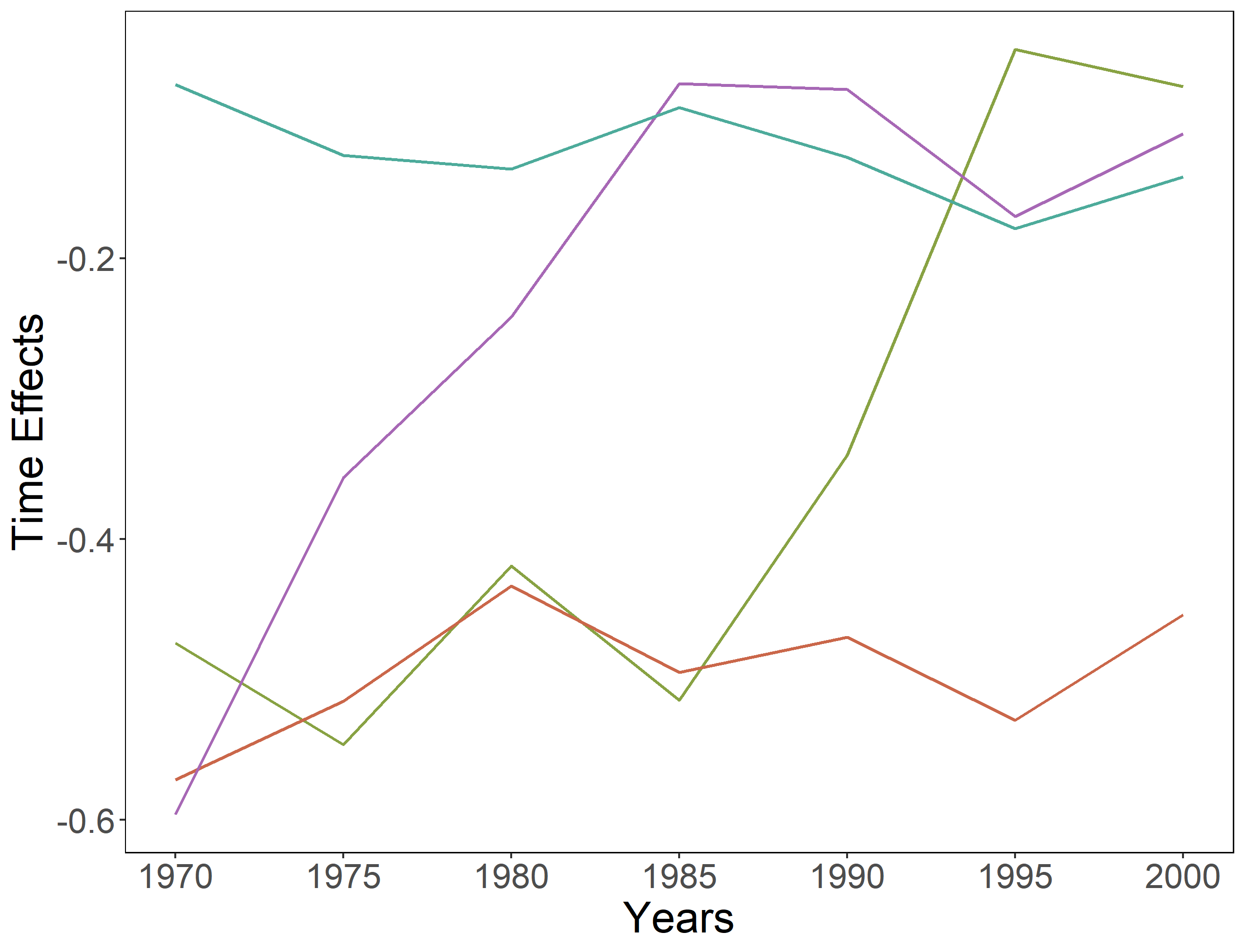}
		\includegraphics[width=0.3\textwidth]{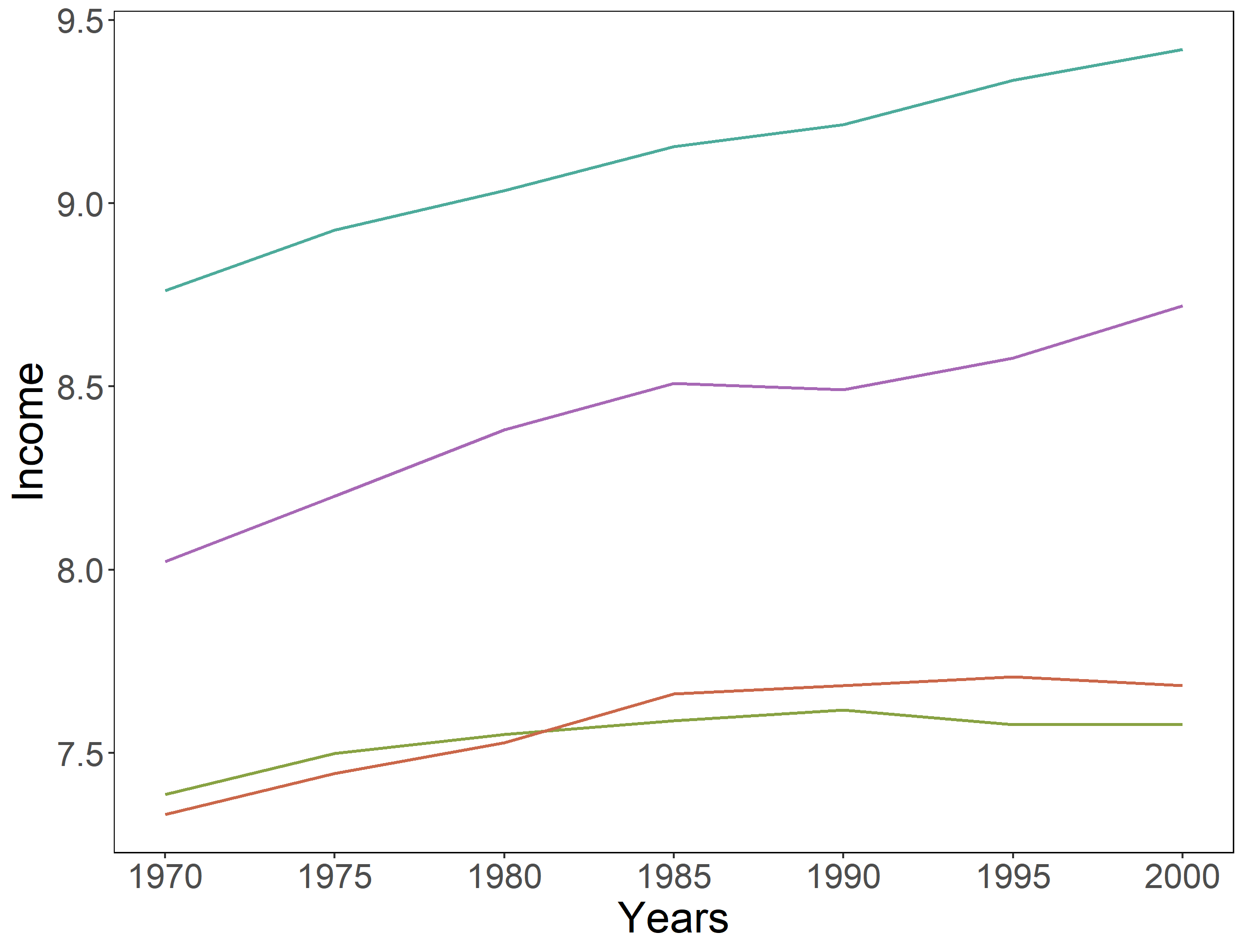}
		\includegraphics[width=0.3\textwidth]{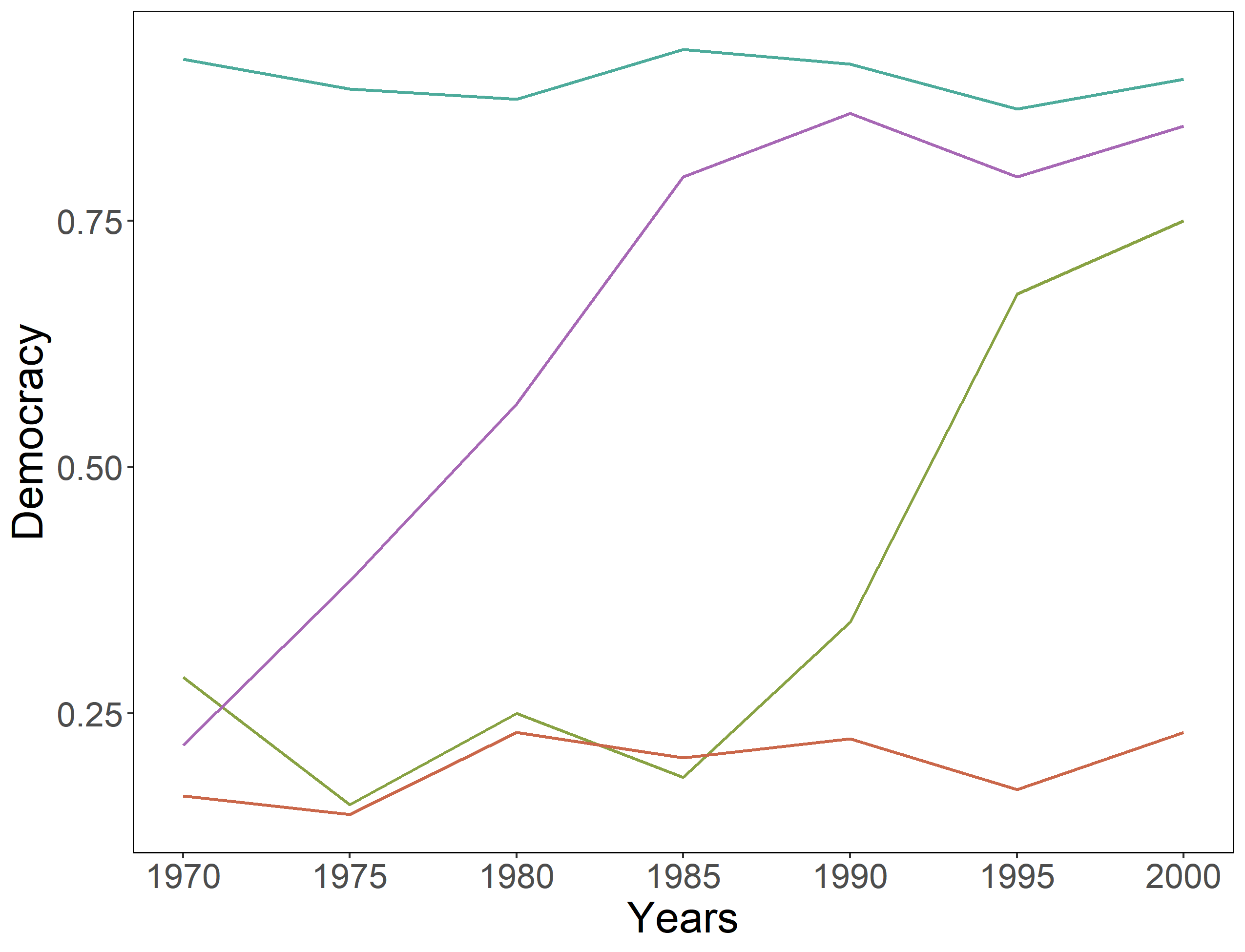}
	\caption{GFE group assignments}
	\end{subfigure}
\caption{Left: Group-specific time effects $\alpha_{gt}$ from WGFE estimates (top) and GFE estimates (bottom).
Middle: Within-group average income. Right: Within-group average of democracy index. Blue: High Democ., Orange: Low Democ., Green: Late transitioners, Purple: Early transitioners.}
\label{fig:time-effects-dem}
\end{figure}

In Figure \ref{fig:time-effects-dem} I have estimated the time effects using WGFE and GFE estimation in the case of $G=4$ groups. I follow the main discussion of \cite{bm:2015} and also report the alternative specification of $G=5$ in the Appendix \ref{app:moreresults} (as in \cite{kim:2019}). We see heterogeneous time patterns that are clearly distinct from each other. We also see two, well separated time paths in blue and orange, which are known as the low and high democracy groups following the convention of previous work. The high-democracy groups include most European Union countries, the USA, and also Colombia, Japan and India. As for the low-democracy groups both assignment rules result in China, Iran, and many African countries. These two groups run parallel to each other across time. 

In contrast, there are two additional groups detected that exhibit transitions from low-democracy to high-democracy in the sample time periods. Keeping with convention, I identify the group that makes the transition sooner as the early-transition group (purple) and the other as the late-transition group (green). The group assignments between these groups are different between the two methods, however both detect Argentina, South Korea, Spain and Greece as early transitioners and Panama, Romania and Taiwan as late transitioners. Differences of the WGFE assignments come from the early and late transitioners poaching members from either the high-democracy or low-democracy groups as determined by GFE assignments. For a complete list of assignments and differences between WGFE and GFE see Appendix \ref{app:moreresults}, Figure \ref{fig:wgfe-list}.

The WGFE may provide more consistent groupings for GFE based on historical account. For instance, the Dominican Republic had unstable democratic institutions throughout the 1960's and their first elected president was the proxy president for their last dictator, who remained in power until 1978. His regime was marked by poor human rights and civil liberties where restrictions were placed on opposition parties. Despite these factors and how the Freedom index measures democracy, GFE assigns the Dominican Republic to the high democracy group while WGFE places the country in the early transitioners. This and some other examples (Greece, Cyprus and Turkey) display a potential advantage of WGFE assignment detecting groupings more robust to the instability that might be found in transitioning groups. 
 
Turning attention to the group fixed effects themselves, the WGFE effects of the low and high democracy groups appear smoother compared to GFE. The late transitioners follow the low democracy group close up until 1985 when it sharply rises. The early transitioners follow a more stable upward trend to the high democracy group. As for the income and democracy plots, the WGFE assignments of groups display more positive correlation between the latent variable and income than GFE. This might explain the multicollinearity issue raised for the coefficient estimate of income using the WGFE estimator. 
\begin{center}
 \begin{figure}[h!] \centering  \begin{tabular}{@{}ccccc@{}}
 \toprule Group & High & Early & Late & Low  \\ \midrule 
$P_g$ & 26 & 13 & 18 & 33\\[5mm] 
  $\sigma_g$ & 0.133 & 0.213 & 0.212 & 0.181 \\
 \bottomrule\\
 \end{tabular} 
 \hspace{1cm}
\begin{tabular}{@{}ccccc@{}}
 \toprule Group & High & Early & Late & Low  \\ \midrule 
$P_g$ & 27 & 29 & 12 & 22\\[5mm] 
  $\sigma_g$ & 0.186 & 0.293 & 0.217 & 0.208 \\
 \bottomrule\\
 \end{tabular} 
 \caption{GFE (Left) and WGFE (Right) estimates of group sizes and variances.}
 \label{fig:sizes-var}
  \end{figure}  
\end{center}
Concerning variances, the late and early transitioners are found to be much different with WGFE estimates where the early transitioners have the most variable democracy outcomes. The stable groups still have the lowest variances, but they have increased from the GFE version. As for group sizes, Huntington's theory suggests democracy was on the rise in this period, so we should observe that most countries are either democracies or early transitioners. From Figure \ref{fig:sizes-var} and \ref{fig:wgfe-list} we see that democracy under WGFE is more dominant in the sample period with more early transitioners, as opposed to predictions from GFE which displays much less early transitioners and a large group of late transitioners. WGFE appears more in line with Huntington's theory in this sense.
\begin{figure}[h!]
	\includegraphics[width = \textwidth]{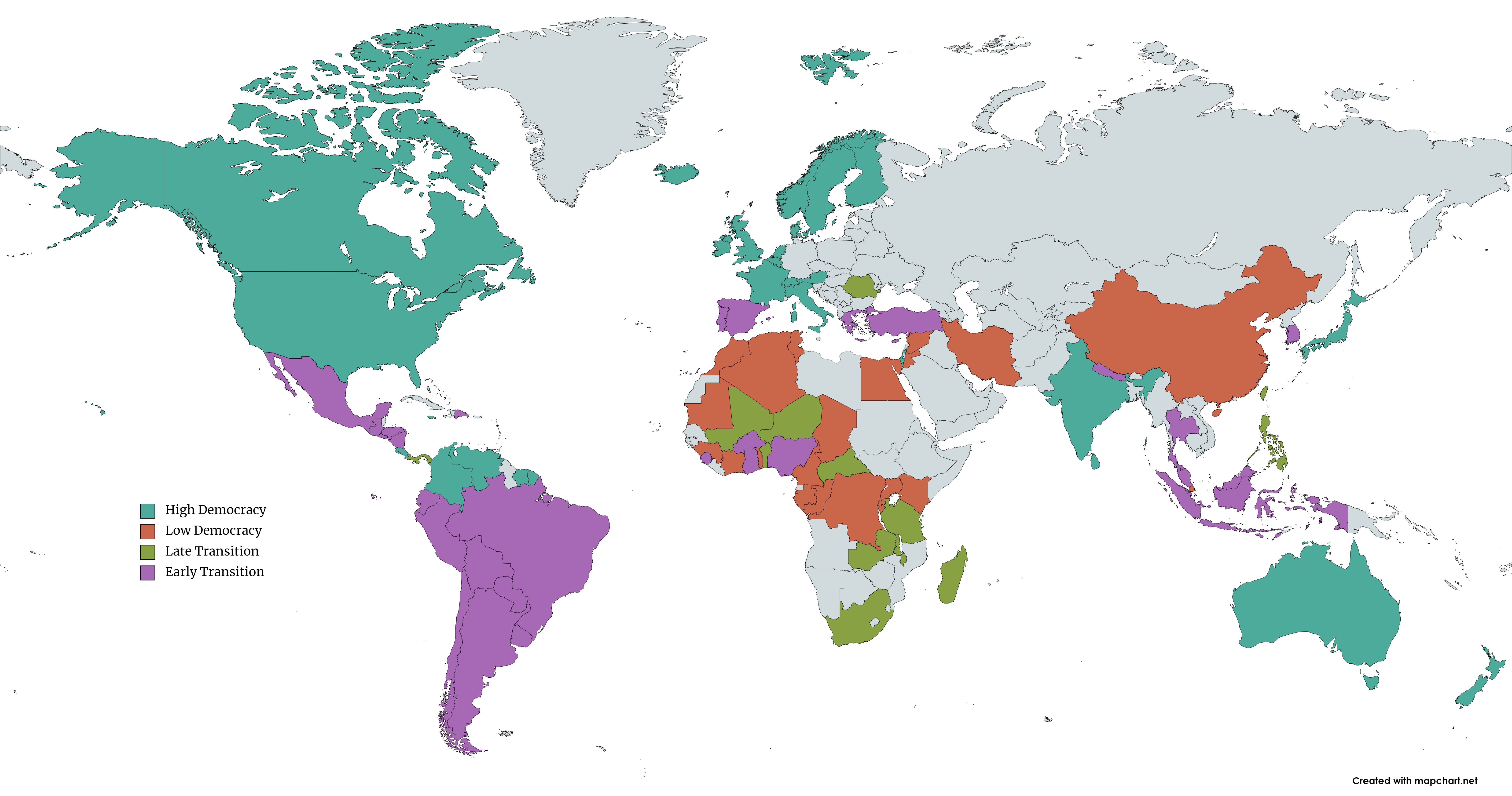}
	\includegraphics[width = \textwidth]{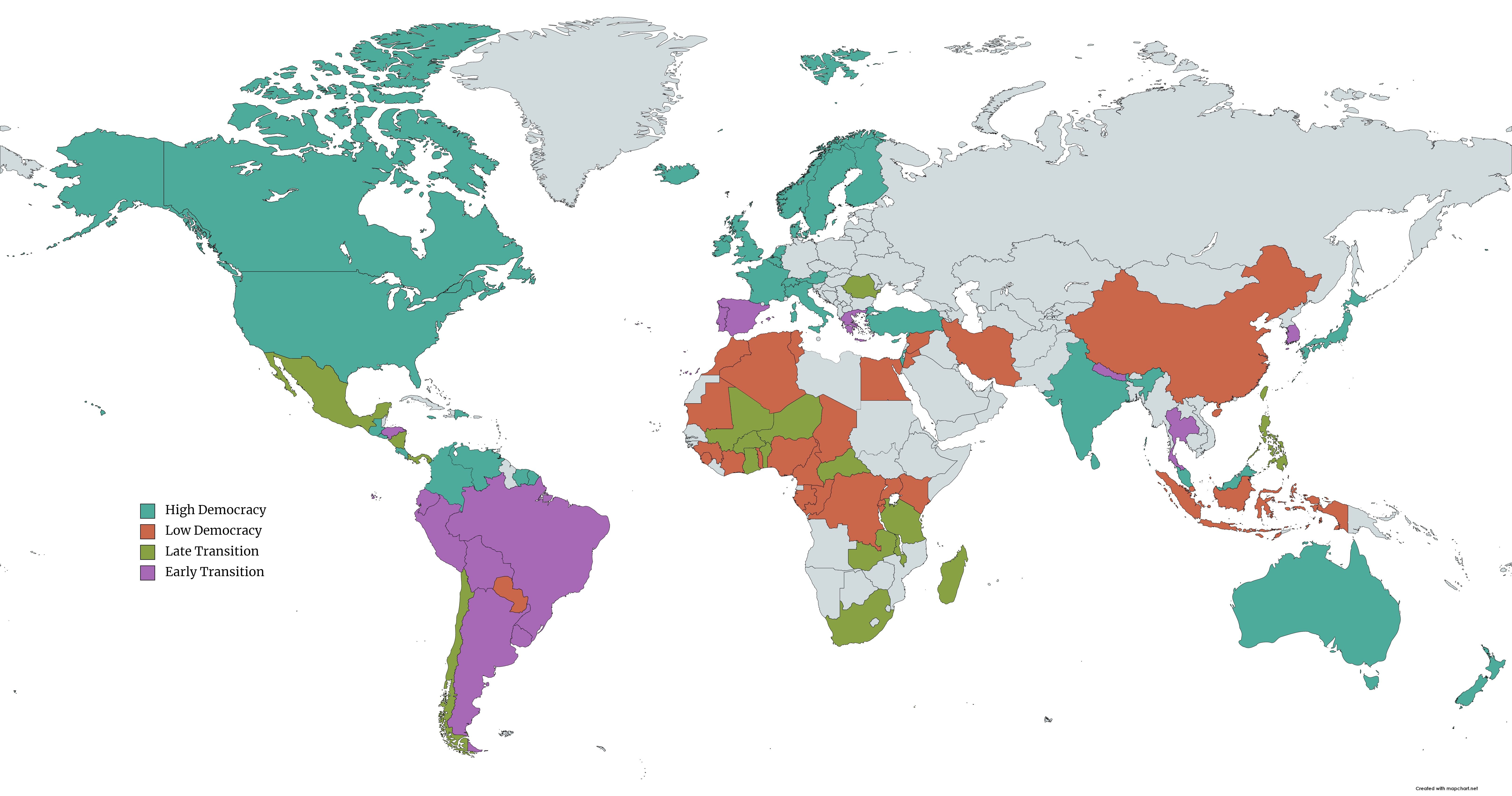}
\caption{$G = 4$ map with Top: WGFE groups, and Bottom: GFE groups}
\label{fig:maps4}
\end{figure}

In Figure \ref{fig:maps4} we see even more local clustering of democracies and early transitioners. Indeed, South America seems to be immersed in transition while Africa being unstable in the sample period having a mix of early, late transitioners with low democracy countries. In the Eastern Mediterranean region I see Turkey, Greece and Cyprus all transitioning as opposed to having Turkey as a high democracy country. In Asia, I see transitioners sandwiched by low democracy China and high democracies Australia, India and Japan. I emphasize the absence of any structure imposed on groupings so this apparent geographical clustering is all the result of estimation.

\newpage
\bibliographystyle{abbrvnat}
\bibliography{cluster-bib}

\newpage
\newgeometry{top=0.5in,bottom=0.5in,right =1in,left =1in}
\appendix
\section{Additional Results}\label{app:moreresults}
\subsection{Income and Democracy}

\begin{center}
 \begin{figure}[h!] \centering
	\begin{subfigure}{\textwidth}\centering
		\begin{tabular}{@{}ccc@{}}
		 \toprule WGFE & High & Low  \\ \midrule 
		$P_g$ & 32 &  58\\[5mm] 
		  $\sigma_g$ & 0.28 &  0.29\\
		 \bottomrule\\
		 \end{tabular} 
		 \hspace{1cm}
		 \begin{tabular}{@{}ccc@{}}
		 \toprule GFE & High & Low  \\ \midrule 
		$P_g$ & 49 &  41\\[5mm] 
		  $\sigma_g$ & 0.21 & 0.20 \\
		 \bottomrule\\
		 \end{tabular} 
	\caption{$G = 2$}
	\end{subfigure}
	\hfill\vspace{1cm}
	\begin{subfigure}{\textwidth}\centering
		\begin{tabular}{@{}cccc@{}}
		 \toprule WGFE & High & Transition & Low  \\ \midrule 
		$P_g$ & 29 &  39 & 22 \\[5mm] 
		  $\sigma_g$ & 0.22 & 0.31 & 0.24 \\
		 \bottomrule\\
		 \end{tabular} 
		 \hspace{1cm}
		 \begin{tabular}{@{}cccc@{}}
		 \toprule GFE & High & Transition & Low  \\ \midrule 
		$P_g$ & 38 &  24 & 28\\[5mm] 
		  $\sigma_g$ & 0.26 & 0.29 & 0.24 \\
		 \bottomrule\\
		 \end{tabular}
	\caption{$G = 3$} 
	\end{subfigure}
	\hfill\vspace{1cm}
	\begin{subfigure}{\textwidth}\centering
		\begin{tabular}{@{}cccccc@{}}
		 \toprule WGFE & High & Early & Middle & Late & Low  \\ \midrule 
		$P_g$ & 27&  18& 12 & 11& 22 \\[5mm] 
		  $\sigma_g$ & 0.19 & 0.27 & 0.27 & 0.20 & 0.20\\
		 \bottomrule\\
		 \end{tabular} 
		 \hspace{1cm}
		\begin{tabular}{@{}cccccc@{}}
		 \toprule GFE & High & Early & Middle & Late & Low  \\ \midrule 
		$P_g$ & 30&  12 & 14 & 13 & 21\\[5mm] 
		  $\sigma_g$ & 0.09 & 0.21 & 0.19  & 0.18 & 0.14\\
		 \bottomrule\\
		 \end{tabular} 
	\caption{$G = 5$} 
	\end{subfigure}
 \caption{ WGFE (Left) and GFE (Right) estimates of group sizes and variances.}
 \label{fig:sizes-var}
  \end{figure}  
\end{center}

\begin{figure}\centering
	\begin{subfigure}{\textwidth}\centering
		\includegraphics[width=0.3\textwidth]{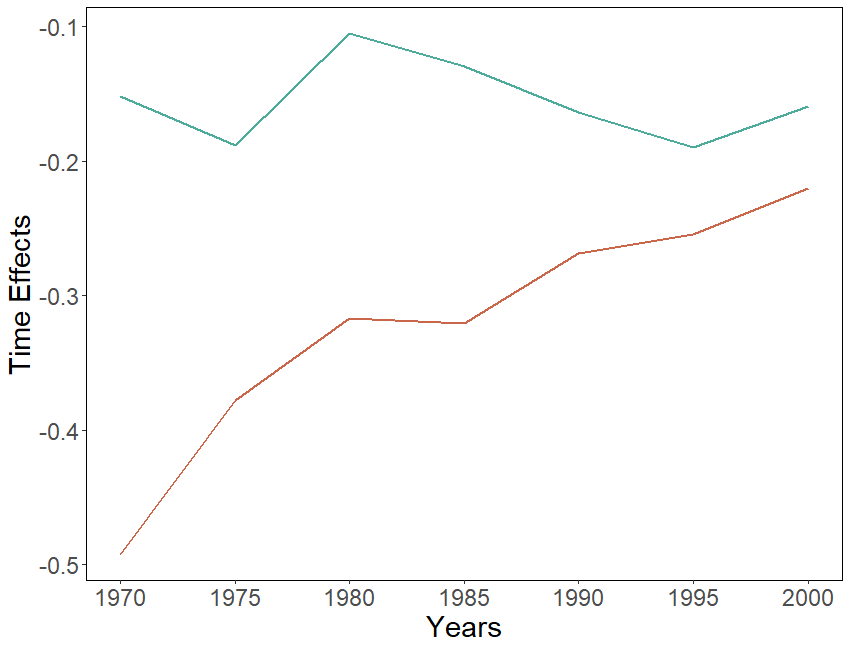}
		\includegraphics[width=0.3\textwidth]{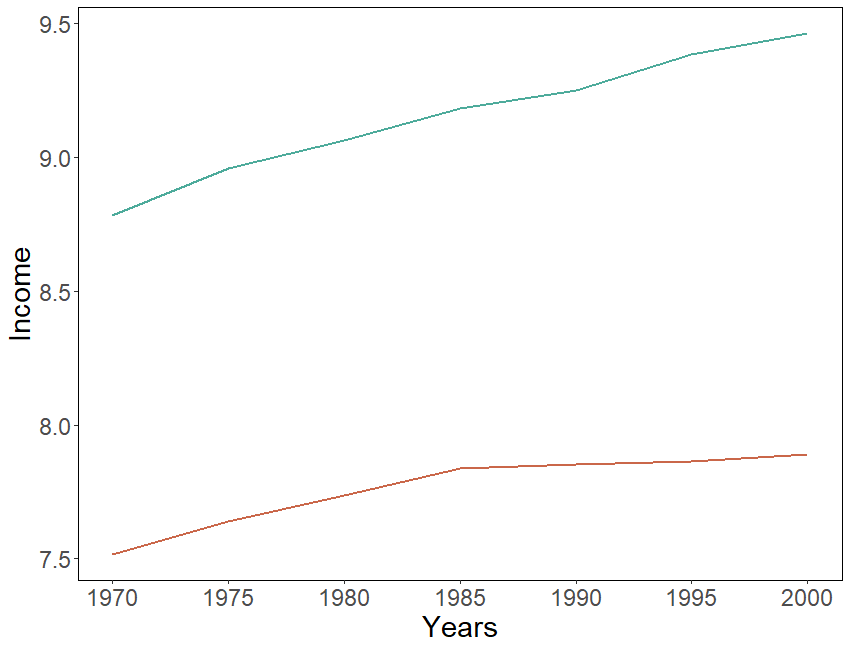}
		\includegraphics[width=0.3\textwidth]{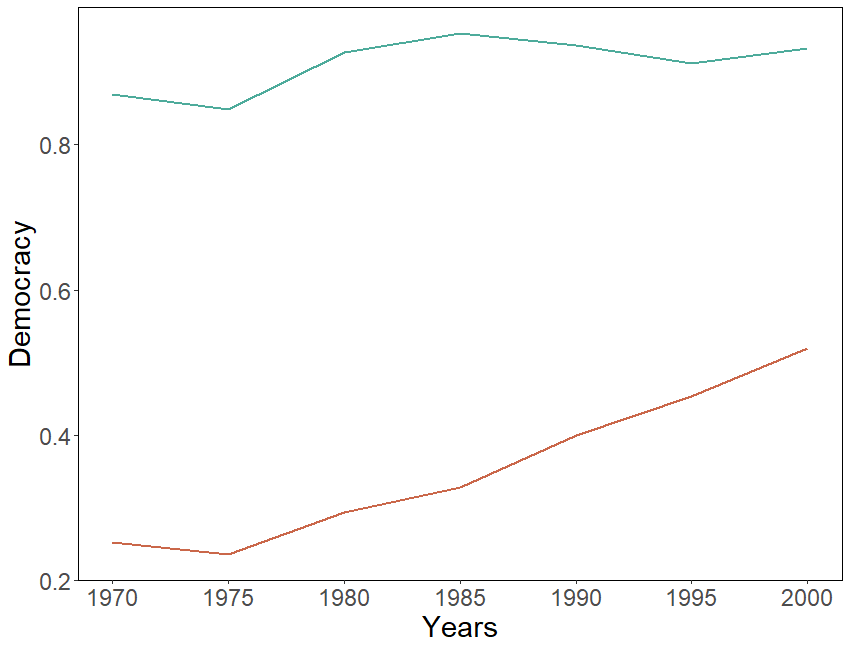}		
		\includegraphics[width=0.3\textwidth]{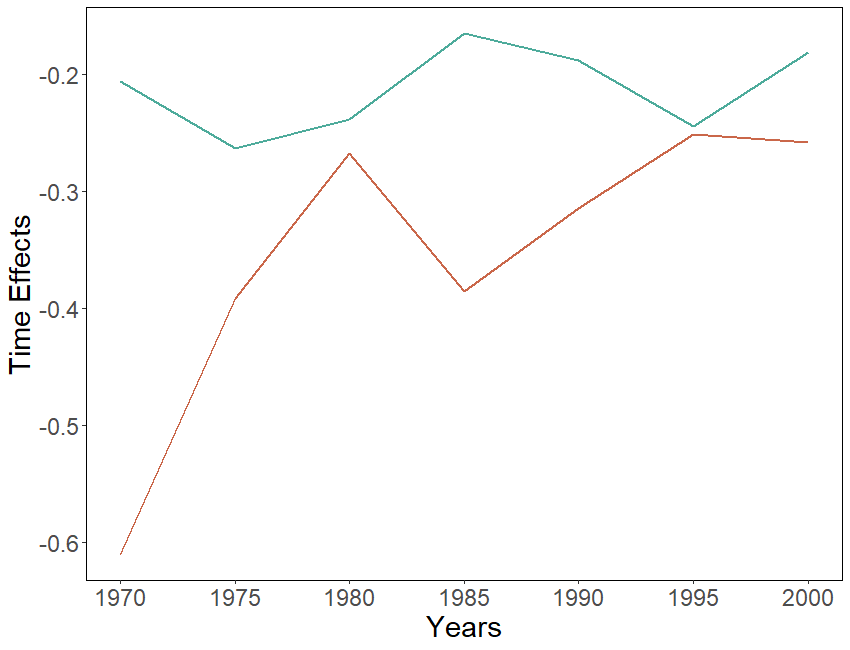}
		\includegraphics[width=0.3\textwidth]{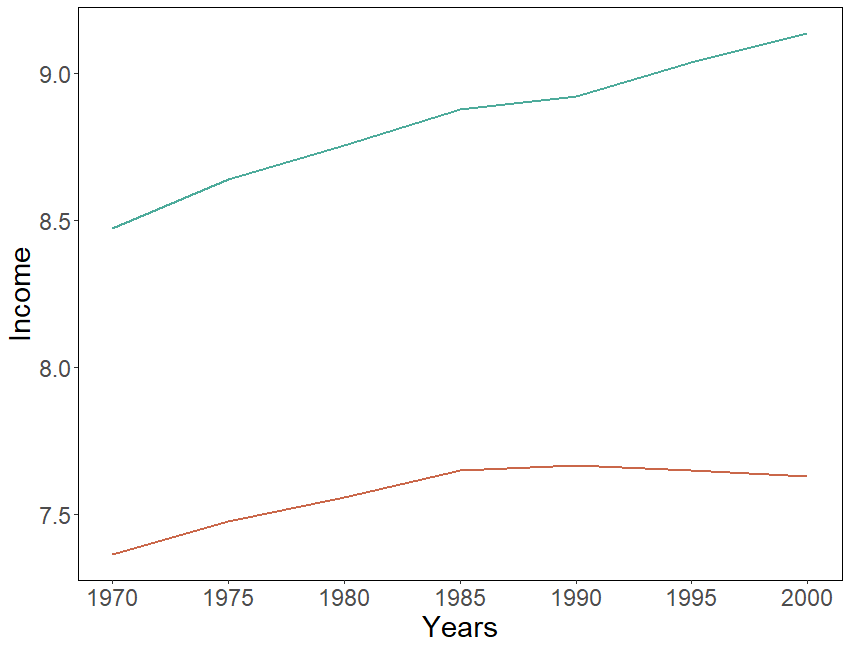}
		\includegraphics[width=0.3\textwidth]{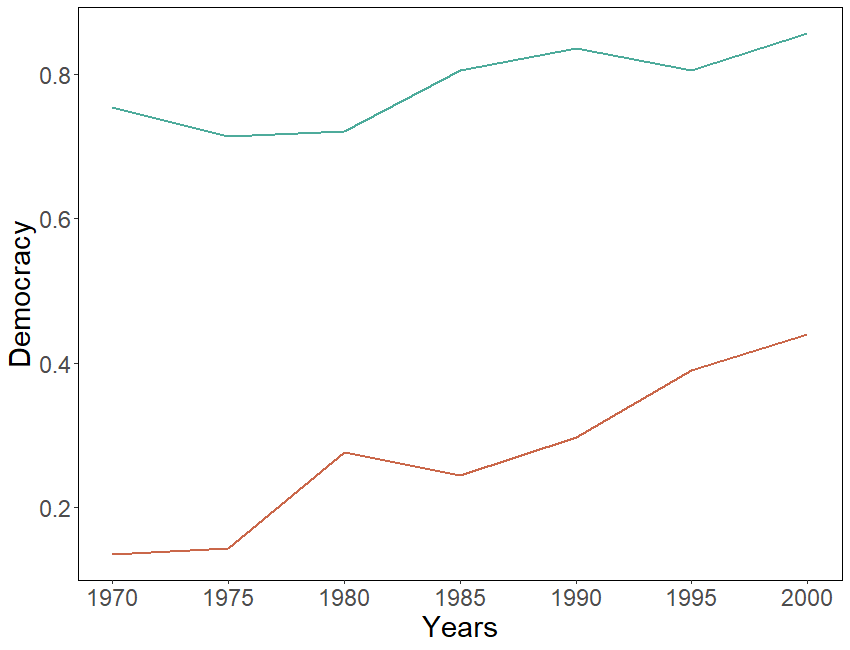}
		\caption{\scriptsize $G=2$ and Top: WGFE, Bottom: GFE}
	\end{subfigure}
	\hfill
	\begin{subfigure}{\textwidth}s\centering
		\includegraphics[width=0.3\textwidth]{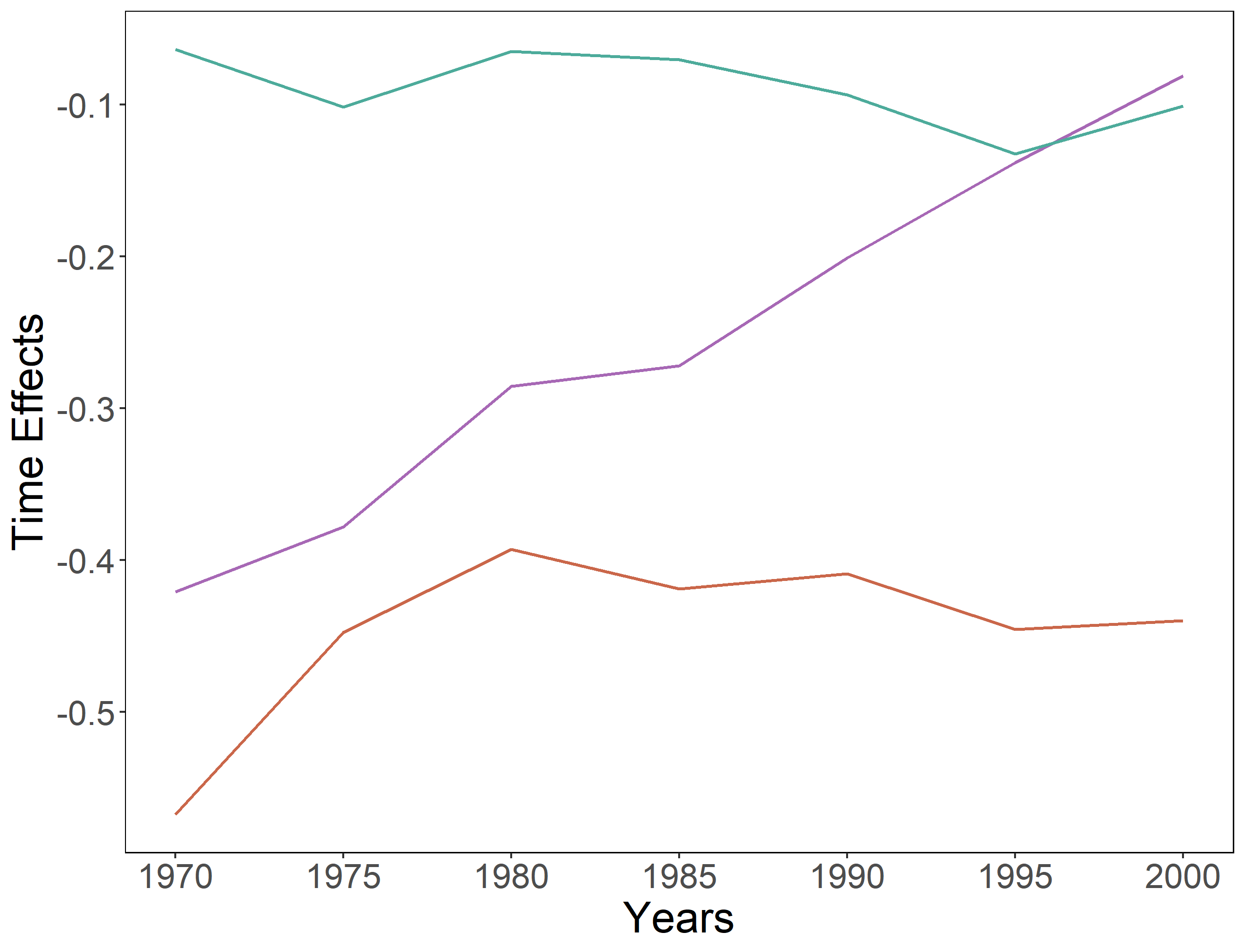}
		\includegraphics[width=0.3\textwidth]{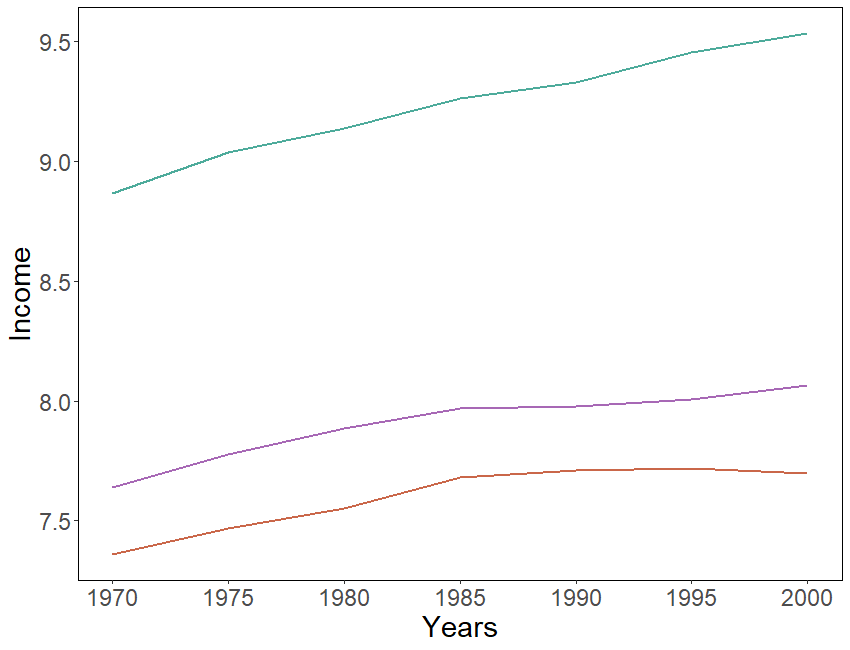}
		\includegraphics[width=0.3\textwidth]{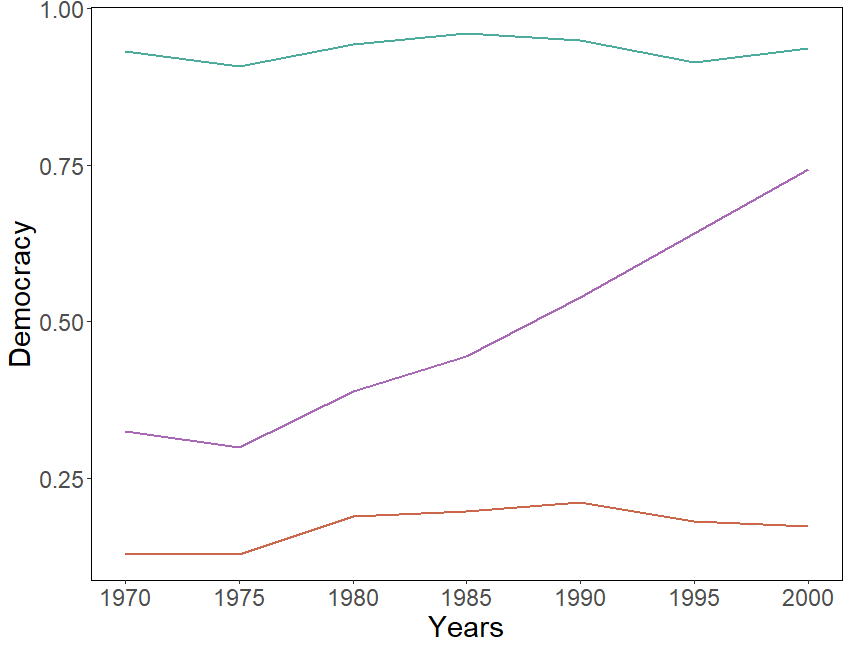}		
		\includegraphics[width=0.3\textwidth]{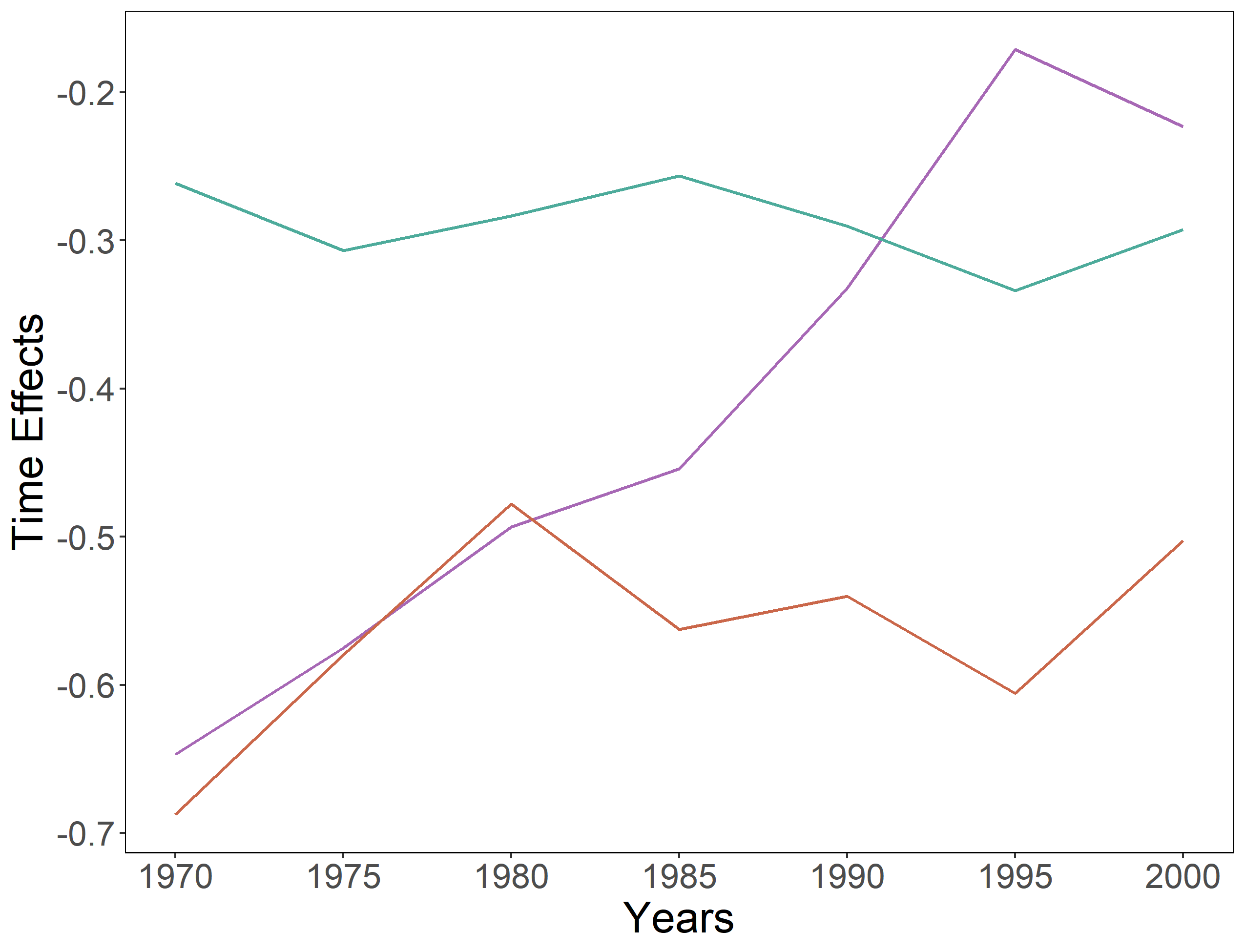}
		\includegraphics[width=0.3\textwidth]{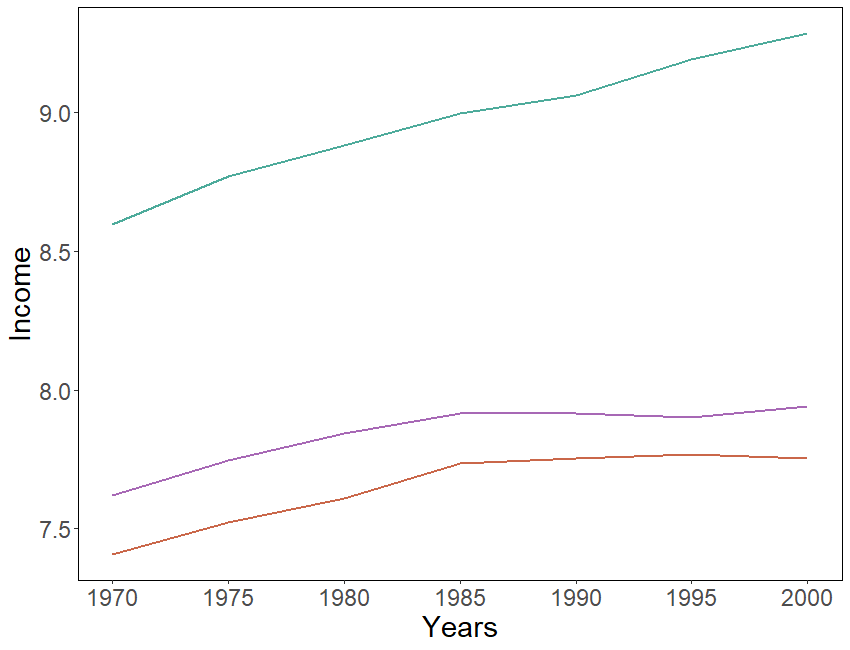}
		\includegraphics[width=0.3\textwidth]{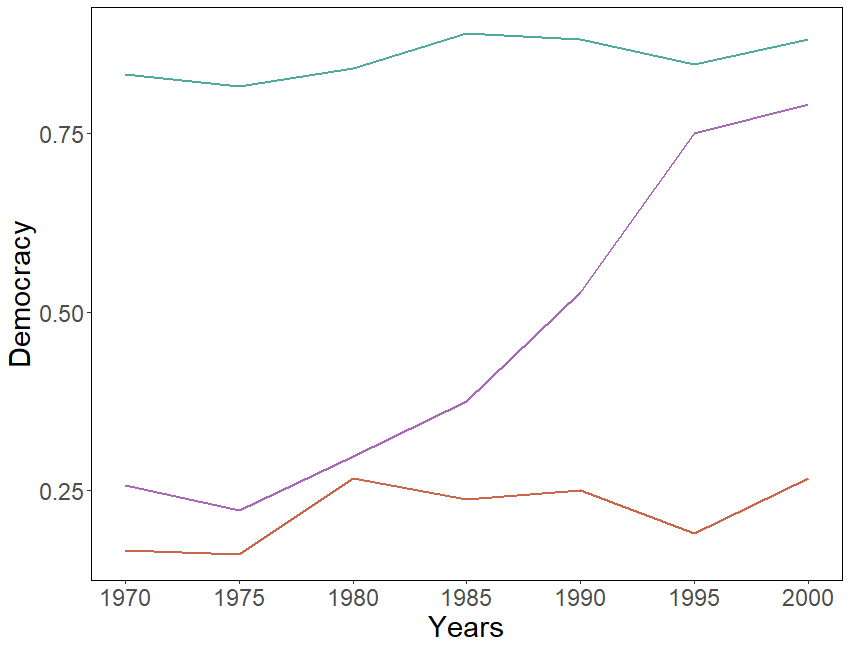}
		\caption{\scriptsize $G=3$ and Top: WGFE, Bottom: GFE}
	\end{subfigure}
	\hfill
	\begin{subfigure}{\textwidth}s\centering
		\includegraphics[width=0.3\textwidth]{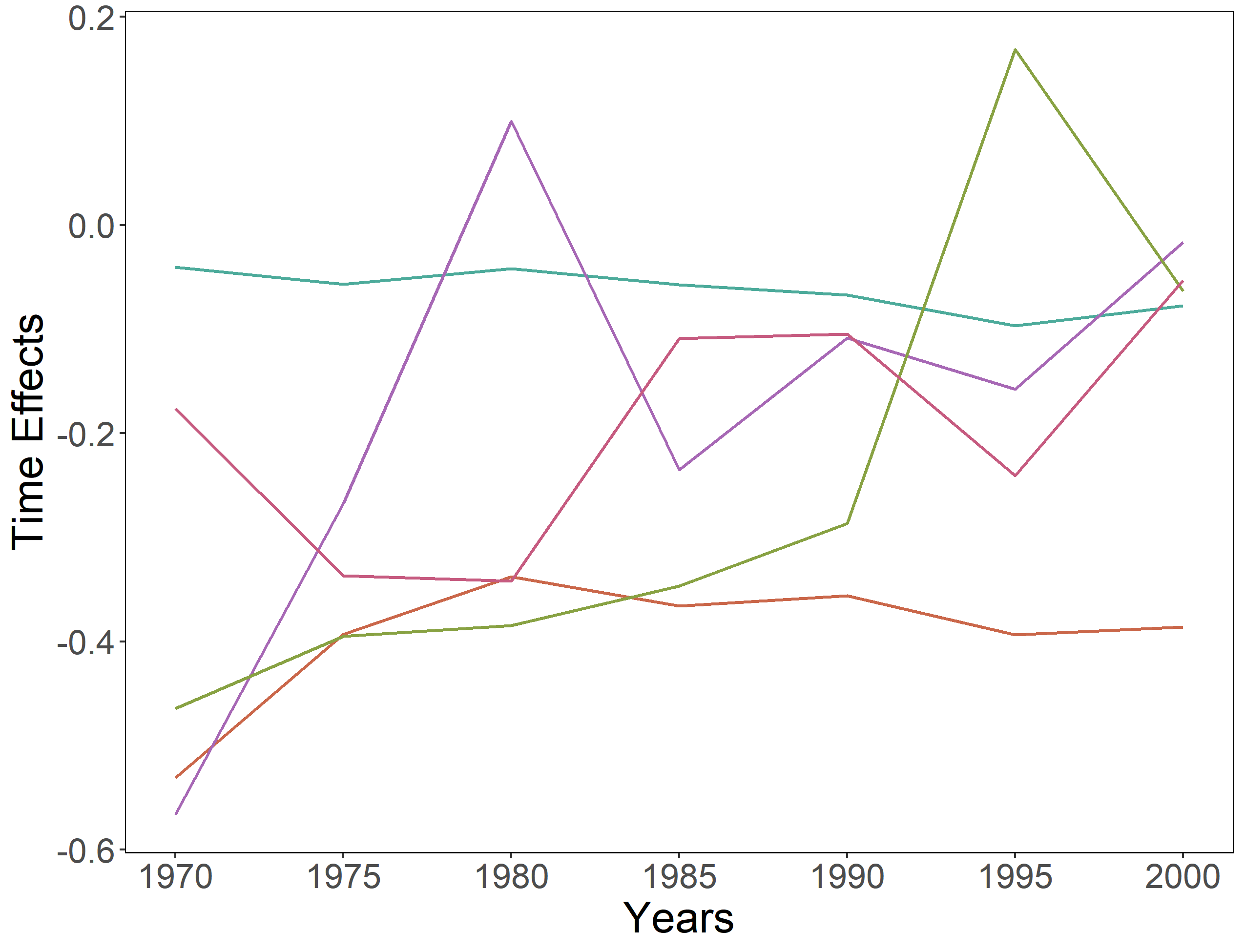}
		\includegraphics[width=0.3\textwidth]{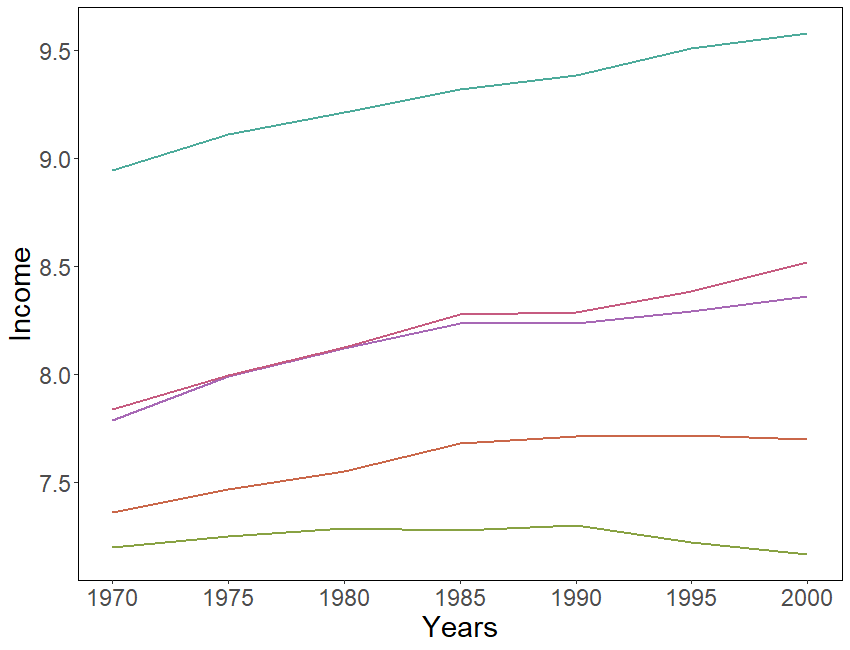}
		\includegraphics[width=0.3\textwidth]{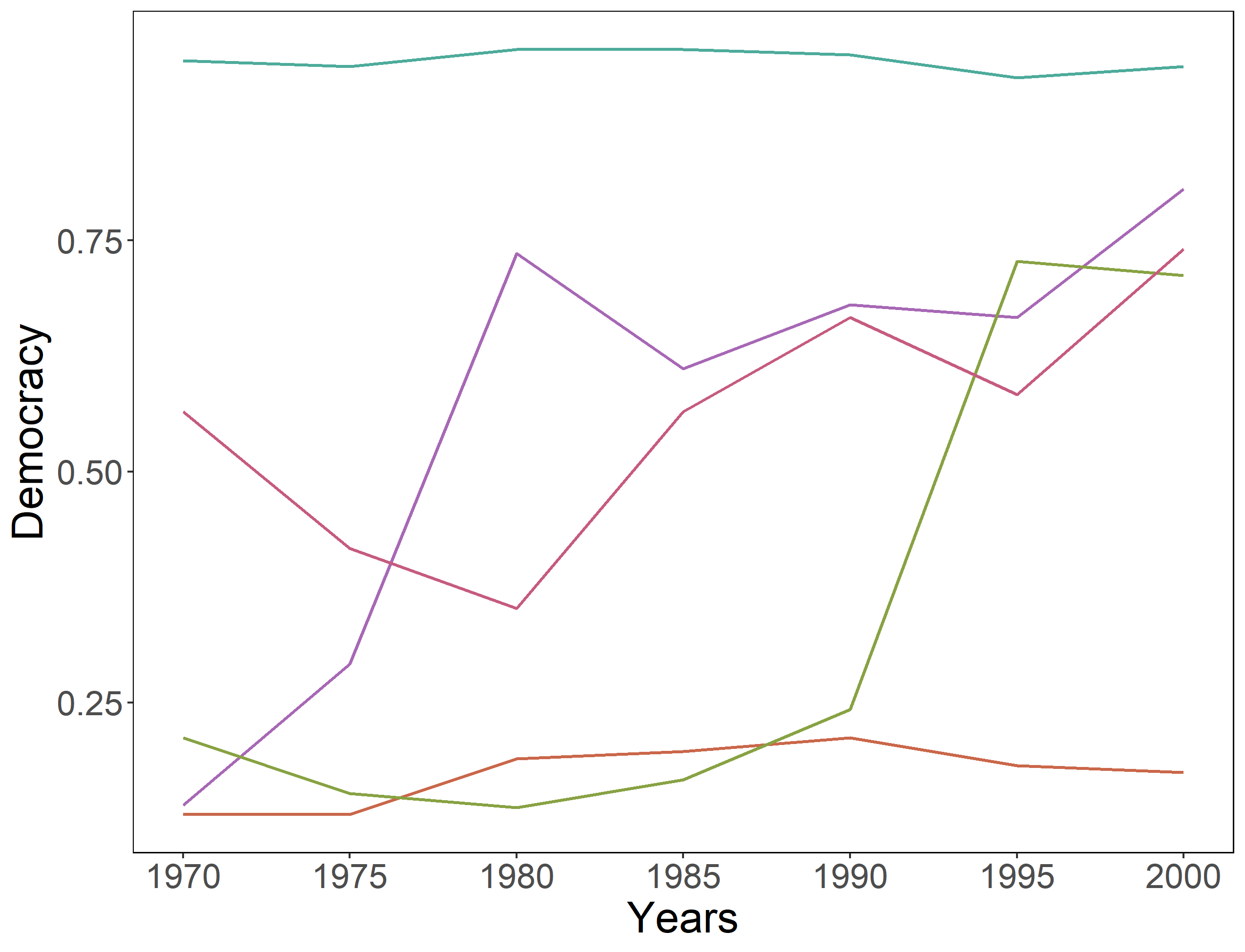}		
		\includegraphics[width=0.3\textwidth]{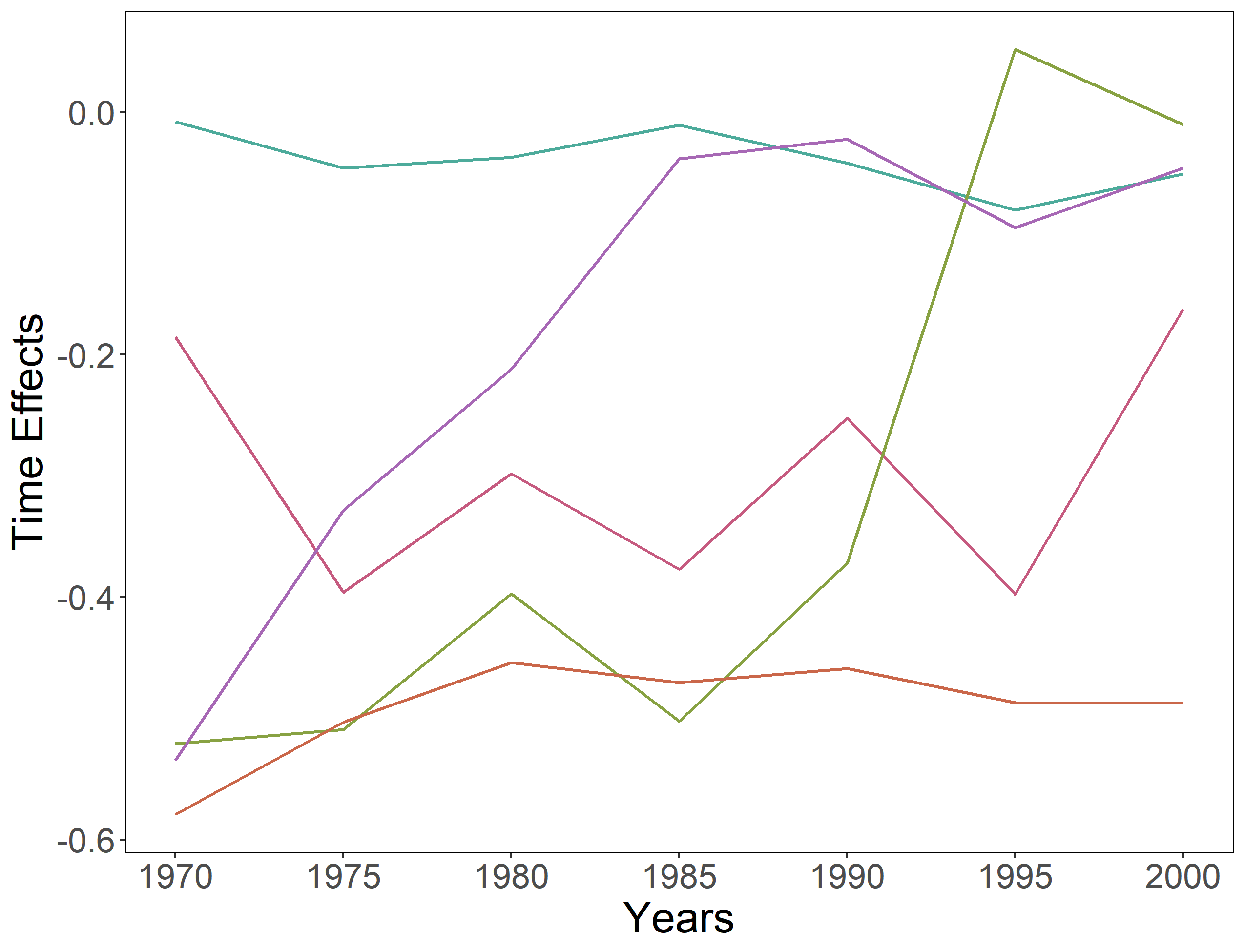}
		\includegraphics[width=0.3\textwidth]{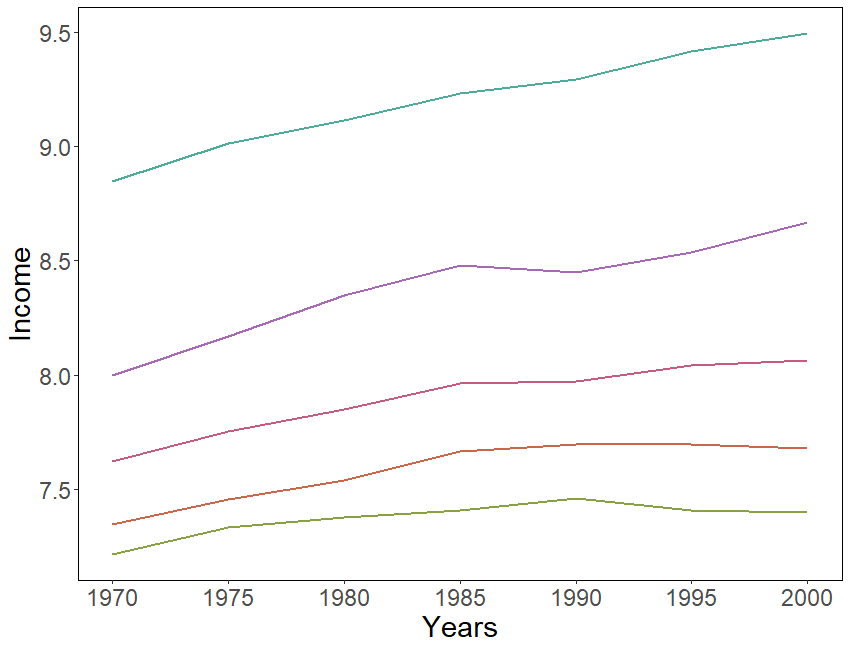}
		\includegraphics[width=0.3\textwidth]{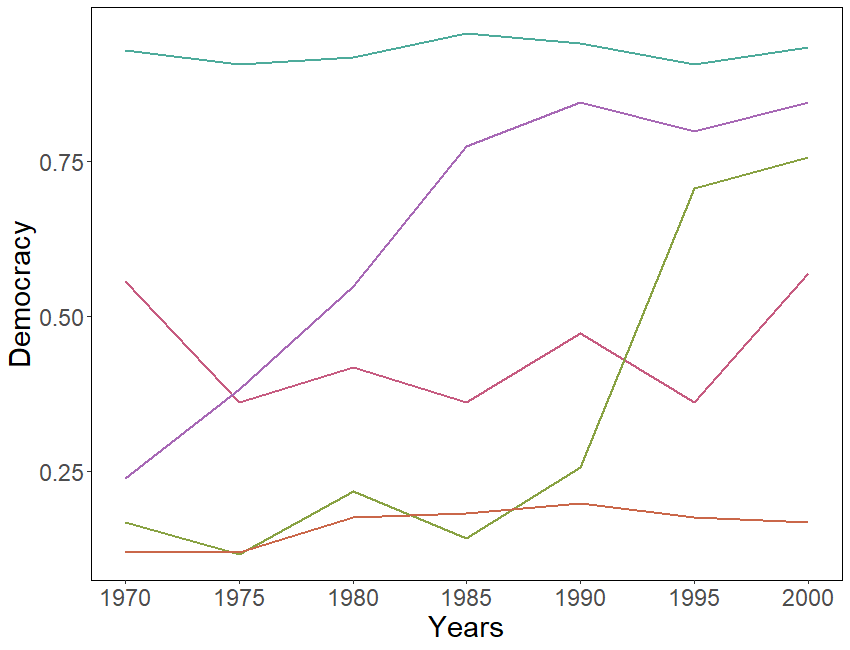}
		\caption{\scriptsize $G=5$ and Top: WGFE, Bottom: GFE}
	\end{subfigure}
\caption{\scriptsize Left: Group-specific time effects $\alpha_{gt}$.
Middle: Within-group average income. Right: Within-group average of democracy index}
\label{fig:time-effects-dem}
\end{figure}

\begin{center}
 \begin{figure} \centering \begin{tabular}{@{}ccccc@{}}
 \toprule  High & Early & Late & Low  \\ \midrule 
 Australia & Argentina & Benin & Algeria \\[2.5mm] 
Austria & Bolivia & Central African Republic & Burundi \\[2.5mm] 
 Belgium & Brazil & Madagascar & Cameroon \\[2.5mm] 
 Canada & Burkina Faso & Malawi & Chad \\[2.5mm]  
Colombia & Chile & Mali & China \\[2.5mm] 
 Costa Rica & Cyprus & Niger & Congo, Dem. Rep. \\[2.5mm] 
 Denmark & Dominican Republic & Panama & Congo, Rep. \\[2.5mm] 
Finland & Ecuador & Phillipines & Cote d'Ivoire \\[2.5mm] 
 France & El Salvador & Romania & Egypt, Arab Rep. \\[2.5mm] 
 Iceland & Ghana & South Africa & Gabon \\[2.5mm] 
 India & Greece & Tanzania & Guinea \\[2.5mm] 
 Ireland & Guatemala & Zambia & Iran \\[2.5mm] 
 Israel & Honduras &  & Jordan \\[2.5mm] 
 Italy & Indonesia &  & Kenya \\[2.5mm] 
 Jamaica & Korea, Rep. &  & Mauritania \\[2.5mm] 
 Japan & Malaysia &  & Morocco \\[2.5mm] 
 Luxembourg & Mexico &  & Rwanda \\[2.5mm] 
Netherlands & Nepal &  & Singapore \\[2.5mm] 
 New Zealand & Nicaragua &  & Syrian Arab Rep. \\[2.5mm] 
 Norway & Nigeria &  & Togo \\[2.5mm] 
 Sri Lanka & Paraguay &  & Tunisia \\[2.5mm] 
 Sweden & Peru &  & Uganda \\[2.5mm] 
 Switzerland & Portugal &  &  \\[2.5mm] 
 Trinidad and Tobago & Sierra Leone &  &  \\[2.5mm] 
 United Kingdom & Spain &  &  \\[2.5mm] 
United State & Taiwan & & \\[2.5mm]
 Venezuela & Thailand &  &  \\[2.5mm] 
\empty &   Turkey &  &  \\[2.5mm] 
\empty &  Uruguay &  &  \\[2.5mm] 
 \bottomrule\\
 \end{tabular} \\

 \caption{List of countries in each group determined by WGFE assignment.}
 \label{fig:wgfe-list}
  \end{figure}  
\end{center}

\begin{center}
 \begin{figure} \centering  \begin{tabular}{@{}ccccc@{}}
 \toprule  High & Early & Late & Low  \\ \midrule 
 Australia & Argentina & Benin & Algeria \\[2.5mm] 
Austria & Bolivia & Burkina Faso & Burundi \\[2.5mm] 
 Belgium & Brazil &  Central African Republic& Cameroon \\[2.5mm] 
 Canada & Ecuador & Chile & Chad \\[2.5mm]  
Colombia & Greece & Ghana & China \\[2.5mm] 
 Costa Rica & Honduras & Madagascar & Congo, Dem. Rep. \\[2.5mm] 
Cyprus &  Korea, Rep. & Malawai & Congo, Rep. \\[2.5mm] 
 Denmark & Nepal & Mali & Cote d'Ivoire \\[2.5mm] 
 Dominican Republic& Peru & Mexico & Egypt, Arab Rep. \\[2.5mm] 
El Salvador & Portugal & Nicaragua & Gabon \\[2.5mm] 
Finland & Spain & Niger & Guinea \\[2.5mm] 
 France & Thailand & Panama & Indonesia \\[2.5mm] 
Guatemala & Uruguay & Phillipines & Iran \\[2.5mm] 
 Iceland &  & Romania & Jordan \\[2.5mm] 
 India &  &  South Africa & Kenya \\[2.5mm] 
 Ireland &  & Taiwan & Mauritania \\[2.5mm] 
 Israel &  & Tanzania & Morocco \\[2.5mm] 
 Italy &  & Zambia & Nigeria \\[2.5mm] 
 Jamaica &  &  & Paraguay \\[2.5mm] 
 Japan & &  & Rwanda\\[2.5mm] 
 Luxembourg &  &  & Sierra Leone \\[2.5mm] 
Malaysia &  & & Singapore \\[2.5mm] 
Netherlands &  &  & Syrian Arab Rep. \\[2.5mm] 
 New Zealand &  &  & Togo \\[2.5mm] 
 Norway &  &  & Tunisia \\[2.5mm] 
 Sri Lanka &  &  & Uganda \\[2.5mm] 
 Sweden &  &  &  \\[2.5mm] 
 Switzerland &  &  &  \\[2.5mm] 
 Trinidad and Tobago &  &  &  \\[2.5mm] 
Turkey &  &  &  \\[2.5mm] 
 United Kingdom & &  &  \\[2.5mm] 
 United States &  &  & \\[2.5mm] 
 Venezuela & &  &  \\[2.5mm] 
\empty &    &  &  \\[2.5mm] 
\empty &   &  &  \\[2.5mm] 
 \bottomrule\\
 \end{tabular} \\

 \caption{List of countries in each group determined by GFE assignment.}
 \label{fig:gfe-list}
  \end{figure}  
\end{center}

\restoregeometry

\begin{center}
\begin{figure}
	\includegraphics[width = \textwidth]{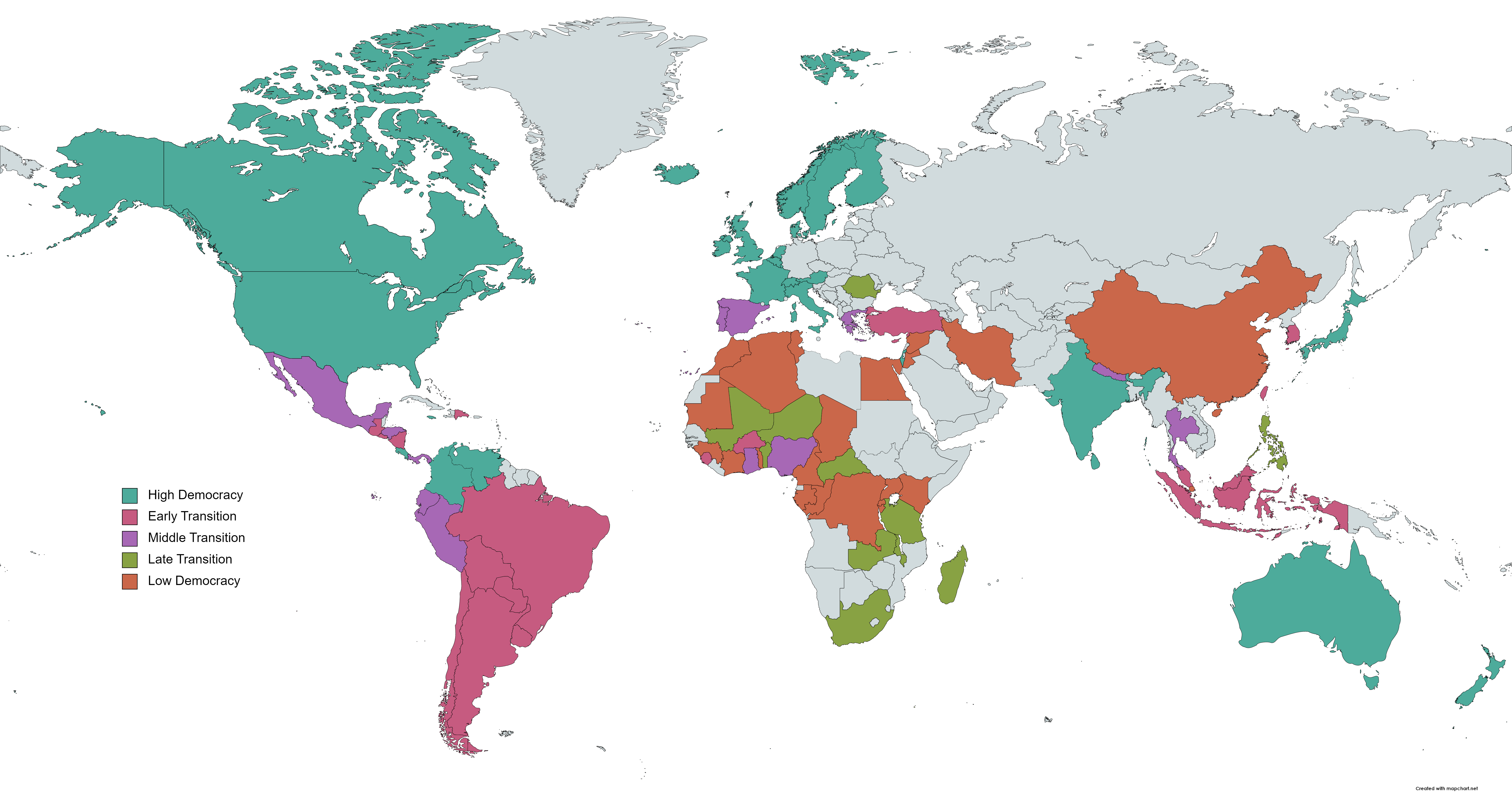}
	\includegraphics[width = \textwidth]{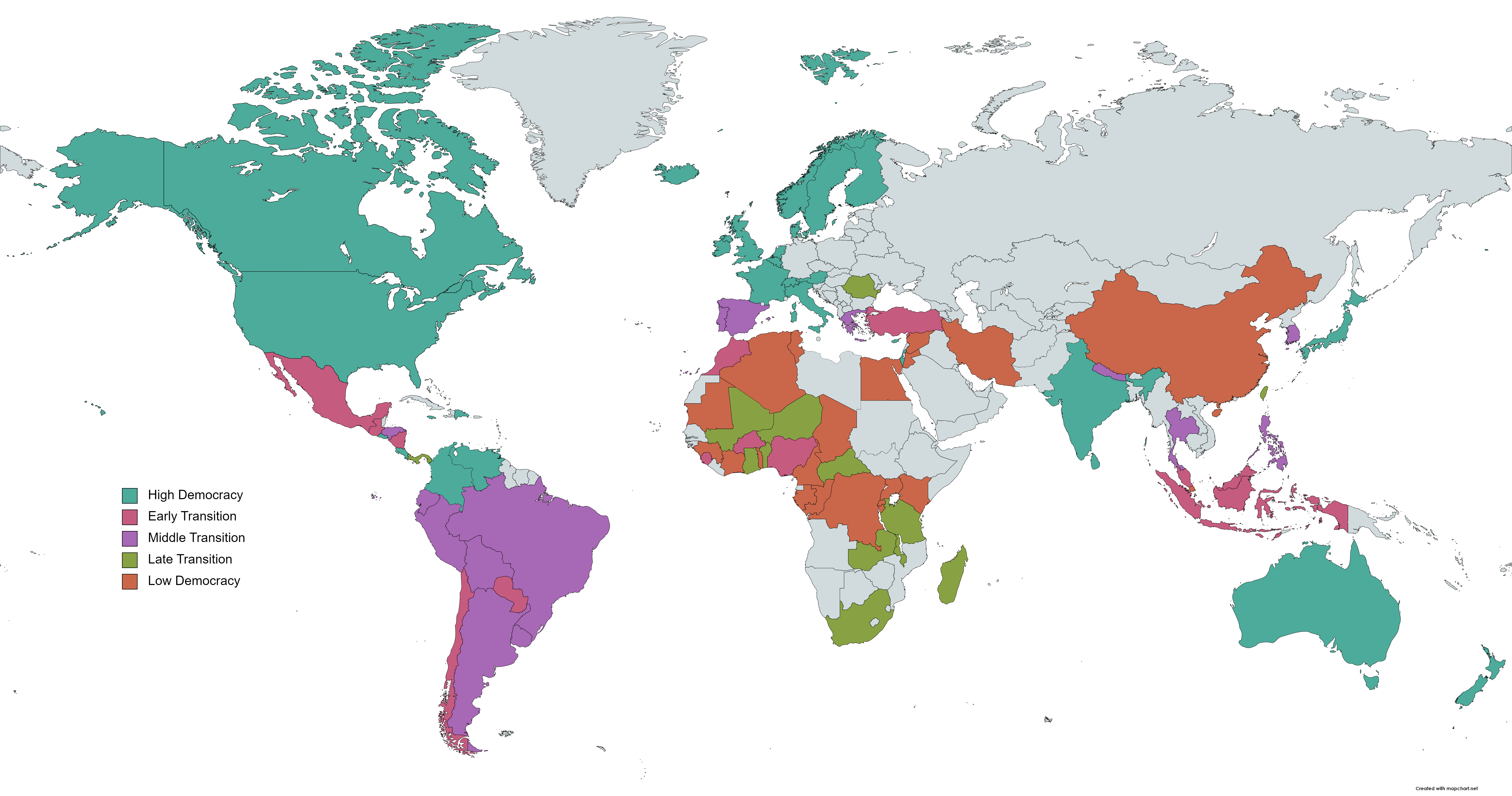}
\caption{$G = 5$ map with Top: WGFE groups, and Bottom: GFE groups}
\label{fig:maps4}
\end{figure}
\end{center}

\section{Variance Estimation}\label{sec:ses}

For all $g=1,\dots,G$, a White estimator for the group $g$ variance is
\begin{equation}\label{wgfe:se}
	\widehat{\sigma}_{g}^2 = \frac{\displaystyle \sumnt \mathds{1}\{\widehat{g}_i = g\} \widehat{u}_{it}^2}{\displaystyle T\sum_{j=1}^N \mathds{1}\{\widehat{g}_i = g\}}
\end{equation}
where $\widehat{u}_{it} = y_{it} - x_{it}'\widehat{\theta} - \widehat{\alpha}_{\widehat{g}_i t}$ are the WGFE residuals and, for any $t = 1,\dots,T$, the White estimator of the variance of the group-specific time effects is
\begin{equation}
	\widehat{\text{Var}}(\widehat{\alpha}_{gt}) = \frac{\displaystyle \sum_{i=1}^N \mathds{1}\{\widehat{g}_i = g\} \widehat{u}_{it}^2}{\displaystyle \left(\sum_{j=1}^N \mathds{1}\{\widehat{g}_i = g\}\right)^2}.
\end{equation}
Theorem \ref{asym:wgfe} suggests an estimator for the variance of $\widehat{\theta}$ as
\begin{equation}
	\widehat{\text{Var}}\left(\sqrt{NT} \widehat{\theta}\right) = \widehat{B}_\theta^{-1}\widehat{V}_\theta \widehat{B}_\theta^{-1}
\end{equation}
where, once again denoting $\overline{x}_{\widehat{g}_i t}$ as the average of $x_{jt}$, $j = 1\dots,N$ in group $\widehat{g}_i$ and denoting $\widehat{\sigma}_{\widehat{g}_i}$ as \eqref{wgfe:se} evaluated at the parameter estimates,
\begin{equation}
	\widehat{B}_\theta^{-1} = \frac{1}{NT}\sumnt \widehat{\sigma}_{\widehat{g}_i}^{-1}(x_{it} - \overline{x}_{\widehat{g}_i t})(x_{it} - \overline{x}_{\widehat{g}_i t})'
\end{equation}
The estimator $\widehat{V}_\theta$ of $V_\theta$ can be based on the estimator of \cite{arellano:1987} clustered at the individual level:
\begin{equation}
	\widehat{V}_\theta = \frac{1}{NT}\sumnt\sum_{s=1}^T \widehat{\sigma}_{\widehat{g}_i}^{-2}(x_{it} - \overline{x}_{\widehat{g}_i t})(x_{is} - \overline{x}_{\widehat{g}_i s})'\widehat{u}_{it}\widehat{u}_{is}
\end{equation}
with a possible bias-correction of degrees of freedom $NT - p$.

\section{Proofs}

\subsection{Proof of Theorem \ref{WC:GFE}}
\begin{proof}
	Begin by rewriting the cluster covariances as 
	\[
		\widehat{\bm{\bm{\Sigma}}}_g = \frac{\widehat{\bm{\Sigma}}^2}{T} I_T = \frac{1}{NT} \sum_{i=1}^N \norm{y_i - x_i\theta - \alpha_{g_i}}^2 I_T.
	\]
	Then, the barycenter covariance reduces to
	\begin{align*}
		\bm{\Omega}_\delta &= \sum_{g = 1}^G \frac{1}{G}\bm{\Omega}_\delta^{1/2} \bigg(\frac{1}{NT} \sum_{i=1}^N \norm{y_i - x_i\theta - \alpha_{g_i}}^2\bigg)^{1/2} I_T\\
							    &= \bm{\Omega}_\delta^{1/2} \bigg(\frac{1}{NT} \sum_{i=1}^N \norm{y_i - x_i\theta - \alpha_{g_i}}^2\bigg)^{1/2} I_T.
	\end{align*}
     Then,
	\begin{equation}
		\bm{\Omega}_\delta = \frac{1}{NT} \sum_{i=1}^N \norm{y_i - x_i\theta - \alpha_{g_i}}^2 I_T
	\end{equation}
	and taking the trace gives 
	\begin{equation}
		\text{tr}[\bm{\Omega}_\delta] = \frac{1}{N} \sum_{i=1}^N \sum_{t=1}^T (y_{it} - x_{it}\theta - \alpha_{g_i t})^2.
	\end{equation}
	which is precisely the criterion function of GFE estimation.
\end{proof}

\subsection{Proof of Theorem \ref{prop:infeasiblecons}}

To show consistency of the infeasible version of WGFE \eqref{est:infeasible}, I show the hypothesis of the standard theorem for consistency of M-estimators holds: 1. uniform weak convergence of the sample criterion to the population criterion, 2. the true values uniquely minimize the population criterion. I prove this in a series of lemmas.

The following is easily seen with a uniform law of large numbers due to compactness of the parameter space and the moment restrictions of Assumption \ref{as:infeas}.
\begin{lemma}\label{lemma:ucongroup}
Suppose Assumption \ref{as:infeas} holds. For all $g = 1,\dots,G$, 
\[
	\sup_{(\theta,\alpha) \in \Theta \times \mathcal{A}^{GT}} \left\vert \widetilde{Q}_g(\theta,\alpha) - Q_g(\theta,\alpha)  \right\vert \to_p 0 
\]
as $N\to \infty$.
\end{lemma}

The following shows uniform convergence of the square roots
\begin{lemma}\label{lemma:uwcongroupsq}
Suppose Assumption \ref{as:infeas} holds. For all $g = 1,\dots,G$, as $N \to \infty$
\[
	\sup_{(\theta,\alpha) \in \Theta \times \mathcal{A}^{GT}} \left\vert \sqrt{\widetilde{Q}_g(\theta,\alpha)} - \sqrt{Q_g(\theta,\alpha)}  \right\vert \to_p 0 .
\]
\end{lemma}
\begin{proof} 
	\begin{align*}
\bigg\vert \sqrt{\widetilde{Q}_g(\theta,\alpha)} - \sqrt{Q_g(\theta,\alpha)} \bigg\vert^2 &\leq \bigg\vert \sqrt{\widetilde{Q}_g(\theta,\alpha)} - \sqrt{Q_g(\theta,\alpha)} \bigg\vert\bigg\vert \sqrt{\widetilde{Q}_g(\theta,\alpha)} + \sqrt{Q_g(\theta,\alpha)} \bigg\vert \\
&= \bigg\vert \widetilde{Q}_g(\theta,\alpha) - Q_g(\theta,\alpha) \bigg\vert 
\end{align*}

Then 
\[
\bigg\vert \sqrt{\widetilde{Q}_g(\theta,\alpha)} - \sqrt{Q_g(\theta,\alpha)} \bigg\vert\leq \sqrt{\bigg\vert \widetilde{Q}_g(\theta,\alpha) - Q_g(\theta,\alpha) \bigg\vert}
\]
so
\begin{align*}
\sup_{(\theta,\alpha) \in \Theta \times \mathcal{A}^{GT}} \bigg\vert \sqrt{\widetilde{Q}_g(\theta,\alpha)} - \sqrt{Q_g(\theta,\alpha)} \bigg\vert &\leq \sup_{(\theta,\alpha) \in \Theta \times \mathcal{A}^{GT}}\sqrt{\bigg\vert \widetilde{Q}_g(\theta,\alpha) - Q_g(\theta,\alpha) \bigg\vert}\\
& = \sqrt{\sup_{(\theta,\alpha) \in \Theta \times \mathcal{A}^{GT}}\bigg\vert \widetilde{Q}_g(\theta,\alpha) - Q_g(\theta,\alpha) \bigg\vert}
\end{align*}

Taking the probability limit as $N \to\infty$ and applying the continuous mapping theorem with Lemma \ref{lemma:ucongroup} and I are done.
\end{proof}

The following shows uniform convergence of the sample criterion function to its population counterpart.

\begin{lemma}\label{lemma:uwconinfeas}
	Suppose Assumption 1 holds. Then, as $N \to \infty$,
\[
	\sup_{(\theta,\alpha) \in \Theta \times \mathcal{A}^{GT}} \left\vert \widetilde{Q}(\theta,\alpha) - Q(\theta,\alpha)  \right\vert \to_p 0.
\]
\end{lemma}
\begin{proof}

First, note that by the weak law of large numbers I have $P_g = P_g^0 + o_p(1)$. Then
\begin{align*}
\sup_{(\theta,\alpha) \in \Theta \times \mathcal{A}^{GT}} \vert \widetilde{Q}(\theta,\alpha) - Q(\theta,\alpha)  \vert &= \sup_{(\theta,\alpha) \in \Theta \times \mathcal{A}^{GT}}\bigg \vert \sum_{g=1}^G P_g^0 \bigg(\sqrt{\widetilde{Q}_g(\theta,\alpha)} - \sqrt{Q_g(\theta,\alpha)}\bigg) + o_p(1) \bigg\vert\\
&\leq  \sum_{g=1}^G P_g^0 \sup_{(\theta,\alpha) \in \Theta \times \mathcal{A}^{GT}} \bigg\vert\sqrt{\widetilde{Q}_g(\theta,\alpha)} - \sqrt{Q_g(\theta,\alpha)}\bigg\vert + o_p(1)
\end{align*}
and the result follows by Lemma \ref{lemma:uwcongroupsq}.

\end{proof}

\begin{lemma}\label{lemma:popmin}
	Suppose Assumption \ref{as:infeas} holds. Then the true values $(\theta^0,\alpha^0)$ uniquely minimize $Q(\theta,\alpha)$.
\end{lemma}
\begin{proof}
Using the data generating process and the identity:
\[
	\sqrt{a} - \sqrt{b} = \frac{a - b}{\sqrt{a} + \sqrt{b}}
\]
 I can write
\begin{align}
	&Q(\theta,\alpha) - Q(\theta^0, \alpha^0)\\ =& \sum_{g=1}^G P_g^0 \left(\sqrt{\Ex{\sum_{t=1}^T \mathds{1}\{g_i^0 = g\}\left(u_{it} + x_{it}'(\theta^0 - \theta) + \left(\alpha_{gt}^0 -  \alpha_{gt}\right)\right)^2}} - \sqrt{\Ex{\sum_{t=1}^T \mathds{1}\{g_i^0 = g\} u_{it}^2}} \right)\notag\\
&= \sum_{g=1}^G P_g^0 \frac{{\Ex{\sum_{t=1}^T \mathds{1}\{g_i^0 = g\}\left(u_{it} + x_{it}'(\theta^0 - \theta) + \left(\alpha_{gt}^0 -  \alpha_{gt}\right)\right)^2}} - {\Ex{\sum_{t=1}^T \mathds{1}\{g_i^0 = g\} u_{it}^2}}}{\sqrt{\Ex{\sum_{t=1}^T \mathds{1}\{g_i^0 = g\}\left(u_{it} + x_{it}'(\theta^0 - \theta) + \left(\alpha_{gt}^0 -  \alpha_{gt}\right)\right)^2}} + \sqrt{\Ex{\sum_{t=1}^T \mathds{1}\{g_i^0 = g\} u_{it}^2}}}\label{lemma:diff}.
\end{align}
For any $g = 1,\dots,G$, there is an upper bound for the denominator independent of parameters. Consider the first term and note that each term in the summand is non negative:
\begin{align}
&\Ex{\sum_{t=1}^T \mathds{1}\{g_i^0 = g\}\left(u_{it} + x_{it}'(\theta^0 - \theta) + \left(\alpha_{gt}^0 -  \alpha_{gt}\right)\right)^2}\notag\\ 
=&\Ex{\sum_{t=1}^T \mathds{1}\{g_i^0 = g\}\left( u_{it}^2 + 2u_{it}\left(x_{it}'(\theta^0 - \theta)  + \left(\alpha_{gt}^0 -  \alpha_{gt}\right)\right) + \left(x_{it}'(\theta^0 - \theta)  + \left(\alpha_{gt}^0 -  \alpha_{gt}\right)\right)^2 \right)}\label{lemma:long}.
\end{align}

Since $\Ex{u_{it} \vert x_{it}, g_i^0} = 0$, I know
\begin{align}
	&\Ex{\mathds{1}\{g_i^0 = g\}u_{it}\left(x_{it}'(\theta^0 - \theta)  + \left(\alpha_{gt}^0 -  \alpha_{gt}\right)\right)}\label{lemma:zero}\\ 
=& \Ex{\mathds{1}\{g_i^0 = g\}\left(x_{it}'(\theta^0 - \theta)  + \left(\alpha_{gt}^0 -  \alpha_{gt}\right)\right) \Ex{u_{it}\vert x_{it}, g_i^0}}\notag\\
=&0\notag.
\end{align}
Therefore, \eqref{lemma:long} can be written as
\begin{align*}
&\Ex{\sum_{t=1}^T \mathds{1}\{g_i^0 = g\} u_{it}^2}  + \Ex{\sum_{t=1}^T \mathds{1}\{g_i^0 = g\}\left(x_{it}'(\theta^0 - \theta)  + \left(\alpha_{gt}^0 -  \alpha_{gt}\right)\right)^2 }\\
\leq& \Ex{\sum_{t=1}^T u_{it}^2}  + \Ex{\sum_{t=1}^T \left(x_{it}'(\theta^0 - \theta)  + \left(\alpha_{gt}^0 -  \alpha_{gt}\right)\right)^2 }\\
\leq& \sum_{t=1}^T \Ex{u_{it}^2}  + \sum_{t=1}^T \Ex{\left(x_{it}'(\theta^0 - \theta)  + \left(\alpha_{gt}^0 -  \alpha_{gt}\right)\right)^2 }\\
\leq& T\sqrt{M} + \sum_{t=1}^T \Ex{\left(x_{it}'(\theta^0 - \theta)  + \left(\alpha_{gt}^0 -  \alpha_{gt}\right)\right)^2 }.
\end{align*}
where the last inequality follows by Assumption \ref{as:infeas}($b$). Set $X_{it} = (1, x_{it})$ as a row vector and $\pi_{gt} = (\alpha_{gt}^0 -  \alpha_{gt}, \theta^0 - \theta)'$. Then continuing with the previous inequality
\begin{align*}
&T\sqrt{M} + \sum_{t=1}^T \Ex{\left(x_{it}'(\theta^0 - \theta)  + \left(\alpha_{gt}^0 -  \alpha_{gt}\right)\right)^2 }\\
=& T\sqrt{M} + \sum_{t=1}^T \Ex{\left(X_{it}'\pi_{gt}\right)^2 }\\
\leq& T\sqrt{M} + \sum_{t=1}^T \norm{X_{it}}^2 \norm{\pi_{gt}}\\
\leq&T\sqrt{M} + K
\end{align*}
due to the Cauchy-Schwarz inequality, compactness of the parameter space and Assumption \ref{as:infeas}($b$), which results in a bound of constant $K>0$. Note that the second term of the denominators for any $g$ is bounded since it was a step in finding the bounds of the first term. Denote the bounding constant as $C^{-1}$ and return to the original difference \eqref{lemma:diff}. Expanding the quadratic cancels the sum of squared error terms and the center term is zero due to \eqref{lemma:zero} so I have
\begin{align*}
	Q(\theta,\alpha) - Q(\theta^0, \alpha^0) &\geq C \sum_{g=1}^G P_g^0\Ex{\sum_{t=1}^T \mathds{1}\{g_i^0 =g\}\left(x_{it}'(\theta^0 - \theta)  + \left(\alpha_{gt}^0 -  \alpha_{gt}\right)\right)^2 }\\
=&C \sum_{g=1}^G (P_g^0)^2\sum_{t=1}^T \Ex{\left(x_{it}'(\theta^0 - \theta)  + \left(\alpha_{gt}^0 -  \alpha_{gt}\right)\right)^2\bigg\vert g_i^0 =g }\\
&\geq C \sum_{g=1}^G (P_g^0)^2\sum_{t=1}^T \Ex{\left(x_{it}'(\theta^0 - \theta)  - \Ex{x_{it}\vert g_i^0 = g}'(\theta^0 - \theta)\right)^2\bigg\vert g_i^0 =g }
\end{align*}
where the last inequality is equality if and only if $\left(\alpha_{gt}^0 -  \alpha_{gt}\right) =  \Ex{x_{it}\vert g_i^0 = g}'(\theta^0 - \theta)$ since it minimizes the mean-squared error. Then, 
\begin{align*}
	&Q(\theta,\alpha) - Q(\theta^0, \alpha^0)\\
 \geq&  C \sum_{g=1}^G (P_g^0)^2\sum_{t=1}^T \Ex{\left(x_{it}'(\theta^0 - \theta)  - \Ex{x_{it}\vert g_i^0 = g}'(\theta^0 - \theta)\right)^2\bigg\vert g_i^0 =g}\\\geq& C\min_{g = 1,\dots,G}\{(P_g^0)^2\} (\theta^0 - \theta)'\left(\sum_{g=1}^G  \Ex{\sum_{t=1}^T\left(x_{it}  - \Ex{x_{it}\vert g_i^0 = g}\right)\left(x_{it}  - \Ex{x_{it}\vert g_i^0 = g}\right)'\bigg\vert g_i^0 =g}\right)(\theta^0 - \theta) \\
\geq&  C\min_{g = 1,\dots,G}\{(P_g^0)^2\}\rho^0 \norm{\theta^0 - \theta}\\
\geq 0
\end{align*}
by Assumption \ref{as:infeas}($e$) where $\rho^0>0$ is the minimum eigenvalue. Equality occurs if and only if $\theta^0 = \theta$ by definiteness of the standard Euclidean norm. Therefore, $\left(\alpha_{gt}^0 -  \alpha_{gt}\right) =  \Ex{x_{it}\vert g_i^0 = g}'(\theta^0 - \theta) = 0$ and I are done.

With Lemma \ref{lemma:uwconinfeas} and Lemma \ref{lemma:popmin}, I invoke consistency of M-estimators and therefore I have $\sqrt{N}$-consistency of $\eqref{est:infeasible}$.

\end{proof}

\subsection{Proof of Asymptotic Normality of the infeasible WGFE estimator}

\begin{proof}[Proof of Asymptotic Normality of $\widetilde{\theta}$]

First I will show asymptotic normality of $\widetilde{\theta}$. Consider 
\begin{gather*}
\sqrt{NT}\left(\widetilde{\theta} - \theta^0\right) = 
\left[ \frac{1}{NT} \sum_{i=1}^N \sum_{t=1}^T \widetilde{\sigma}_{g_i^0}^{-1}(\widetilde{\theta},\widetilde{\alpha}) (x_{it} - \overline{x}_{g_i^0 t})(x_{it} - \overline{x}_{g_i^0 t})' \right]^{-1} \frac{1}{\sqrt{NT}}  \sum_{i=1}^N \sum_{t=1}^T \widetilde{\sigma}_{g_i^0}^{-1}(\widetilde{\theta},\widetilde{\alpha}) (x_{it} - \overline{x}_{g_i^0 t}) u_{it}
\end{gather*}
where the true model was substituted into the estimator definition.

First I will show that the weight terms converge to the true group variances. Indeed,
\begin{align*}
\widetilde{\sigma}_{g}^{2}(\widetilde{\theta},\widetilde{\alpha}) &=  \frac{1}{T\sum_{j=1}^N \mathds{1}\{g_j^0 = g\}} \sumnt \mathds{1}\{g_i^0 = g\}\left(u_{it} + x_{it}'(\theta^0 - \widetilde{\theta}) + (\alpha_{g_i^0 t}^0 - \widetilde{\alpha}_{gt})\right)^2\\
&= \Ex{\mathds{1}\{g_i^0 = g\} u_{it}^2} + o_p(1)\\
&= \sigma_g^2 + o_p(1)
\end{align*}
by a weak law of large numbers and since the parameter estimators are consistent by Proposition \ref{prop:infeasiblecons}.

Therefore, 
\[
\sqrt{NT}\left(\widetilde{\theta} - \theta^0\right) = \left[ \frac{1}{NT} \sum_{i=1}^N \sum_{t=1}^T \sigma_{g_i^0}^{-1} (x_{it} - \overline{x}_{g_i^0 t})(x_{it} - \overline{x}_{g_i^0 t})' \right]^{-1} \frac{1}{\sqrt{NT}}  \sum_{i=1}^N \sum_{t=1}^T \sigma_{g_i^0}^{-1} (x_{it} - \overline{x}_{g_i^0 t}) u_{it} + o_p(1).
\]

Then, the variance conditional on covariates is 
\begin{gather*}
	\text{Var}\left(\sqrt{NT}\left(\widetilde{\theta} - \theta^0\right)\right) = \left[ \frac{1}{NT} \sum_{i=1}^N \sum_{t=1}^T \sigma_{g_i^0}^{-1} (x_{it} - \overline{x}_{g_i^0 t})(x_{it} - \overline{x}_{g_i^0 t})' \right]^{-1} \text{Var}\left( \frac{1}{\sqrt{NT}}  \sum_{i=1}^N \sum_{t=1}^T \sigma_{g_i^0}^{-1} (x_{it} - \overline{x}_{g_i^0 t}) u_{it} \right)\\ \times \left[ \frac{1}{NT} \sum_{i=1}^N \sum_{t=1}^T \sigma_{g_i^0}^{-1} (x_{it} - \overline{x}_{g_i^0 t})(x_{it} - \overline{x}_{g_i^0 t})' \right]^{-1} + o_p(1).
\end{gather*}
The term in the middle can be rewritten as
\begin{gather*}
	\text{Var}\left( \frac{1}{\sqrt{NT}}  \sum_{i=1}^N \sum_{t=1}^T \sigma_{g_i^0}^{-1} (x_{it} - \overline{x}_{g_i^0 t}) u_{it} \right)
= \frac{1}{NT}\Ex{\left( \sum_{i=1}^N \sum_{t=1}^T \sigma_{g_i^0}^{-1} (x_{it} - \overline{x}_{g_i^0 t}) u_{it}\right)\left( \sum_{j=1}^N \sum_{s=1}^T \sigma_{g_j^0}^{-1} (x_{js} - \overline{x}_{g_j^0 s}) u_{js}\right)'}\\
- \frac{1}{NT}\Ex{\sum_{i=1}^N \sum_{t=1}^T \sigma_{g_i^0}^{-1} (x_{it} - \overline{x}_{g_i^0 t}) u_{it}}\Ex{\sum_{j=1}^N \sum_{s=1}^T \sigma_{g_j^0}^{-1} (x_{js} - \overline{x}_{g_j^0 s}) u_{js}}'
\end{gather*}
where the second term is zero since $\Ex{u_{it}\vert x_{it}} = 0$. Then,
\begin{align*}
	&\Ex{\left( \sum_{i=1}^N \sum_{t=1}^T \sigma_{g_i^0}^{-1} (x_{it} - \overline{x}_{g_i^0 t}) u_{it}\right)\left( \sum_{j=1}^N \sum_{s=1}^T \sigma_{g_j^0}^{-1} (x_{js} - \overline{x}_{g_j^0 s}) u_{js}\right)'}\\ &= \Ex{\left( \sum_{i=1}^N \sum_{t=1}^T \sigma_{g_i^0}^{-1} (x_{it} - \overline{x}_{g_i^0 t}) u_{it}\right)\left( \sum_{j=1}^N \sum_{s=1}^T \sigma_{g_j^0}^{-1} (x_{js} - \overline{x}_{g_j^0 s})' u_{js}\right)}\\
&=\Ex{\sum_{i=1}^N \sum_{j=1}^N \sum_{t=1}^T \sum_{s = 1}^T \left[\sigma_{g_i^0}\sigma_{g_j^0}\right]^{-1}  (x_{it} - \overline{x}_{g_i^0 t})(x_{js} - \overline{x}_{g_j^0 s})' u_{it}u_{js}}\\
&= \sum_{i=1}^N \sum_{j=1}^N \sum_{t=1}^T \sum_{s = 1}^T \left[\sigma_{g_i^0}\sigma_{g_j^0}\right]^{-1}  \Ex{(x_{it} - \overline{x}_{g_i^0 t})(x_{js} - \overline{x}_{g_j^0 s})' u_{it}u_{js}}.
\end{align*}
Therefore,
\begin{align*}
	\text{Var}\left(\sqrt{NT}\left(\widetilde{\theta} - \theta^0\right)\right) =& \left[ \frac{1}{NT} \sum_{i=1}^N \sum_{t=1}^T \sigma_{g_i^0}^{-1} (x_{it} - \overline{x}_{g_i^0 t})(x_{it} - \overline{x}_{g_i^0 t})' \right]^{-1}\\ &\times \frac{1}{NT}\sum_{i=1}^N \sum_{j=1}^N \sum_{t=1}^T \sum_{s = 1}^T \left[\sigma_{g_i^0}\sigma_{g_j^0}\right]^{-1}  \Ex{(x_{it} - \overline{x}_{g_i^0 t})(x_{js} - \overline{x}_{g_j^0 s})' u_{it}u_{js}}\\ &\times \left[ \frac{1}{NT} \sum_{i=1}^N \sum_{t=1}^T \sigma_{g_i^0}^{-1} (x_{it} - \overline{x}_{g_i^0 t})(x_{it} - \overline{x}_{g_i^0 t})' \right]^{-1} + o_p(1).
\end{align*}
and, by Assumption \ref{asym:infeas} ($b$),
\[
	\text{Var}\left(\sqrt{NT}\left(\widetilde{\theta} - \theta^0\right)\right) \longrightarrow_p B_\theta^{-1} V_\theta B_{\theta}^{-1}.
\]

Finally, by Assumption \ref{asym:infeas} ($c$),
\begin{align*}
\sqrt{NT}\left(\widetilde{\theta} - \theta^0\right) &= \left[ \frac{1}{NT} \sum_{i=1}^N \sum_{t=1}^T \widetilde{\sigma}_{g_i^0}^{-1}(\widetilde{\theta},\widetilde{\alpha}) (x_{it} - \overline{x}_{g_i^0 t})(x_{it} - \overline{x}_{g_i^0 t})' \right]^{-1} \frac{1}{\sqrt{NT}}  \sum_{i=1}^N \sum_{t=1}^T \widetilde{\sigma}_{g_i^0}^{-1}(\widetilde{\theta},\widetilde{\alpha}) (x_{it} - \overline{x}_{g_i^0 t}) u_{it}\\
&= \left[ \frac{1}{NT} \sum_{i=1}^N \sum_{t=1}^T\sigma_{g_i^0}^{-1}(x_{it} - \overline{x}_{g_i^0 t})(x_{it} - \overline{x}_{g_i^0 t})' \right]^{-1} \frac{1}{\sqrt{NT}}  \sum_{i=1}^N \sum_{t=1}^T \sigma_{g_i^0}^{-1} (x_{it} - \overline{x}_{g_i^0 t}) u_{it} + o_p(1)\\
&\longrightarrow_d B_\theta^{-1} Z 
\end{align*}
where $Z \sim N(0,V_\theta)$. 

\end{proof}

\begin{proof}[Proof of Asymptotic Normality of $\widetilde{\alpha}$]

Consider the following for any $g=1,\dots,G$ and $t = 1,\dots,T$
\[
	\sqrt{N}\left(\widetilde{\alpha}_{gt} - \alpha_{g_i^0}^0 \right) = \frac{\frac{1}{\sqrt{N}}\sum_{i=1}^N \mathds{1}\{g_i^0 = g\} \left(u_{it} - x_{it}'\left(\theta^0 - \widetilde{\theta}\right)\right)}{\frac{1}{N}\sum_{j=1}^N \mathds{1}\{g_j^0 = g\}}
\]
which is obtained by substituting the population model for $y_{it}$.

Then,
\[
	\sqrt{N}\left(\widetilde{\alpha}_{gt} - \alpha_{g_i^0}^0 \right) = \frac{\frac{1}{\sqrt{N}}\sum_{i=1}^N \mathds{1}\{g_i^0 = g\} u_{it} }{\frac{1}{N}\sum_{j=1}^N \mathds{1}\{g_j^0 = g\}}  + \left[\frac{\frac{1}{\sqrt{N}}\sum_{i=1}^N \mathds{1}\{g_i^0 = g\} x_{it}}{\frac{1}{N}\sum_{j=1}^N \mathds{1}\{g_j^0 = g\}} \right]'\left(\theta^0 - \widetilde{\theta}\right).
\]
By Assumption \ref{as:infeas} ($a,c$) I know that 
\[
	\frac{\sum_{i=1}^N \mathds{1}\{g_i^0 = g\} x_{it}}{\sum_{j=1}^N \mathds{1}\{g_j^0 = g\}} = O_p(1/\sqrt{NT}).
\]
Therefore, by the consistency of $\widetilde{\theta}$, and for any $g = 1,\dots,G$ and $t = 1,\dots,T$
\[
	\sqrt{N}\left(\widetilde{\alpha}_{gt} - \alpha_{g_i^0}^0 \right) = \frac{\frac{1}{\sqrt{N}}\sum_{i=1}^N \mathds{1}\{g_i^0 = g\} u_{it} }{\frac{1}{N}\sum_{j=1}^N \mathds{1}\{g_j^0 = g\}}  +o_p(1).
\]
The conditional variance for all $g$ and $t$ is easily calculated as 
\begin{align*}
	\text{Var}\left(\sqrt{N}\left(\widetilde{\alpha}_{gt} - \alpha_{g_i^0}^0 \right)\right) = \left[ \frac{1}{N}\sum_{j=1}^N \mathds{1}\{g_j^0 = g\}\right]^{-2} \times \frac{1}{N} \sum_{i=1}^N\Ex{\sum_{j=1}^N \mathds{1}\{g_i^0 = g\}\mathds{1}\{g_j^0 = g\}u_{it} u_{jt}}
\end{align*}
hence
\[
	\sqrt{N}\left(\widetilde{\alpha}_{gt} - \alpha_{g_i^0}^0 \right) \longrightarrow_p \left[\mathbb{P}(g_i^0 = g) \right]^{-2}Z_\alpha
\]
where $Z_\alpha \sim N(0,v_{gt})$

\end{proof}

\subsection{Proof of Theorem \ref{prop:conswgfe} (Consistency of the WGFE Estimator)}

The argument for consistency of $\widehat{\theta}$ largely follows that of \cite{bm:2015}. First, define an auxiliary criterion function
\begin{align}
	\widetilde{Q}(\theta,\alpha,\gamma) 
&= \sum_{g=1}^G P_g \sqrt{ \frac{1}{T\sum_{j=1}^N P_g^j}\sumnt P_g^i \left(x_{it}'(\theta^0 - \theta) + \alpha_{g_i^0 t}^0 - \alpha_{g_it}\right)^2 
+ \frac{1}{T\sum_{j=1}^N P_g^j}\sum_{i=1}^N \sum_{t=1}^T P_g^i u_{it}^2} \notag\\
&=  \sum_{g=1}^G  \sqrt{P_g \frac{1}{NT}\sum_{i=1}^N \sum_{t=1}^T P_g^i \left(x_{it}'(\theta^0 - \theta) + \alpha_{g_i^0 t}^0 - \alpha_{g_it}\right)^2 
+P_g \frac{1}{NT}\sum_{i=1}^N \sum_{t=1}^T P_g^i u_{it}^2}
\end{align}
where I have taken $P_g(\gamma) = P_g$. Note that the minimizers of this function are the true values so I want to show that this auxiliary criterion function uniformly converges to the WGFE objective function \eqref{wgfe:obj}. To do this I first prove a few lemmas.

The following lemma shows that the variance of the sample barycenter of an arbitrary grouping of error terms will converge to its population counterpart.
\begin{lemma}\label{sig:con}
Suppose Assumption \ref{as:con} holds. Let $\widetilde{\gamma}$ be a random variable with support $\{1,\dots,G\}$ and let $\{\gt_i\}$ be an i.i.d. sample from $\widetilde{\gamma}$. Then
\[
	\sum_{g=1}^G P_{g}(\widetilde{\gamma}) \sqrt{\frac {1}{T\sum_{j=1}^N \mathds{1}\{\gt_i = g\}} \sumnt \mathds{1}\{\gt_i = g\}u_{it}^2} \longrightarrow_p \sum_{g=1}^G \mathbb{P}(\gt_i = g)\sqrt{ \Ex{\mathds{1}\{\gt_i = g\}u_{it}^2}}.
\]
\end{lemma}

\begin{proof}
	First, for $g = 1,\dots,G$, let $\mu_g = \Ex{\mathds{1}\{\gt_i = g\}u_{it}^2}$. Then, by Chebyshev's inequality,
\begin{align*}
	&\mathbb{P}\bigg(\bigg\vert \frac {1}{T\sum_{j=1}^N \mathds{1}\{\gt_i = g\}} \sumnt \mathds{1}\{\gt_i = g\} u_{it}^2 - \mu_g \bigg\vert \geq \varepsilon\bigg)\\ &\leq \frac{1}{\varepsilon^2} \text{Var}\left( \frac {1}{T\sum_{j=1}^N \mathds{1}\{\gt_i = g\}} \sumnt \mathds{1}\{\gt_i = g\} u_{it}^2 \right) \\
& \leq \frac{1}{T^2 \varepsilon^2\left( \sum_{j=1}^N \mathds{1}\{\gt_i = g\}\right)^2 }\text{Var}\left( \sumnt  \mathds{1}\{\gt_i = g\} u_{it}^2 \right) \\
&=\frac{1}{T^2 \varepsilon^2\left( \sum_{j=1}^N \mathds{1}\{\gt_i = g\}\right) }\text{Var}\left( \sum_{t=1}^T \mathds{1}\{\gt_i = g\} u_{it}^2 \right), \>\>\> \text{independent across }i \\
& \leq \frac{2}{T^2 \varepsilon^2\left( \sum_{j=1}^N \mathds{1}\{\gt_i = g\}\right) }\sum_{t=1}^T\sum_{s=1}^T \text{Cov}\left( \mathds{1}\{\gt_i = g\}u_{it}^2,  \mathds{1}\{\gt_i = g\}u_{is}^2\right) \\
& \leq \frac{1}{T^2 \varepsilon^2\left( \sum_{j=1}^N \mathds{1}\{\gt_i = g\}\right)} \sum_{t=1}^T\sum_{s=1}^T \sqrt{\text{Var}(\mathds{1}\{\gt_i = g\}u_{it}^2)}\sqrt{\text{Var}(\mathds{1}\{\gt_i = g\}u_{is}^2)}\\
&\leq \frac{1}{T^2 \varepsilon^2\sum_{j=1}^N \mathds{1}\{\gt_i = g\}}\sum_{t=1}^T\sum_{s=1}^T \sqrt{\Ex{ u_{it}^4 }}\sqrt{\Ex{ u_{is}^4 }} \\
& \leq \frac{M}{\varepsilon^2\sum_{j=1}^N \mathds{1}\{\gt_i = g\} }  
\end{align*}
where the bound on the fourth moments by $M>0$ is due to Assumption \ref{as:con}($b$) and I see that this probability converges to zero as $N$ approaches infinity. Therefore, since $\widetilde{\gamma}$ was an arbitrarily chosen random variable and
\[
P_{g}(\widetilde{\gamma}) = \frac{1}{N}\sum_{i = 1}^N \mathds{1}\{\gt_i = g\} \longrightarrow_p \Ex{\mathds{1}\{\gt_i = g\}} = \mathbb{P}(\gt_i = g)
\]
by the weak law of large numbers, using the continuous mapping theorem finishes the proof.
\end{proof}

The following lemma shows that the sample barycenter variance of any arbitrary grouping of error terms has an asymptotic lower bound as the barycenter variance with the true grouping.

\begin{lemma}\label{con:barycon}
	Suppose Assumption \ref{as:con} holds. Let $\widetilde{\gamma}$ be a random variable with support $\{1,\dots,G\}$ and let $\{\gt_i\}$ be an i.i.d. sample from $\widetilde{\gamma}$. Then there exists $C \geq 0$ such that
\begin{gather*}
	\sum_{g=1}^G P_{g}(\widetilde{\gamma})\sqrt{ \frac {1}{T\sum_{j=1}^N \mathds{1}\{\gt_i = g\} } \sumnt \mathds{1}\{\gt_i = g\} u_{it}^2}
- \sum_{g=1}^G P_{g}(\gamma^0) \sqrt{ \frac {1}{T\sum_{j=1}^N \mathds{1}\{g_i^0 = g\} } \sumnt \mathds{1}\{g_i^0 = g\}u_{it}^2}\\
= C + o_p(1)
\end{gather*}
\end{lemma}
\begin{proof}
	By Lemma \ref{sig:con} and the Slutsky theorem, I have 
\begin{align*}
&\sum_{g=1}^G P_{g}(\widetilde{\gamma})\sqrt{\frac {1}{NT} \sumnt \mathds{1}\{\gt_i = g\}  u_{it}^2} - \sum_{g=1}^G P_{g}(\gamma^0)\sqrt{ \frac {1}{NT} \sumnt \mathds{1}\{g_i^0 = g\}u_{it}^2}\\
	=& \sum_{g=1}^G \mathbb{P}(\gt_i = g)\sqrt{ \Ex{\mathds{1}\{\gt_i = g\}u_{it}^2}} - \sum_{g=1}^G \mathbb{P}(g_i^0 = g)\sqrt{ \Ex{\mathds{1}\{g_i^0 = g\}u_{it}^2}} + o_p(1).
\end{align*}
as $N,T$ approach infinity.

By definition,
\[
	Q(\theta^0,\alpha^0,\gamma^0) = \sum_{g=1}^G \mathbb{P}(g_i^0 = g) \sqrt{\Ex{\mathds{1}\{g_i^0 = g\}u_{it}^2}}
\]
which is a minimized value among $(\theta,\alpha,\gamma)$ so
\[
	 \sum_{g=1}^G \mathbb{P}(\gt_i = g) \sqrt{\Ex{\mathds{1}\{\gt_i = g\}u_{it}^2}} = Q(\theta^0,\alpha^0,\widetilde{\gamma}) \geq Q(\theta^0,\alpha^0,\gamma^0).
\]
Therefore, there exists $C \geq 0 $ such that 
\[
 \sum_{g=1}^G \mathbb{P}(\gt_i = g)\sqrt{ \Ex{\mathds{1}\{\gt_i = g\}u_{it}^2}} - \sum_{g=1}^G \mathbb{P}(g_i^0 = g)\sqrt{ \Ex{\mathds{1}\{g_i^0 = g\}u_{it}^2}} = C + o_p(1).
\]
\end{proof}

\begin{lemma}\label{uwcon:aux}
Suppose Assumption \ref{as:con} holds. Then
\begin{equation}
\sup_{(\theta,\alpha,\gamma) \in \Theta \times \mathcal{A}^{GT} \times \Gamma}
 \left\vert \widehat{Q}(\theta,\alpha,\gamma) - \widetilde{Q}(\theta,\alpha,\gamma) \right\vert \longrightarrow_p  0
\end{equation}
as $N,T$ approach infinity.
\end{lemma}

\begin{proof}

Let $(\theta,\alpha,\gamma) \in \Theta \times \mathcal{A}^{GT} \times \Gamma$.  Then
\begin{gather}
	 \widehat{Q}(\theta,\alpha,\gamma) - \widetilde{Q}(\theta,\alpha,\gamma) \notag \\
= \sum_{g=1}^G \sqrt{P_g \frac{1}{NT}\sum_{i=1}^N \sum_{t=1}^T P_g^i \left(x_{it}'(\theta^0 - \theta) + \alpha_{g_i^0 t}^0 - \alpha_{g_it}\right)^2 
+  P_g^i u_{it}^2 + 2 P_g^i u_{it}\left(x_{it}'(\theta^0 - \theta) + \alpha_{g_i^0 t}^0 - \alpha_{g_it}\right) } \label{uwcon:diff}\\
 - \sum_{g=1}^G  \sqrt{P_g  \frac{1}{NT}\sum_{i=1}^N \sum_{t=1}^T P_g^i \left(x_{it}'(\theta^0 - \theta) + \alpha_{g_i^0 t}^0 - \alpha_{g_it}\right)^2 
+P_g^i u_{it}^2} \notag
\end{gather}
For $a,b >0$ I have the identity
\begin{equation}\label{id:alg}
	\sqrt{a} - \sqrt{b} = \frac{\sqrt{a} - \sqrt{b}}{1} \left(\frac{\sqrt{a} + \sqrt{b}}{\sqrt{a} + \sqrt{b}}\right) = \frac{a - b}{\sqrt{a} + \sqrt{b}}
\end{equation}
which can be applied to \eqref{uwcon:diff} as
\begin{gather}
	\sum_{g=1}^G \varphi_g(\theta,\alpha,\gamma) \left(P_g \frac{2}{NT}\sumnt P_g^i u_{it}\left(x_{it}'(\theta^0 - \theta) + \alpha_{g_i^0 t}^0 - \alpha_{g_it}\right) \right)
\end{gather}
where 
\begin{equation}
	\varphi_g(\theta,\alpha,\gamma) = \left(  \sqrt{\widehat{Q}_g(\theta,\alpha,\gamma)} + \sqrt{\widetilde{Q}_g(\theta,\alpha,\gamma)} \right)^{-1}
\end{equation}
with the $\widehat{Q}_g$ and $\widetilde{Q}_g$ denoting the terms within the weighted sum of square roots of $\widehat{Q}$ and $\widetilde{Q}$, respectively.

To show that the difference between the sample criterion function and auxiliary criterion is $o_p(1)$, I need to show the following for all $g = 1,\dots,G$:
\begin{itemize}
	\item[$i$.] $\frac{2}{NT}\sumnt P_g^i u_{it}\left(x_{it}'(\theta^0 - \theta) + \alpha_{g_i^0 t}^0 - \alpha_{g_it}\right)  = o_p(1)$; and\\
	\item[$ii$.] $\varphi_g(\theta,\alpha,\gamma) = O_p(1)$.
\end{itemize}

\noindent{\it Proof of i}. Expanding the sum reveals
\[
	\left(\frac{2}{NT}\sumnt P_g^i u_{it}x_{it}\right)'(\theta^0 - \theta) + \frac{2}{NT}\sumnt P_g^i \alpha_{g_i^0 t}^0u_{it} -  \frac{2}{NT}\sumnt P_g^i \alpha_{g_it} u_{it}
\]
which I will show each term is $o_p(1)$.

Define $\omega_g^i = P_g^i/N \leq 1/N$ so $\sum_i \omega_g^i = P_g$. Then for the first term:
\begin{align}
	\Ex{\norm{\sum_{i=1}^N \omega_g^i \frac{1}{T} \sum_{t=1}^T  u_{it}x_{it}}^2} &\leq \Ex{\left( \sum_{i=1}^N \omega_g^i  \norm{ \frac{1}{T} \sum_{t=1}^Tu_{it}x_{it}}\right)^2} \\
&\leq \Ex{\left( \sum_{i=1}^N \frac{1}{N}  \norm{ \frac{1}{T} \sum_{t=1}^Tu_{it}x_{it}}\right)^2} \\
&\leq \Ex{ \frac{1}{N} \sum_{i=1}^N \norm{\frac{1}{T} \sum_{t=1}^Tu_{it}x_{it}}^2} \\
& = \Ex{ \frac{1}{NT^2} \sumnt \sum_{s=1}^T  u_{it}u_{is} x_{it}'x_{is} }\\
& = \Ex{\frac{1}{NT^2} \sumnt \sum_{s=1}^T  u_{it}u_{is} x_{it}'x_{is}}\\
& \leq \frac{M}{T}
\end{align}
where the inequalities are due to the triangle inequality and then Jensen's inequality and Assumption \ref{as:con} ($c$). Hence, the first term is $o_p(1)$ due to this inequality and Assumption \ref{as:con} ($a$) that the parameter space is a compact subset of $\R^p$ so $\norm{\theta^0 - \theta}$ is bounded for any $\theta$.

For the last two terms, it is enough to show the third is $o_p(1)$. For every $g = 1,\dots,G$,
\begin{equation}
	\left( \frac{1}{NT} \sumnt P_g^i \alpha_{g_i t} u_{it} \right)^2 
= \left( \frac{1}{T} \sum_{t=1}^T \alpha_{gt} \left(\frac{1}{N} \sum_{i=1}^N P_g^i u_{it} \right) \right)^2
\leq \left( \frac{1}{T} \sum_{t=1}^T \alpha_{gt} ^2 \right) \left( \frac{1}{T} \sum_{t=1}^T  \left(\frac{1}{N} \sum_{i=1}^N P_g^i u_{it}  \right)^2 \right)
\end{equation}
where the left term is uniformly bounded because of compactness of the parameter space Assumption \ref{as:con} ($a$). Then,
\begin{align*}
	\frac{1}{T} \sum_{t=1}^T \left( \frac{1}{N} \sum_{i=1}^N P_g^i u_{it} \right)^2 
&= \frac{1}{TN^2} \sum_{i=1}^N \sum_{j=1}^N P_g^i P_g^j \sum_{t=1}^T u_{it} u_{jt}
\leq \frac{1}{N^2} \sum_{i=1}^N \sum_{j=1}^N \left\vert \frac{1}{T}\sum_{t=1}^T u_{it} u_{jt} \right\vert\\
&\leq  \frac{1}{N^2} \sum_{i=1}^N \sum_{j=1}^N \left\vert \frac{1}{T}\sum_{t=1}^T \Ex{u_{it} u_{jt}} \right\vert
+ \frac{1}{N^2} \sum_{i=1}^N \sum_{j=1}^N \left\vert \frac{1}{T}\sum_{t=1}^T  u_{it} u_{jt} - \Ex{u_{it} u_{jt}} \right\vert.
\end{align*}
By Assumption \ref{as:con} ($d$), ${\textstyle \frac{1}{N^2} \sum_{i=1}^N \sum_{j=1}^N \left\vert \frac{1}{T}\sum_{t=1}^T \Ex{u_{it} u_{jt}} \right\vert \leq M/N}$. Also, by the Cauchy-Schwarz inequality,
\[
	\left( \frac{1}{N^2} \sum_{i=1}^N \sum_{j=1}^N \left\vert \frac{1}{T}\sum_{t=1}^T  u_{it} u_{jt} - \Ex{u_{it} u_{jt}} \right\vert \right)^2
 \leq  \frac{1}{N^2} \sum_{i=1}^N \sum_{j=1}^N \left( \frac{1}{T}\sum_{t=1}^T  u_{it} u_{jt} - \Ex{u_{it} u_{jt}} \right)^2
\]
which is bounded in expectation by Assumption \ref{as:con} ($e$). 

Therefore, point one is $o_p(1)$ and it remains to show that point two is $O_p(1)$. However, by compactness of the parameter space and Lemma \ref{sig:con} it is easily seen that it is $O_p(1)$.

Therefore, since $g$ was arbitrarily chosen, I have that 
\[
	 \widehat{Q}(\theta,\alpha,\gamma) - \widetilde{Q}(\theta,\alpha,\gamma) = o_p(1).
\]
\end{proof}

\textnormal{The following lemma shows that the true values are unique minimizers of the auxiliary criterion function in the probability limit.}

\begin{lemma}\label{lb:aux}
Suppose Assumption \ref{as:con} holds. There exists $K > 0$ such that for all $(\theta,\alpha,\gamma) \in \Theta \times \mathcal{A}^{GT} \times \Gamma_G$, 
\[
	\widetilde{Q}(\theta,\alpha,\gamma)  - \widetilde{Q}(\theta^0,\alpha^0,\gamma^0) \geq  K \norm{\theta^0 - \theta}^2 + o_p(1)
\]
\end{lemma}

\begin{proof}

For shorthand, take $P_g^0 = P_g(\gamma^0)$. I know $P_g \neq 0$ for some $g= 1,\dots,G$, so denote
\begin{align*}
	\left[\widetilde{\varphi}_g(\theta,\alpha,\gamma)\right]^{-1} &= \sqrt{P_g \frac{1}{TN}\sum_{i=1}^N \sum_{t=1}^T \mathds{1}\{g_i = g\} \left(x_{it}'(\theta^0 - \theta) +\{g_i = g\} \alpha_{g_i^0 t}^0 - \alpha_{g_it}\right)^2 + P_g\frac{1}{TN}\sum_{i=1}^N \sum_{t=1}^T \mathds{1}\{g_i = g\}u_{it}^2}\\
 &+ \sqrt{P_{g}\frac{1}{TN}\sum_{i=1}^N \sum_{t=1}^T \mathds{1}\{g_i = g\} u_{it}^2}.
\end{align*}

Adding by zero to
\begin{align*}
	&\widetilde{Q}(\theta,\alpha,\gamma) -  \widetilde{Q}(\theta^0,\alpha^0,\gamma^0)\\
&= \sum_{g=1}^G \sqrt{P_g \frac{1}{NT}\sumnt \mathds{1}\{g_i = g\} \left( x_{it}'(\theta^0 - \theta) + \alpha_{g_i^0 t}^0 - \alpha_{g_i t}\right)^2 + P_g\frac{1}{NT} \sumnt \mathds{1}\{g_i = g\} u_{it}^2 }\\
&\underbrace{- \sum_{g=1}^G  \sqrt{P_g\frac{1}{NT} \sumnt \mathds{1}\{g_i = g\} u_{it}^2} + \sum_{g=1}^G  \sqrt{P_g\frac{1}{NT} \sumnt \mathds{1}\{g_i = g\} u_{it}^2}}_{=0} - \sum_{g=1}^G  \sqrt{P_{g}^0\frac{1}{NT} \sumnt \mathds{1}\{g_i^0 = g\} u_{it}^2}
\end{align*}
then applying the identity (\ref{id:alg}) and Lemma \ref{con:barycon} gives
\begin{align*}
	&\widetilde{Q}(\theta,\alpha,\gamma) -  \widetilde{Q}(\theta^0,\alpha^0,\gamma^0)\\
&= \sum_{g=1}^G \sqrt{P_g \frac{1}{NT}\sumnt P_g^i \left( x_{it}'(\theta^0 - \theta) + \alpha_{g_i^0 t}^0 - \alpha_{g_i t}\right)^2 + P_g\frac{1}{NT} \sumnt P_g^i u_{it}^2 }
- \sqrt{P_g\frac{1}{NT} \sumnt P_g^i u_{it}^2}\\
&+ C + o_p(1)\\
&\geq  \sum_{g=1}^G \widetilde{\varphi}_g(\theta,\alpha,\gamma) P_g \left[\frac{1}{NT}\sumnt P_g^i \left( x_{it}'(\theta^0 - \theta) + \alpha_{g_i^0 t}^0 - \alpha_{g_i t}\right)^2\right] + o_p(1)
\end{align*}
where this inequality is over those $g = 1,\dots,G$ such that $\mathbb{P}(g_i = g) > 0$, in other words those groups that are non empty and also
\begin{align*}
	\left[\widetilde{\varphi}_g(\theta,\alpha,\gamma)\right]^{-1} &= \sqrt{P_g \frac{1}{NT}\sum_{i=1}^N \sum_{t=1}^T P_g^i \left(x_{it}'(\theta^0 - \theta) + \alpha_{g_i^0 t}^0 - \alpha_{g_it}\right)^2 + P_g\frac{1}{NT}\sum_{i=1}^N \sum_{t=1}^T P_g^i u_{it}^2}\\
 &+ \sqrt{P_{g}\frac{1}{NT}\sum_{i=1}^N \sum_{t=1}^T P_{g}^i u_{it}^2}\\
&\leq \sqrt{\frac{1}{NT}\sum_{i=1}^N \sum_{t=1}^T \left(x_{it}'(\theta^0 - \theta) + \alpha_{g_i^0 t}^0 - \alpha_{g_it}\right)^2} + 2\sqrt{\frac{1}{TN}\sum_{i=1}^N \sum_{t=1}^T u_{it}^2}\\
&\leq  \sqrt{  \frac{1}{NT} \sumnt \norm{x_{it}}^2\norm{\theta^0 - \theta}^2} + \sqrt{\frac{1}{NT} \sumnt \norm{x_{it}}\norm{\theta^0 - \theta} \left\vert \alpha_{g_i^0 t}^0 - \alpha_{g_i t} \right\vert }\\
&+ \sqrt{\frac{1}{NT} \sumnt \left(\alpha_{g_i^0 t}^0 - \alpha_{g_i t}\right)^2} + 2\sqrt{\frac{1}{TN}\sum_{i=1}^N \sum_{t=1}^T u_{it}^2}
\end{align*}
where all the terms can be bounded due to Assumption \ref{as:con} ($a$,$b$) and Assumption 1 ($c$) so there exists a constant $J$ for any $g$ such that $\widetilde{\varphi}_g(\theta,\alpha,\gamma) > J$.

Then, since the (within-group) mean minimizes the sum of these squared deviations and I have assumed some groups are non empty, the first term is bounded below as
\begin{align*}
&\widetilde{Q}(\theta,\alpha,\gamma) -  \widetilde{Q}(\theta^0,\alpha^0,\gamma^0)\\ 
 \geq& J \min_{\substack{P_g >0\\ g= 1,\dots,G}}\{P_g\}\sum_{g=1}^G  \frac{1}{NT}\sumnt P_g^i \left( x_{it}'(\theta^0 - \theta) + \alpha_{g_i^0 t}^0 - \alpha_{g_i t} \right)^2 + o_p(1)\\
=& J \min_{\substack{P_g >0\\ g= 1,\dots,G}}\{P_g\}\frac{1}{NT}\sumnt\left( x_{it}'(\theta^0 - \theta) + \alpha_{g_i^0 t}^0 - \alpha_{g_i t} \right)^2 + o_p(1)\\
\geq&J \min_{\substack{P_g >0\\ g= 1,\dots,G}}\{P_g\}\frac{1}{NT}\sumnt  \left( x_{it}'(\theta^0 - \theta) - \overline{x}_{g_i\wedge g_i^0  t}'(\theta^0 - \theta) \right)^2 + o_p(1)\\
\geq& J  \min_{\substack{P_g >0\\ g= 1,\dots,G}}\{P_g\}\min_{\gamma \in \Gamma}(\theta^0 - \theta)'\left(  \frac{1}{NT}\sumnt  \left( x_{it} -  \overline{x}_{g_i\wedge g_i^0 t}\right)\left( x_{it} -  \overline{x}_{g_i\wedge g_i^0 t}\right)' \right)(\theta^0 - \theta) + o_p(1)\\
\geq& J \min_{\gamma\in \Gamma}\{\widehat{\rho}(\gamma)\}\norm{\theta^0 - \theta}^2 + o_p(1) \\
 =&K \norm{\theta^0 - \theta}^2 + o_p(1)
\end{align*}
where the last equality is due to the convergence of probability of $\min_\gamma\{ \widehat{\rho}(\gamma)\}$ to a nonzero constant by Assumption  \ref{as:con}($f$).
\end{proof}

\begin{proof}[Proof of consistency of the WGFE estimator]
I now show consistency of the WGFE estimator $\widehat{\theta}$ of $\theta^0$. By Lemma \ref{uwcon:aux} and the definition of the WC estimator, I have
\begin{equation}
	\widetilde{Q}(\widehat{\theta},\widehat{\alpha},\widehat{\gamma}) = \widehat{Q}(\widehat{\theta},\widehat{\alpha},\widehat{\gamma}) + o_p(1)
\leq \widehat{Q}(\theta^0,\alpha^0,\gamma^0) + o_p(1) = \widetilde{Q}(\theta^0,\alpha^0,\gamma^0) + o_p(1).
\end{equation}
Hence,
\[\widetilde{Q}(\widehat{\theta},\widehat{\alpha},\widehat{\gamma}) - \widetilde{Q}(\theta^0,\alpha^0,\gamma^0) \leq o_p(1).\]

Then, by Lemma \ref{lb:aux},
\[
o_p(1)\leq C \norm{\theta^0 - \widehat{\theta}}^2 + o_p(1) \leq\widetilde{Q}(\widehat{\theta},\widehat{\alpha},\widehat{\gamma}) - \widetilde{Q}(\theta^0,\alpha^0,\gamma^0) \leq o_p(1).
\]
Hence,
\[
	\norm{\widehat{\theta} - \theta^0}^2 = o_p(1).
\]

Next, I show the second part of the proposition:
\[
	\frac{1}{NT} \sumnt \left( \widehat{\alpha}_{\widehat{g}_i t} - \alpha_{g_i^0 t}^0 \right)^2 \longrightarrow_p 0
\]
as $N,T \to \infty$.

Start with 
\begin{align*}
&\left\vert \widetilde{Q}(\widehat{\theta},\widehat{\alpha},\widehat{\gamma}) - \widetilde{Q}(\theta^0,\widehat{\alpha},\widehat{\gamma})\right\vert\\
&= \bigg\vert\sum_{g=1}^G \sqrt{P_g \frac{1}{NT} \sumnt P_g^i \left(\left( x_{it}'(\theta^0 - \widehat{\theta}) + \alpha_{g_i^0 t}^0 - \widehat{\alpha}_{\widehat{g}_i t}\right)^2 + u_{it}^2 \right)}\\
&- \sum_{g=1}^G \sqrt{P_g \frac{1}{NT} \sumnt P_g^i \left(\left(\alpha_{g_i^0 t}^0 - \widehat{\alpha}_{\widehat{g}_i t}\right)^2 + u_{it}^2\right)} \bigg\vert\\
&= \left\vert \sum_{g=1}^G \varphi_g(\widehat{\theta},\widehat{\alpha},\widehat{\gamma})  P_g \frac{1}{NT}\sumnt P_g^i \left(\left( x_{it}'(\theta^0 - \widehat{\theta}) + \alpha_{g_i^0 t}^0 - \widehat{\alpha}_{\widehat{g}_i t}\right)^2 - \left(\alpha_{g_i^0 t}^0 - \widehat{\alpha}_{\widehat{g}_i t}\right)^2 \right) \right\vert
\end{align*}
where 
\begin{align*}
[\varphi_g(\widehat{\theta},\widehat{\alpha},\widehat{\gamma})]^{-1}
=& \sqrt{P_g \frac{1}{NT}\sumnt P_g^i\left( x_{it}'(\theta^0 - \widehat{\theta}) + \alpha_{g_i^0 t}^0 - \widehat{\alpha}_{\widehat{g}_i t}\right)^2 + P_g \frac{1}{NT}\sumnt P_g^i u_{it}^2}\\
+& \sqrt{P_{g} \frac{1}{NT}\sumnt P_{g}^i \left( (\alpha_{g_i^0 t}^0 - \widehat{\alpha}_{\widehat{g}_i t})^2 + u_{it}^2 \right)}
\end{align*}
is the sum of the square roots similar to what is found on the previous pages that is bounded above by a constant. Then
\begin{align*}
&\left\vert \widetilde{Q}(\widehat{\theta},\widehat{\alpha},\widehat{\gamma}) - \widetilde{Q}(\theta^0,\widehat{\alpha},\widehat{\gamma})\right\vert\\
&\leq \sum_{g=1}^G \varphi_g(\widehat{\theta},\widehat{\alpha},\widehat{\gamma})  \frac{1}{NT}\sumnt  \left\vert\left( x_{it}'(\theta^0 - \widehat{\theta}\right)\left(x_{it}'(\theta^0 - \widehat{\theta}) + 2\left(\alpha_{g_i^0 t}^0 - \widehat{\alpha}_{\widehat{g}_i t}\right)\right )\right\vert\\
&\leq C \times \norm{\theta^0 - \widehat{\theta}}^2 \times \frac{1}{NT} \sumnt \norm{x_{it}}^2 + C\times 4\sup_{\alpha_t \in \mathcal{A}} \vert \alpha_t\vert\times\norm{\theta^0 - \widehat{\theta}}\times \frac{1}{NT} \sumnt \norm{x_{it}}
\end{align*}
which is $o_p(1)$ because of the consistency of $\widehat{\theta}$.

Therefore,
\[
o_p(1) = \widetilde{Q}(\widehat{\theta},\widehat{\alpha},\widehat{\gamma}) - \widetilde{Q}(\theta^0,\widehat{\alpha},\widehat{\gamma})
\leq \widetilde{Q}(\theta^0,\alpha^0,\gamma^0) - \widetilde{Q}(\theta^0,\widehat{\alpha},\widehat{\gamma}) + o_p(1).
\]
Then
\begin{align*}
	&\widetilde{Q}(\theta^0,\alpha^0,\gamma^0) - \widetilde{Q}(\theta^0,\widehat{\alpha},\widehat{\gamma}) \\
=& \sum_{g=1}^G \sqrt{P_{g^0} \frac{1}{NT} \sumnt P_{g^0}^i u_{it}^2} - \sqrt{P_g \frac{1}{NT} \sumnt P_g^i \left( \left(\alpha_{g_i^0 t}^0 - \widehat{\alpha}_{\widehat{g}_i t} \right)^2 + u_{it}^2 \right)}\\
=& \sum_{g=1}^G \sqrt{P_{g^0} \frac{1}{NT} \sumnt P_{g^0}^i u_{it}^2} \underbrace{- \sqrt{P_{g} \frac{1}{NT} \sumnt P_{g}^i u_{it}^2} + \sqrt{P_{g} \frac{1}{NT} \sumnt P_{g}^i u_{it}^2}}_{=0}\\
 -& \sqrt{P_g \frac{1}{NT} \sumnt P_g^i \left( \left(\alpha_{g_i^0 t}^0 - \widehat{\alpha}_{\widehat{g}_i t} \right)^2 + u_{it}^2 \right)}\\
\end{align*}
Then, by Lemma \ref{con:barycon},
\begin{align*}
-& \left( \sum_{g=1}^G \sqrt{P_{g} \frac{1}{NT} \sumnt P_{g}^i u_{it}^2}-\sum_{g=1}^G \sqrt{P_{g^0} \frac{1}{NT} \sumnt P_{g^0}^i u_{it}^2} \right) - \sum_{g=1}^G P_g \varphi_g \frac{1}{NT} \sumnt P_g^i \left(\alpha_{g_i^0 t}^0 - \widehat{\alpha}_{\widehat{g}_i t} \right)^2 \\
\leq& -C + o_p(1) - \sum_{g=1}^G P_g \frac{1}{NT} \sumnt P_g^i \left(\alpha_{g_i^0 t}^0 - \widehat{\alpha}_{\widehat{g}_i t} \right)^2 \\
\leq& -\sum_{g=1}^G P_g \frac{1}{NT} \sumnt P_g^i \left(\alpha_{g_i^0 t}^0 - \widehat{\alpha}_{\widehat{g}_i t} \right)^2 + o_p(1) \\
\leq& o_p(1)
\end{align*}

Then,
\[
	o_p(1) = \sum_{g=1}^G P_g \frac{1}{NT} \sumnt P_g^i \left(\alpha_{g_i^0 t}^0 - \widehat{\alpha}_{\widehat{g}_i t} \right)^2  
\]
which implies for each $g=1,\dots,G$ (since these are all non negative terms),
\begin{equation}
	o_p(1) = \frac{1}{NT} \sumnt P_g^i \left(\alpha_{g_i^0 t}^0 - \widehat{\alpha}_{\widehat{g}_i t} \right)^2 
\end{equation}
so that
\begin{align}
o_p(1) &=  \sum_{g=1}^G \frac{1}{NT} \sumnt P_g^i \left(\alpha_{g_i^0 t}^0 - \widehat{\alpha}_{\widehat{g}_i t} \right)^2\\
&= \frac{1}{NT} \sumnt\left(\alpha_{g_i^0 t}^0 - \widehat{\alpha}_{\widehat{g}_i t} \right)^2
\end{align}

\end{proof}

\subsection{Proof of Theorem \ref{prop:congroups} (Consistency of group assignments)}
\textnormal{
Firstly I establish that WC group-specific effects are consistent with respect to the Hausdorff distance $d_H$ in $\R^{GT}$, defined by
\[
	d_h(a,b) = \max\left\{\max_{g\in\{1,\dots,G\}}\left(\min_{\gt \in \{1,\dots,G\}} \frac{1}{T} \sum_{t=1}^T (a_{\gt t} - b_{gt})^2 \right), \max_{g \in \{1,\dots,G\}} \left( \max_{\gt \in \{1,\dots,G\}} \frac{1}{T} \sum_{t=1}^T (a_{\gt t} - b_{gt})^2 \right)   \right\}.
\]
}
\begin{lemma}\label{lemma:hausdorff}
 Let Assumption \ref{as:con} and Assumption \ref{as:ae}($a,b$) hold. Then, as $N,T \to \infty$,
\[
	d_H(\widehat{\alpha},\alpha^0) \to_p 0.
\]
\end{lemma}
\begin{proof}
The proof is identical to that of Lemma B3 in \cite{bm:2015}. I work on the terms in the maximum. I first show for any $g=1,\dots,G$ that
\begin{equation}\label{cons:alpha:1}
	\min_{\gt \in \{1,\dots,G\}} \frac{1}{T} \sum_{t=1}^T (\widehat{\alpha}_{\gt t} - \alpha_{gt}^0)^2 \to_p 0. 
\end{equation}
Let $g \in \{1,\dots,G\}$. I have 
\[
\frac{1}{NT} \sum_{i=1}^N \left( \min_{\gt \in \{1,\dots,G\}} \sum_{t=1}^T \mathds{1}\{g_i^0 = g\}(\widehat{\alpha}_{\gt t} - \alpha_{gt}^0)^2 \right) = \left( \frac{1}{N} \sum_{i=1}^N \mathds{1}\{g_i^0 = g\}\right) \left( \min_{\gt \in \{1,\dots,G\}} \frac{1}{T}\sum_{t=1}^T (\widehat{\alpha}_{\gt t} - \alpha_{gt}^0)^2 \right).
\]
By Assumption \ref{as:ae}($a$) (group probabilities are non zero) it is enough to show that for all $g$, as $N$ and $T$ approach infinity 
\[
\frac{1}{NT} \sum_{i=1}^N \left( \min_{\gt \in \{1,\dots,G\}} \sum_{t=1}^T \mathds{1}\{g_i^0 = g\}(\widehat{\alpha}_{\gt t} - \alpha_{gt}^0)^2 \right) \to_p 0
\]
therefore to this end,
\begin{align*}
\frac{1}{N} \sum_{i=1}^N \mathds{1}\{g_i^0 = g\} \left( \min_{\gt \in \{1,\dots,G\}} \frac{1}{T}\sum_{t=1}^T (\widehat{\alpha}_{\gt t} - \alpha_{gt}^0)^2 \right)
&\leq \frac{1}{N} \sum_{i=1}^N \mathds{1}\{g_i^0 = g\} \left( \frac{1}{T}\sum_{t=1}^T (\widehat{\alpha}_{\widehat{g}_i t} - \alpha_{gt}^0)^2 \right) \\
&\leq \frac{1}{NT} \sumnt \left(\widehat{\alpha}_{\widehat{g}_i t} - \alpha_{g_i^0 t}^0 \right)^2
\end{align*}
which is $o_p(1)$ by Proposition \ref{prop:conswgfe}, which shows \eqref{cons:alpha:1}. 

Now, for the second entry in the maximum, I first define the mapping
\[
	s(g) = \argmin_{\gt \in \{1,\dots,G\}} \frac{1}{T} \sum_{t=1}^T (\widehat{\alpha}_{\gt t} - \alpha_{gt}^0)^2.
\]
and show that it is one-to-one, with probability approaching one as $T$ approaches infinity. 

Let $g\neq \widetilde{g}$. By applying the reverse triangle inequality twice I have
\begin{align*}
	\norm{\widehat{\alpha}_{s(g) } - \widehat{\alpha}_{s(\gt) }} 
&= \norm{(\alpha_g^0- \alpha_{\gt}^0) - (\alpha_g^0 - \widehat{\alpha}_{s(g) }) - (\widehat{\alpha}_{s(\gt)} - \alpha_{\gt}^0)} \\
&\geq  \norm{\alpha_g^0- \alpha_{\gt}^0} - \norm{\alpha_g^0 - \widehat{\alpha}_{s(g) }} - \norm{\widehat{\alpha}_{s(\gt)} - \alpha_{\gt}^0}.
\end{align*}
Therefore, 
\begin{align*}
	\frac{1}{\sqrt{T}}\norm{\widehat{\alpha}_{s(g) } - \widehat{\alpha}_{s(\gt) }}
&= \left(\frac{1}{T} \sum_{t=1}^T (\widehat{\alpha}_{s(g) t} - \widehat{\alpha}_{s(\gt) t })^2 \right)\\
&\geq \left(\frac{1}{T} \sum_{t=1}^T (\alpha_g^0- \alpha_{\gt}^0)^2 \right)^{1/2} - \left(\frac{1}{T} \sum_{t=1}^T(\alpha_g^0 - \widehat{\alpha}_{s(g) })^2 \right)^{1/2} - \left(\frac{1}{T} \sum_{t=1}^T (\widehat{\alpha}_{s(\gt)} - \alpha_{\gt}^0)^2 \right)^{1/2}\\
&> (c_{g,\gt} + C_{g,\gt} \sqrt{M})^{1/2} + 0 > 0 
\end{align*}
by Assumption \ref{as:ae}($b$) and \eqref{cons:alpha:1}. It follows that, with probability approaching one, $s(g) \neq s(\gt)$ for all $g\neq \gt$ i.e. $s(g)$ is one-to-one. In particular, there exists an inverse mapping $s^{-1}$ of $s$ and with probability approaching one I have for all $\gt \in \{1,\dots,G\}$:
\[
\min_{g \in \{1,\dots,G\}}\frac{1}{T}\sum_{t=1}^T (\widehat{\alpha}_{\gt t} - \alpha_{gt}^0)^2 
\leq \frac{1}{T}\sum_{t=1}^T \left(\widehat{\alpha}_{\gt t} - \alpha_{ s^{-1}(\gt)t}^0\right)^2
= \min_{h \in\{1,\dots,G\}}\frac{1}{T}\sum_{t=1}^T (\widehat{\alpha}_{h t} - \alpha_{s^{-1}(\gt)t}^0)^2 \to_p 0
\]
where I used \eqref{cons:alpha:1} and the fact that $\gt = s(s^{-1}(\gt))$. This along with \eqref{cons:alpha:1} completes the proof.
\end{proof}

\textnormal{
The proof of Lemma \ref{lemma:hausdorff} shows that there exists a permutation $s: \{1,\dots,G\} \to \{1,\dots,G\}$ such that
\[
	\frac{1}{T} \sum_{t=1}^T (\widehat{\alpha}_{s(g) t} - \alpha_{gt}^0)^2 \to_p 0.
\]
By relabeling the elements of $\widehat{\alpha}$ I can take $s(g) = g$ and this is a convention that is adopted for what remains.}

\textnormal{
The group membership profiles $\gamma_N =(g_1,\dots,g_N) \in \Gamma_{G}^N$ characterize a histogram 
\[
\lambda(\gamma_N) = N^{-1}\left(\sum_{i=1}^N \mathds{1}\{g_i = 1\},\dots, \sum_{i=1}^N \mathds{1}\{g_i = G\} \right) \in \Delta^G
\]
which converge to a probability mass function given the $\gamma_N = \{g_i\}_{i=1}^N$. In particular, the WGFE estimator $\widehat{\gamma}_N$ defines a sample mass function that will converge in probability to a probability mass function.}\textnormal{ Now, define 
\begin{equation}
	\widehat{\sigma}_g^2(\theta,\alpha,\gamma_{N}) = \frac{1}{T\sum_{j=1}^N \mathds{1}\{g_j = g\}} \sumnt \mathds{1}\{g_i = g\}(u_{it} + x_{it}'(\theta^0 - \theta) + \alpha_{gt}^0 - \alpha_{gt})^2
\end{equation}
for any $(\theta,\alpha,\gamma_{N}) \in \Theta \times \mathcal{A}^{GT}\times \Gamma_G^N$.}

\textnormal{
By a weak law of large numbers, I have for any $g$ as $N\to\infty$,
\begin{equation}
	\widehat{\sigma}_g^{2}(\theta,\alpha,\gamma_N)\longrightarrow_p \sigma_g^2(\theta,\alpha,\gamma) = \frac{1}{T}\sum_{t=1}^T\Ex{ \mathds{1}\{g_i = g\}(u_{it} + x_{it}'(\theta^0 - \theta) + \alpha_{gt}^0 - \alpha_{gt})^2}
\end{equation}
where $\gamma$ is defined by $\lambda(\gamma^N) \to_p \lambda(\gamma) \in \Delta^G$. When $\theta = \theta^0$ and $\alpha = \alpha^0$ then this is a fixed-$T$ variance of a subpopulation of $u_{it}$ defined by the random partition $\gamma$.}

\textnormal{
For any $\eta>0$ I define the subset of parameters $(\theta,\alpha) \in \mathcal{N}_\eta \subset \Theta \times \mathcal{A}^{GT}$ that satisfy  $\norm{\theta - \theta^0} ^2 <\eta$,  $\frac{1}{T}\norm{\alpha_g - \alpha_g^0}^2 < \eta$ and \[\left\vert \frac{\widehat{\sigma}_g(\theta,\alpha,\widehat{\gamma}_N)}{\widehat{\sigma}_{\gt}(\theta,\alpha,\widehat{\gamma}_N)}  - \frac{\widehat{\sigma}_g(\theta^0,\alpha^0,\widehat{\gamma}_N)}{\widehat{\sigma}_{\gt}(\theta^0,\alpha^0,\widehat{\gamma}_N)}\right\vert < \eta\] for all $g,\gt = 1, \dots, G$. The first step to showing consistency of group assignments is showing the sample probability of missassignment of individuals converges at exponential rate to zero with respect to $T$ over all parameter values within a small enough neighborhood of the true values. Denote for fixed $(\theta,\alpha)$ the optimal assignment of individual $i= 1,\dots,N$ according to criterion \eqref{wgfe:obj} as $\widehat{g}_i(\theta,\alpha,\widehat{\gamma}_N)$ where $\widehat{\gamma}_N$ is the collection of optimal assignments used to determine sample group variances, which is also a function of $(\theta,\alpha)$.
}
\begin{lemma}\label{lemma:missassign}
	For $\eta >0$ small enough I have, for all $\delta>0$,
	\[
		\sup_{(\theta,\alpha) \in \mathcal{N}_\eta} \frac{1}{N}\sum_{i =1}^N \mathds{1}\{\widehat{g}_i(\theta,\alpha,\widehat{\gamma}_N) \neq g_i^0\} = o_p(T^{-\delta})
	\]
	as $N,T \to \infty$.
\end{lemma}
\begin{proof}
For this proof I follow Bonhomme \& Manresa (2015) wherever possible, however adjustments need to be made with the inclusion of second moment information.We suppress the $\widehat{\gamma}_N$ notation in the group variances until it is convenient to include it once again i.e. $\widehat{\sigma}_g(\theta,\alpha,\widehat{\gamma}_N) =\widehat{\sigma}_g(\theta,\alpha)$. By definition, I have, for all $g = 1, \dots,G$:
\begin{align*}
	&\mathds{1}\{\widehat{g}_i(\theta,\alpha,\widehat{\gamma}_N) = g\}\\ \leq& \mathds{1}\left\{ \frac{1}{\widehat{\sigma}_g(\theta,\alpha)} \sum_{t=1}^T (y_{it} - x_{it}'\theta - \alpha_{gt})^2 +  \widehat{\sigma}_g(\theta,\alpha) \leq \frac{1}{\widehat{\sigma}_{g_i^0}(\theta,\alpha)} \sum_{t=1}^T (y_{it} - x_{it}'\theta - \alpha_{g_i^0 t})^2 +  \widehat{\sigma}_{g_i^0}(\theta,\alpha)\right\}.
\end{align*}

Then,
\begin{align*}
	&\frac{1}{N}\sum_{i =1}^N  \mathds{1}\{\widehat{g}_i(\theta,\alpha,\widehat{\gamma}_N) \neq g_i^0\} 
= \sum_{g=1}^G \frac{1}{N}\sum_{i=1}^N \mathds{1}\{g_i^0 \neq g\} \mathds{1}\{\widehat{g}_i(\theta,\alpha,\widehat{\gamma}_N) = g\}\\
\hspace{-15mm}\leq& \sum_{g=1}^G \frac{1}{N}\sum_{i=1}^N \mathds{1}\{g_i^0 \neq g\}\mathds{1}\biggr\{\frac{1}{\widehat{\sigma}_g(\theta,\alpha)} \sum_{t=1}^T (y_{it} - x_{it}'\theta - \alpha_{gt})^2 +  \widehat{\sigma}_g(\theta,\alpha) \\\leq& \frac{1}{\widehat{\sigma}_{g_i^0}(\theta,\alpha)} \sum_{t=1}^T (y_{it} - x_{it}'\theta - \alpha_{g_i^0 t})^2 +  \widehat{\sigma}_{g_i^0}(\theta,\alpha)\biggr\}
\end{align*}
Define $Z_{ig}(\theta,\alpha,\widehat{\gamma}_N)$ for $(\theta,\alpha) \in \mathcal{N}_\eta$ as the inner-most summand. I will bound this term by a quantity that does not depend on the parameters. Note first that
\begin{gather*}
Z_{ig}(\theta,\alpha,\widehat{\gamma}_N) = \mathds{1}\{g_i^0 \neq g\}\\ 
\times \mathds{1}\biggr\{\frac{1}{\widehat{\sigma}_g(\theta,\alpha)} \sum_{t=1}^T (u_{it} + x_{it}'(\theta^0 - \theta) + \alpha_{g_i^0 t}^0 - \alpha_{gt})^2 +  \widehat{\sigma}_g(\theta,\alpha)\\ \leq \frac{1}{\widehat{\sigma}_{g_i^0}(\theta,\alpha)} \sum_{t=1}^T (u_{it} + x_{it}'(\theta^0 - \theta) + \alpha_{g_i^0 t}^0 - \alpha_{g_i^0 t})^2 +  \widehat{\sigma}_{g_i^0}(\theta,\alpha)\biggr\}\\
\leq\max_{g\neq \gt}\mathds{1}\biggr\{\frac{1}{\widehat{\sigma}_g(\theta,\alpha)} \sum_{t=1}^T (u_{it} + x_{it}'(\theta^0 - \theta) + \alpha_{\gt t}^0 - \alpha_{gt})^2 +  \widehat{\sigma}_g(\theta,\alpha) \\\leq \frac{1}{\widehat{\sigma}_{\gt}(\theta,\alpha)} \sum_{t=1}^T (u_{it} + x_{it}'(\theta^0 - \theta) + \alpha_{\gt t}^0 - \alpha_{\gt t})^2 +  \widehat{\sigma}_{\gt}(\theta,\alpha)\biggr\}
\end{gather*}

Before establishing a bound, I will rewrite the inequality condition through simple algebra. I can rewrite the expression within the indicator function of the previous inequality as
\begin{gather*}
	\sum_{t=1}^T (\alpha_{\gt t} - \alpha_{gt})\left(u_{it} + x_{it}'(\theta^0 - \theta) + \alpha_{\gt t}^0 - \frac{\alpha_{\gt t} + \alpha_{gt}}{2}\right) \\
\leq \frac{1}{2}\left(\frac{\widehat{\sigma}_g(\theta,\alpha) }{\widehat{\sigma}_{\gt}(\theta,\alpha)} - 1\right)\left( \sum_{t=1}^T (u_{it} + x_{it}'(\theta^0 - \theta))^2 + \sum_{t=1}^T (\alpha_{\gt t}^0 - \alpha_{\gt t})^2 \right) \\+ \frac{\widehat{\sigma}_g(\theta,\alpha)}{\widehat{\sigma}_{\gt} (\theta,\alpha)} \sum_{t=1}^T (\alpha_{\gt t}^0 - \alpha_{\gt t})(u_{it} + x_{it}'\left(\theta^0 - \theta)\right)\\ + \widehat{\sigma}_g(\theta,\alpha)\left(\frac{\widehat{\sigma}_{\gt}(\theta,\alpha) - \widehat{\sigma}_g(\theta,\alpha)}{2}\right)
\end{gather*}
Make an addition on both sides by
\[
	A_T = \sum_{t=1}^T (\alpha_{\gt t}^0 - \alpha_{g t}^0)\left(u_{it} + \alpha_{\gt t}^0 - \frac{\alpha_{\gt t}^0 + \alpha_{g t}^0}{2}\right) - \sum_{t=1}^T (\alpha_{\gt t} - \alpha_{gt})\left(u_{it} + x_{it}'(\theta^0 - \theta) + \alpha_{\gt t}^0 - \frac{\alpha_{\gt t} + \alpha_{gt}}{2}\right) .
\]
Then, the inequality becomes
\begin{gather*}
\sum_{t=1}^T (\alpha_{\gt t}^0 - \alpha_{g t}^0)\left(u_{it} + \alpha_{\gt t}^0 - \left(\frac{\alpha_{\gt t}^0 + \alpha_{g t}^0}{2}\right) \right)\\
\leq A_T +\frac{1}{2}\left(\frac{\widehat{\sigma}_g(\theta,\alpha) }{\widehat{\sigma}_{\gt}(\theta,\alpha)} - 1\right)\left( \sum_{t=1}^T (u_{it} + x_{it}'(\theta^0 - \theta))^2 + \sum_{t=1}^T (\alpha_{\gt t}^0 - \alpha_{\gt t})^2 \right) \\+ \frac{\widehat{\sigma}_g(\theta,\alpha)}{\widehat{\sigma}_{\gt} (\theta,\alpha)} \sum_{t=1}^T (\alpha_{\gt t}^0 - \alpha_{\gt t})(u_{it} + x_{it}'\left(\theta^0 - \theta)\right)\\ + \widehat{\sigma}_g(\theta,\alpha)\left(\frac{\widehat{\sigma}_{\gt}(\theta,\alpha) - \widehat{\sigma}_g(\theta,\alpha)}{2}\right).
\end{gather*}
Using this form, consider another bound of $Z_{ig}$ through absolute values
\begin{gather}
Z_{ig}(\theta,\alpha,\widehat{\gamma}_N) \leq \max_{g\neq \gt} \mathds{1}\bigg\{ \sum_{t=1}^T (\alpha_{\gt t}^0 - \alpha_{g t}^0)\left(u_{it} + \alpha_{\gt t}^0 - \left(\frac{\alpha_{\gt t}^0 + \alpha_{g t}^0}{2}\right) \right)\notag\\ 
\leq |A_T| + \frac{1}{2}\left\vert\frac{\widehat{\sigma}_g(\theta,\alpha) }{\widehat{\sigma}_{\gt}(\theta,\alpha)} - 1\right\vert\left( \sum_{t=1}^T u_{it}^2 + \sum_{t=1}^T|x_{it}'(\theta^0 - \theta)|^2 
+ 2\sum_{t=1}^Tu_{it}x_{it}'(\theta^0 - \theta) + \sum_{t=1}^T (\alpha_{\gt t}^0 - \alpha_{\gt t})^2 \right) \label{firstterm}\\+ \frac{\widehat{\sigma}_g(\theta,\alpha)}{\widehat{\sigma}_{\gt} (\theta,\alpha)} \sum_{t=1}^T (\alpha_{\gt t}^0 - \alpha_{\gt t})(u_{it} + x_{it}'\left(\theta^0 - \theta)\right)\label{secondterm}\\ + \frac{\widehat{\sigma}_g(\theta,\alpha) \widehat{\sigma}_{\gt}(\theta,\alpha)}{2} \left\vert\frac{\widehat{\sigma}_g(\theta,\alpha) }{\widehat{\sigma}_{\gt}(\theta,\alpha)} - 1\right\vert\bigg\}
\end{gather}
Next, I will show a bound that does not depend on $(\theta,\alpha)$ by bounding each individual term on the right-hand side of the inner inequality.

From \cite{bm:2015}, I have the bound
\[
	\vert A_T \vert \leq TC_1 \sqrt{\eta} \left( \frac{1}{T}\sum_{t=1}^T u_{it}^2\right)^{1/2} 
+ TC_2 \sqrt{\eta} \left(\frac{1}{T} \sum_{t=1}^T \norm{x_{it}} \right) + TC_3\sqrt{\eta}
\]
where $C_1,C_2, C_3$ are constants that are independent of $\eta$ and $T$, which come from compactness of the parameter space.

The terms in the first pair of parenthesis of \eqref{firstterm} are bounded by
\begin{gather*}
	\sum_{t=1}^T u_{it}^2 + \sum_{t=1}^T|x_{it}'(\theta^0 - \theta)|^2 
+ 2\sum_{t=1}^Tu_{it}x_{it}'(\theta^0 - \theta) + \sum_{t=1}^T (\alpha_{\gt t}^0 - \alpha_{\gt t})^2 \\
\leq T \left( \frac{1}{T}\sum_{t=1}^T u_{it}^2\right) + T\eta \left( \frac{1}{T} \sum_{t=1}^T \norm{x_{it}}^2\right) +
2T\sqrt{\eta} \left(\frac{1}{T} \sum_{t=1}^T u_{it}^2 \right)^{1/2}  \left(\frac{1}{T} \sum_{t=1}^T \norm{x_{it}}^2\right)^{1/2}
+ T\eta
\end{gather*}
from the Cauchy-Schwarz inequality on the second and third terms along with the definition of $\mathcal{N}_\eta$ and the definition of $\mathcal{N}_\eta$ for the last term.

The terms in \eqref{secondterm} are once again bounded using the Cauchy-Schwarz inequality
\begin{gather*}
	\sum_{t=1}^T (\alpha_{\gt t}^0 - \alpha_{\gt t})u_{it} + \sum_{t=1}^T (\alpha_{\gt t}^0 - \alpha_{\gt t})\left(x_{it}'\left(\theta^0 - \theta\right)\right)\\
\leq T\sqrt{\eta} \left(\frac{1}{T}\sum_{t=1}^T u_{it}^2 \right)^{1/2} + T\eta\left(\frac{1}{T} \sum_{t=1}^T \norm{x_{it}}^2\right)^{1/2}.
\end{gather*}

As for the variance terms, I bring back to notation indicating the collection of group assignments and bound as:
\begin{align*}
\left\vert\frac{\widehat{\sigma}_g(\theta,\alpha,\widehat{\gamma}_N) }{\widehat{\sigma}_{\gt}(\theta,\alpha,\widehat{\gamma}_N)} - 1\right\vert & \leq \left\vert\frac{\widehat{\sigma}_g(\theta,\alpha,\widehat{\gamma}_N) }{\widehat{\sigma}_{\gt}(\theta,\alpha,\widehat{\gamma}_N)} - \frac{\widehat{\sigma}_g(\theta^0,\alpha^0,\widehat{\gamma}_N) }{\widehat{\sigma}_{\gt}(\theta^0,\alpha^0,\widehat{\gamma}_N)} +\frac{\widehat{\sigma}_g(\theta^0,\alpha^0,\widehat{\gamma}_N) }{\widehat{\sigma}_{\gt}(\theta^0,\alpha^0,\widehat{\gamma}_N)}  -  1\right\vert\\
&\leq \left\vert\frac{\widehat{\sigma}_g(\theta,\alpha,\widehat{\gamma}_N) }{\widehat{\sigma}_{\gt}(\theta,\alpha,\widehat{\gamma}_N)} - \frac{\widehat{\sigma}_g(\theta^0,\alpha^0,\widehat{\gamma}_N) }{\widehat{\sigma}_{\gt}(\theta^0,\alpha^0,\widehat{\gamma}_N)} \right\vert + \left\vert \frac{\widehat{\sigma}_g(\theta^0,\alpha^0,\widehat{\gamma}_N) }{\widehat{\sigma}_{\gt}(\theta^0,\alpha^0,\widehat{\gamma}_N)} - 1\right\vert\\
& < \eta + \sup_{\lambda(\gamma_N) \in \Delta^G} \left\vert \frac{\widehat{\sigma}_g(\theta^0,\alpha^0,\gamma_N) }{\widehat{\sigma}_{\gt}(\theta^0,\alpha^0,\gamma_N)} - 1 \right\vert
\end{align*}
Next, I bound the ratio. I apply the triangle inequality after two additions by zero:
\begin{gather}
\frac{\widehat{\sigma}_g(\theta,\alpha,\widehat{\gamma}_N)}{\widehat{\sigma}_{\gt}(\theta,\alpha,\widehat{\gamma}_N)} \leq \left\vert \frac{\widehat{\sigma}_g(\theta,\alpha,\widehat{\gamma}_N)}{\widehat{\sigma}_{\gt}(\theta,\alpha,\widehat{\gamma}_N)} - \frac{\widehat{\sigma}_g(\theta^0,\alpha^0,\widehat{\gamma}_N)}{\widehat{\sigma}_{\gt}(\theta^0,\alpha^0,\widehat{\gamma}_N)} \right\vert + \left\vert \frac{\widehat{\sigma}_g(\theta^0,\alpha^0,\widehat{\gamma}_N)}{\widehat{\sigma}_{\gt}(\theta^0,\alpha^0,\widehat{\gamma}_N)} - 1\right\vert + 1 \\\leq
 \eta + \sup_{\lambda(\gamma_N) \in \Delta^G}\left\vert \frac{\widehat{\sigma}_g(\theta^0,\alpha^0,{\gamma}_N)}{\widehat{\sigma}_{\gt}(\theta^0,\alpha^0,{\gamma}_N)} - 1\right\vert + 1 
\end{gather}

Next, I can easily bound $\widehat{\sigma}_g(\theta,\alpha,\widehat{\gamma}_N)$ independent of $(\theta,\alpha,\widehat{\gamma}_N)$ using compactness of the parameter space. To see this,
\begin{gather*}
\widehat{\sigma}_g^2(\theta,\alpha,\widehat{\gamma}_N) 
= \frac{1}{\sum_{i=1}^N \mathds{1}\{\widehat{g}_i= g\}}\sum_{i=1}^N \mathds{1}\{\widehat{g}_i= g\}\Biggr( \frac{1}{T} \sum_{t=1}^T u_{it}^2 + \frac{1}{T} \sum_{t=1}^T \vert x_{it}'(\theta^0 - \theta)\vert^2 + \frac{1}{T} \sum_{t=1}^T (\alpha_{\gt t}^0 - \alpha_{g t})^2 \\ 
+ 2\frac{1}{T} \sum_{t=1}^T x_{it}'(\theta^0 - \theta)(\alpha_{\gt t}^0 - \alpha_{g t}) + 2\frac{1}{T} \sum_{t=1}^Tu_{it} x_{it}'(\theta^0 - \theta) \\+2\frac{1}{T} \sum_{t=1}^Tu_{it} (\alpha_{\gt t}^0 - \alpha_{g t})\Biggr)\\ 
\leq  \max_{\gamma\in\Gamma_G^N}\sum_{i=1}^N \frac{\mathds{1}\{{g}_i= g\}}{\sum_{i=1}^N \mathds{1}\{{g}_i= g\}}\Biggr[ \left(\frac{1}{T} \sum_{t=1}^T u_{it}^2 \right) + \eta \left(\frac{1}{T} \sum_{t=1}^T \norm{x_{it}}^2 \right) + C_\alpha + 2C_\alpha\sqrt{\eta} \left( \frac{1}{T} \sum_{t=1}^T\norm{x_{it}} \right)\\ + 2\sqrt{\eta}\left(\frac{1}{T} \sum_{t=1}^T u_{it}^2\right)^{1/2}\left(\frac{1}{T} \sum_{t=1}^T\norm{x_{it}}^2 \right)^{1/2} + 2\sqrt{\eta}C_\alpha \left(\frac{1}{T} \sum_{t=1}^T u_{it}^2\right)^{1/2}\Biggr]
\end{gather*}
where $C_\alpha$ is the constant that bounds $\vert\alpha_{\gt t}^0 - \alpha_{gt}\vert$ for all $g,\gt, t$ due to compactness of $\mathcal{A}$. Set this bound as $\widehat{\sigma}(N,T)$ so that 
\begin{equation}
	\widehat{\sigma}_g(\theta,\alpha,\widehat{\gamma}_N)\widehat{\sigma}_{\gt}(\theta,\alpha,\widehat{\gamma}_N) \leq \widehat{\sigma}^2(N,T).
\end{equation}

For \eqref{firstterm},
\begin{gather*}
\sum_{t=1}^T u_{it}^2 + 2\sum_{t=1}^Tu_{it} x_{it}'(\theta^0 - \theta) + \sum_{t=1}^T \vert x_{it}'(\theta^0 - \theta)\vert^2
+ \sum_{t=1}^T (\alpha_{\gt t}^0 - \alpha_{\gt t})^2\\
\leq T\left(\frac{1}{T}\sum_{t=1}^T u_{it}^2 \right) + 2T\sqrt{\eta}\left(\frac{1}{T}\sum_{t=1}^T u_{it}^2 \right)^{1/2}\left(\frac{1}{T}\sum_{t=1}^T \norm{x_{it}}^2 \right)^{1/2} +\eta T\left(\frac{1}{T}\sum_{t=1}^T \norm{x_{it}}^2\right) + T\eta.
\end{gather*}
For \eqref{secondterm}, 
\begin{gather*}
\sum_{t=1}^T (\alpha_{\gt t}^0 - \alpha_{\gt t}) u_{it} + \sum_{t=1}^T (\alpha_{\gt t}^0 - \alpha_{\gt t})x_{it}'(\theta^0 - \theta) \\
\leq T\sqrt{\eta} \left(\frac{1}{T}\sum_{t=1}^T u_{it}^2 \right)^{1/2} + T\eta\left(\frac{1}{T}\sum_{t=1}^T \norm{x_{it}}^2\right)
\end{gather*}

Therefore, setting $\widehat{C}_{g,\gt} = \sup_{\lambda(\gamma_N) \in \Delta^G} \left\vert \frac{\widehat{\sigma}_g(\theta^0,\alpha^0,\gamma_N) }{\widehat{\sigma}_{\gt}(\theta^0,\alpha^0,\gamma_N)} - 1 \right\vert$, I have
\begin{gather}
Z_{ig}(\theta,\alpha,\widehat{\gamma}_N) \leq \max_{g\neq \gt} \mathds{1}\Biggr\{ \sum_{t=1}^T (\alpha_{\gt t}^0 - \alpha_{g t}^0)\left(u_{it} + \alpha_{\gt t}^0 - \left(\frac{\alpha_{\gt t}^0 + \alpha_{g t}^0}{2}\right) \right)\label{zsquiggle}\\ 
\leq  TC_1 \sqrt{\eta} \left( \frac{1}{T}\sum_{t=1}^T u_{it}^2\right)^{1/2} 
+ TC_2 \sqrt{\eta} \left(\frac{1}{T} \sum_{t=1}^T \norm{x_{it}} \right) + TC_3\sqrt{\eta}\\
+\frac{1}{2}\left(\eta +\widehat{C}_{g,\gt}\right) \left(T\left(\frac{1}{T}\sum_{t=1}^T u_{it}^2 \right) + 2T\sqrt{\eta}\left(\frac{1}{T}\sum_{t=1}^T u_{it}^2 \right)^{1/2}\left(\frac{1}{T}\sum_{t=1}^T \norm{x_{it}}^2 \right)^{1/2} +\eta T\left(\frac{1}{T}\sum_{t=1}^T \norm{x_{it}}^2\right) + T\eta \right)\\
+ \left(\eta + \widehat{C}_{g,\gt} + 1\right)\left(T\sqrt{\eta} \left(\frac{1}{T}\sum_{t=1}^T u_{it}^2 \right)^{1/2} + T\eta\left(\frac{1}{T}\sum_{t=1}^T \norm{x_{it}}^2\right) \right)\\
+\frac{1}{2}\widehat{\sigma}^2(N,T)\widehat{C}_{g,\gt} \Biggr\}
\end{gather}
Therefore, the right-hand side of the inequality does not depend on $(\theta,\alpha,\widehat{\gamma}_N)$ so, setting it as $\widetilde{Z}_{ig}$, I see that $\sup_{(\theta,\alpha) \in \mathcal{N}_\eta} Z_{ig}(\theta,\alpha,\widehat{\gamma}_N) \leq \widetilde{Z}_{ig}$. Hence,
\[
	\sup_{(\theta,\alpha) \in \mathcal{N}_\eta} \frac{1}{N}\sum_{i=1}^N \mathds{1}\{\widehat{g}_i(\theta,\alpha) \neq g_i^0\} \leq \frac{1}{N} \sum_{i=1}^N \sum_{g=1}^G \widetilde{Z}_{ig}.
\]
Before moving on I will rearrange terms to get a cleaner bound by letting $\eta^* \geq \max\{\eta,\eta^2,\eta\sqrt{\eta}\}$: starting from the first term,
\[
	T\sqrt{\eta}\left(C_1 \left( \frac{1}{T}\sum_{t=1}^T u_{it}^2\right)^{1/2} 
+ C_2  \left(\frac{1}{T} \sum_{t=1}^T \norm{x_{it}} \right) + C_3\right) \leq T\eta^*\left(C_1 \left( \frac{1}{T}\sum_{t=1}^T u_{it}^2\right)^{1/2} 
+ C_2  \left(\frac{1}{T} \sum_{t=1}^T \norm{x_{it}} \right) + C_3\right). 
\]
The second line is bounded by
\[
T\widehat{C}_{g,\gt} \left(\frac{1}{T}\sum_{t=1}^T u_{it}^2 \right) + T\eta^*(1 + \widehat{C}_{g,\gt})\left[\left(\frac{1}{T}\sum_{t=1}^T u_{it}^2\right) + 2 \left(\frac{1}{T}\sum_{t=1}^T u_{it}^2 \right)^{1/2}\left(\frac{1}{T}\sum_{t=1}^T \norm{x_{it}}^2 \right)^{1/2} + \left(\frac{1}{T}\sum_{t=1}^T \norm{x_{it}}^2 \right) + 1 \right]
\]
where I take special note that the first term is independent of $\eta^*$.

The third line is bounded by
\[
	T\eta^*(2 + \widehat{C}_{g,\gt})\left[\left( \frac{1}{T} \sum_{t=1}^T u_{it}^2 \right)^{1/2} + \left(\frac{1}{T}\sum_{t=1}^T \norm{x_{it}}^2 \right)\right].
\]
The fourth term will end up negligible and so we'll come back to it shortly in the proof.

Now, fix $\widetilde{M} > \max\{\sqrt{M}, M^*\}$, where these are constants given in Assumption \ref{as:infeas}($b$) and Assumption \ref{as:ae}($e$). Note by Jensen's inequality I have $\Ex{u_{it}^2} \leq \sqrt{M}$ and by Cauchy-Schwarz $M^* \leq \frac{1}{T}\sum_{t=1}^T \norm{x_{it}} \leq \left(\frac{1}{T}\sum_{t=1}^T \norm{x_{it}}^2\right)^{1/2}$. Therefore using standard probability algebra (union rule) and for all $g = 1,\dots,G$:
\begin{gather*}
\mathbb{P}(\widetilde{Z}_{ig} = 1)\\ \leq \sum_{\gt \neq g} \mathbb{P}\Biggr[ \sum_{t=1}^T (\alpha_{\gt t}^0 - \alpha_{g t}^0)u_{it} 
\leq  -\frac{1}{2}\sum_{t=1}^T (\alpha_{\gt t}^0 - \alpha_{gt}^0)^2 + \frac{1}{2}T\widehat{C}_{g,\gt} \left(\frac{1}{T}\sum_{t=1}^T u_{it}^2 \right)\\
 + T\eta^*(1 + \widehat{C}_{g,\gt})\left(\left(\frac{1}{T}\sum_{t=1}^T u_{it}^2\right) + 2 \left(\frac{1}{T}\sum_{t=1}^T u_{it}^2 \right)^{1/2}\left(\frac{1}{T}\sum_{t=1}^T \norm{x_{it}}^2 \right)^{1/2} + \left(\frac{1}{T}\sum_{t=1}^T \norm{x_{it}}^2 \right) + 1 \right)\\  
+ T\eta^*\left(C_1 \left( \frac{1}{T}\sum_{t=1}^T u_{it}^2\right)^{1/2} 
+ C_2  \left(\frac{1}{T} \sum_{t=1}^T \norm{x_{it}} \right) + C_3\right)\\  
+T\eta^*(2 + \widehat{C}_{g,\gt})\left(\left( \frac{1}{T} \sum_{t=1}^T u_{it}^2 \right)^{1/2} + \left(\frac{1}{T}\sum_{t=1}^T \norm{x_{it}}^2 \right)\right)\\
 +\frac{1}{2}\widehat{\sigma}^2(N,T)\widehat{C}_{g,\gt} 
 \Biggr]
\end{gather*}
which is bounded by
\begin{gather*}
\mathbb{P}(\widetilde{Z}_{ig}(\widehat{\gamma}^N) = 1) \leq \sum_{g\neq \gt} \mathbb{P}\left(\frac{1}{T}\sum_{t=1}^T \norm{x_{it}} \geq \widetilde{M} \right)
+ \mathbb{P}\left(\frac{1}{T} \sum_{t=1}^T u_{it}^2 \geq \widetilde{M} \right)\\ + \mathbb{P}\left( \sup_{\lambda(\gamma_N) \in \Delta^G} \left\vert \frac{\widehat{\sigma}_g(\theta^0,\alpha^0,\gamma_N) }{\widehat{\sigma}_{\gt}(\theta^0,\alpha^0,\gamma_N)} - 1 \right\vert \geq \sup_{\lambda(\gamma) \in \Delta^G} \left\vert \frac{{\sigma}_g(\theta^0,\alpha^0,\gamma) }{{\sigma}_{\gt}(\theta^0,\alpha^0,\gamma)} - 1 \right\vert \equiv C_{g,\gt} \right)\\ 
+ \mathbb{P}\left( \frac{1}{T} \sum_{t=1}^T (\alpha_{\gt t}^0 - \alpha_{gt}^0)^2 \leq \frac{c_{g,\gt}}{2} + \sup_{\lambda(\gamma) \in \Delta^G} \left\vert \frac{{\sigma}_g(\theta^0,\alpha^0,\gamma) }{{\sigma}_{\gt}(\theta^0,\alpha^0,\gamma)} - 1 \right\vert \widetilde{M}\right)\\
\hspace{-15mm}+ \sum_{g\neq \gt} \mathbb{P}\Biggr( 
\sum_{t=1}^T (\alpha_{\gt t}^0 - \alpha_{g t}^0)u_{it} 
\leq  -T \frac{c_{g,\gt}}{4}  + T\eta^*\left( (5C_{g,\gt} + C_2 + 6)\widetilde{M} + (C_1 + C_{g,\gt} + 2) \sqrt{\widetilde{M}} + (C_3 + C_{g,\gt} + 1)\right) + C_0\Biggr).
\end{gather*}
where the term independent of $\eta^*$ has been eliminated and the constant comes from the union rule:
\begin{align*}
\frac{C_{g,\gt}}{2}\widehat{\sigma}^2(N,T) = C_0 &= \frac{C_{g,\gt}}{2}\left(\widetilde{M} + \eta^*\left( (3 + 2 C_\alpha)\widetilde{M} + 2C_\alpha\sqrt{\widetilde{M}}\right) + C_\alpha\right)\max_{\gamma\in\Gamma_G^N} \sum_{i=1}^N \frac{\mathds{1}\{g_i = g\}}{\sum_{i=1}^N \mathds{1}\{g_i = g\}}\\
&=  \frac{C_{g,\gt}}{2}\left(\widetilde{M} + \eta^*\left( (3 + 2 C_\alpha)\widetilde{M} + 2C_\alpha\sqrt{\widetilde{M}}\right) + C_\alpha\right) \times 1
\end{align*}
given $C_\alpha$ some constant from the boundedness property of the compact subspace $\mathcal{A}$. 

I will show the exponential rates for each term in the sum. But first, a lemma from \cite{bm:2015}. 

\begin{lemma}[Lemma B5 of \cite{bm:2015}]\label{lemma:b3}
Let $z_t$ be a strongly mixing process with zero mean, with strong mixing coefficients $\alpha[t] \leq e^{-at^{d_1}}$, and with tail probabilities $\mathbb{P}(|z_t| > z) \leq e^{1 - (z/b)^{d_2}}$, where $a,b,d_1,d_2>0$ are constants. Then, for all $z>0$ I have, for all $\delta >0$, as $T \to \infty$,
\[
	T^{\delta} \mathbb{P}\left( \left\vert \frac{1}{T} \sum_{t=1}^T z_t \right\vert \geq z \right) \to 0.
\]
\end{lemma}

The second term is $o(T^{-\delta})$ by Lemma B5 of \cite{bm:2015} and Assumption \ref{as:infeas}($b$). To see this, set $z_t = u_{it}^2 - \Ex{u_{it}^2}$, which is necessarily strongly mixing by Assumption \ref{as:ae}($b$), and taking $z = \widetilde{M} - \sqrt{M} > 0$ then for any $\delta > 0$
\begin{align*}
&o_p(T^{-\delta}) \\ &= \mathbb{P}\left(\left\vert \frac{1}{T}\sum_{t=1}^T \left( u_{it}^2 - \Ex{u_{it}^2}\right) \right\vert \geq z \right)\\ &= \mathbb{P}\left(\left\{\frac{1}{T}\sum_{t=1}^T \left( u_{it}^2 - \Ex{u_{it}^2}\right) \leq -z\right\}\bigcup \left\{ \frac{1}{T}\sum_{t=1}^T \left( u_{it}^2 - \Ex{u_{it}^2}\right) \geq z \right\}\right)\\
&\geq \mathbb{P}\left( \frac{1}{T}\sum_{t=1}^T  u_{it}^2  \geq \widetilde{M} -  (\sqrt{M} - \Ex{u_{it}^2}) \right)\\
&\geq\mathbb{P}\left( \frac{1}{T}\sum_{t=1}^T  u_{it}^2  \geq \widetilde{M}\right)
\end{align*}
where the last inequality is due to $\Ex{u_{it}^2} \leq \sqrt{M}$ by Assumption \ref{as:infeas}($b$). Therefore, 
\[
	\mathbb{P}\left( \frac{1}{T}\sum_{t=1}^T  u_{it}^2  \geq \widetilde{M}\right) = o(T^{-\delta}).
\]
The third term is $o(1)$ as $N\to \infty$.

The fourth term is seen as $o(T^{-\delta})$ from a modification of the argument from \cite{bm:2015}. Since
\[\lim_{T\to \infty} \frac{1}{T}\sum_{t=1}^T \Ex{(\alpha_{\gt t}^0 - \alpha_{gt}^0)^2} > c_{g,\gt} + C_{g,\gt} \widetilde{M},\] I have for $T$ large enough,
\[
	 \frac{1}{T}\sum_{t=1}^T \Ex{(\alpha_{\gt t}^0 - \alpha_{gt}^0)^2} \geq c_{g,\gt} + C_{g,\gt}  \widetilde{M}.
\]
Then, applying Lemma B5 of \cite{bm:2015} with $z_t = (\alpha_{\gt t}^0 - \alpha_{gt}^0)^2 - \Ex{(\alpha_{\gt t}^0 - \alpha_{gt}^0)^2}$, which satisfies appropriate mixing and tail conditions by assumptions \ref{as:infeas}($a$) and \ref{as:ae}($b$), and setting $z = c_{g,\gt}/2 + C_{g,\gt}\widetilde{M}$ yields
\begin{align*}
o_p(T^{-\delta}) &= \mathbb{P}\left(\left\vert \frac{1}{T}\sum_{t=1}^T (\alpha_{\gt t}^0 - \alpha_{gt}^0)^2 - \Ex{(\alpha_{\gt t}^0 - \alpha_{gt}^0)^2}\right\vert \geq z \right)\\ &= \mathbb{P}\left(\left\{\frac{1}{T}\sum_{t=1}^T (\alpha_{\gt t}^0 - \alpha_{gt}^0)^2 - \Ex{(\alpha_{\gt t}^0 - \alpha_{gt}^0)^2} \leq -z\right\}\bigcup \left\{ \frac{1}{T}\sum_{t=1}^T (\alpha_{\gt t}^0 - \alpha_{gt}^0)^2 - \Ex{(\alpha_{\gt t}^0 - \alpha_{gt}^0)^2} \geq z \right\}\right)\\
&\geq \mathbb{P}\left( \frac{1}{T}\sum_{t=1}^T  (\alpha_{\gt t}^0 - \alpha_{gt}^0)^2 \leq -c_{g,\gt}/2  + \frac{1}{T}\sum_{t=1}^T \Ex{(\alpha_{\gt t}^0 - \alpha_{gt}^0)^2} ) \right)\\
&\geq \mathbb{P}\left( \frac{1}{T}\sum_{t=1}^T  (\alpha_{\gt t}^0 - \alpha_{gt}^0)^2 \leq c_{g,\gt}/2  + C_{g,\gt}  \widetilde{M}) \right)
\end{align*}
therefore,
\[
	\mathbb{P} \left( \frac{1}{T} \sum_{t=1}^T (\alpha_{\gt t}^0 - \alpha_{gt}^0)^2 \leq \frac{c_{g,\gt}}{2} + C_{g,\gt}  \widetilde{M}\right) = o_p(T^{-\delta})
\]
for any $\delta>0$.

The last term is also $o(T^{-\delta})$. First, denote $c$ as the minimum of the collection of $c_{g,\gt}$ over all $g\neq \gt$. Then, choose 
\begin{equation}\label{cond:eta}
	\eta^* \leq \frac{c}{8\left( (5C_{g,\gt} + C_2 + 6)\widetilde{M} + (C_1 + C_{g,\gt} + 2) \sqrt{\widetilde{M}} + (C_3 + C_{g,\gt} + 1)\right)}
\end{equation}
which does not depend on $T$. 

Then for such $\eta^*$ I have for all $g \neq \gt$ the bound
\begin{gather*}
	 \mathbb{P}\Biggr( 
\sum_{t=1}^T (\alpha_{\gt t}^0 - \alpha_{g t}^0)u_{it} 
\leq  - \frac{c_{g,\gt}}{4}  + \eta^*\left( (5C_{g,\gt} + C_2 + 6)\widetilde{M} + (C_1 + C_{g,\gt} + 2) \sqrt{\widetilde{M}} + (C_3 + C_{g,\gt} + 1)\right) + \frac{1}{T} C_0\Biggr)\\
\leq \mathbb{P} \left( \frac{1}{T}
\sum_{t=1}^T (\alpha_{\gt t}^0 - \alpha_{g t}^0)u_{it} 
\leq -\frac{c_{g,\gt}}{8} + \frac{1}{T} C_0\right).
\end{gather*}

By Assumption \ref{as:ae}($b$), the process $\{(\alpha_{\gt t}^0 - \alpha_{gt}^0)u_{it}\}$ has zero mean and is strongly mixing with faster-than-polynomial decay rate. Moreover, for all $i$, $t$, and $m>0$,
\[
	\mathbb{P}(\vert (\alpha_{\gt t}^0 - \alpha_{gt}^0)u_{it}\vert > m) \leq \mathbb{P}\left( \vert u_{it} \vert > \frac{m}{2 \sup_{\alpha_t \in \mathcal{A}}}\right),
\]
so $\{(\alpha_{\gt t}^0 - \alpha_{gt}^0)u_{it}\}$ also satisfies the tail condition of Assumption \ref{as:ae}($c$) with a different constant $\widetilde{b}>0$ rather than $b>0$.

I then apply Lemma \ref{lemma:b3} with $z_t = (\alpha_{\gt t}^0 - \alpha_{g t}^0)u_{it} $ and $z = z_{T^*} = \frac{c_{g,\gt}}{8} - \frac{1}{T^*}C_0 >0$ for $T^*$ large enough. So, for $T$ large enough I see that
\[
	\mathbb{P} \left( \frac{1}{T}
\sum_{t=1}^T (\alpha_{\gt t}^0 - \alpha_{g t}^0)u_{it} 
\leq -\frac{c_{g,\gt}}{8} + \frac{1}{T}C_0 \right) = o(T^{-\delta})
\]

Then, combining all results I have
\[
	\frac{1}{N}\sum_{i=1}^N \sum_{g=1}^G \mathbb{P}\left( \widetilde{Z}_{ig} = 1\right) \leq G(G-1) 
\sup_{i=1,\dots,N} \mathbb{P}\left(\frac{1}{N} \sum_{i=1}^N \norm{x_{it}} \geq \widetilde{M}\right) + o(T^{-\delta}) = o(T^{-\delta}).
\]
To complete the proof, for $\eta$ small enough (satisfying \eqref{cond:eta}) and $T$ large enough, I have for all $\delta > 0$ and for all $\varepsilon >0$, 
\begin{align*}
	\mathbb{P}\left(\sup_{(\theta,\alpha)\in \mathcal{N}_\eta} \frac{1}{N}\sum_{i=1}^N \mathds{1}\{\widehat{g}_{i}(\theta,\alpha,\widehat{\gamma}_N) \neq g_i^0\} > \varepsilon T^{-\delta}\right) &\leq \mathbb{P}\left( \frac{1}{N}\sum_{i=1}^N\frac{1}{G}\sum_{g=1}^G \widetilde{Z}_{ig} > \varepsilon T^{-\delta}\right)\\ &\leq \frac{\Ex{\frac{1}{N}\sum_{i=1}^N\frac{1}{G}\sum_{g=1}^G \widetilde{Z}_{ig}}}{\varepsilon T^{-\delta}} = o(1)
\end{align*}
where I used the Markov inequality. Therefore, 
	\[
		\sup_{(\theta,\alpha) \in \mathcal{N}_\eta} \frac{1}{N}\sum_{i =1}^N \mathds{1}\{\widehat{g}_i(\theta,\alpha,\widehat{\gamma}_N) \neq g_i^0\} = o_p(T^{-\delta}).
	\]
\end{proof}

\textnormal{ With this lemma, I can show the asymptotic properties claimed in the proof for $\widehat{\theta}, \widehat{\alpha}$ and $\widehat{g}$. Define}
\begin{align*}
	\widehat{Q}(\theta,\alpha) &= \sum_{g=1}^G \left(\frac{1}{N}\sum_{j=1}^N \mathds{1}\{\widehat{g}_j(\theta,\alpha) = g\}\right) \sqrt{\widehat{Q}_g(\theta,\alpha)}\\
		&\widehat{Q}_g(\theta,\alpha) = \frac{1}{T\sum_{j=1}^N \mathds{1}\{\widehat{g}_j(\theta,\alpha) = g\}} \sum_{i=1}^N \sum_{t=1}^T \mathds{1}\{\widehat{g}_i(\theta,\alpha) = g\}(y_{it} - x_{it}'\theta - \alpha_{gt})^2\\
\widetilde{Q}(\theta,\alpha) &=\sum_{g=1}^G \left(\frac{1}{N}\sum_{j=1}^N \mathds{1}\{g_j^0 = g\}\right) \sqrt{\widetilde{Q}_g(\theta,\alpha)}\\
		&\widetilde{Q}_g(\theta,\alpha) = \frac{1}{T\sum_{j=1}^N \mathds{1}\{g_j^0 = g\}} \sum_{i=1}^N \sum_{t=1}^T \mathds{1}\{g_i^0 = g\}(y_{it} - x_{it}'\theta - \alpha_{gt})^2\\
\end{align*}
\textnormal{which are the objective functions from the proof of Proposition \ref{prop:conswgfe}, fixed at the WC groupings and true groupings, respectively. Denote $(\widehat{\theta}, \widehat{\alpha})$ and $(\widetilde{\theta}, \widetilde{\alpha})$ as the minimizers of these functions, respectively.}

\begin{lemma}
For $\eta >0 $ small enough and for all $\delta >0$,
\[
	\sup_{(\theta,\alpha)\in \mathcal{N}_\eta} \left\vert \widehat{Q}(\theta,\alpha) -\widetilde{Q}(\theta,\alpha)  \right\vert = o_p(T^{-\delta}).
\]
\end{lemma}
\begin{proof}

Let $\eta>0$ and consider $(\theta,\alpha) \in \mathcal{N}_\eta$. Note that $P_g(\widehat{\gamma}(\theta,\alpha)) = \frac{1}{N}\sum_{j=1}^N\mathds{1}\{\widehat{g}_i(\theta,\alpha) = g\}$. I will suppress the $\widehat{\gamma}_N$ notation in the optimal group assignment. Consider the difference
\begin{align*}
\widehat{Q}(\theta,\alpha) -\widetilde{Q}(\theta,\alpha) &= \sum_{g=1}^G \frac{\left(\frac{1}{N}\sum_{j=1}^N \mathds{1}\{\widehat{g}_j(\theta,\alpha) = g\}\right)^2\widehat{Q}_g(\theta,\alpha) - \left(\frac{1}{N}\sum_{j=1}^N \mathds{1}\{g_j^0 = g\}\right)^2 \widetilde{Q}_g(\theta,\alpha)}{\left(\frac{1}{N}\sum_{j=1}^N \mathds{1}\{\widehat{g}_j(\theta,\alpha) = g\}\right) \sqrt{\widehat{Q}_g(\theta,\alpha)} + \left(\frac{1}{N}\sum_{j=1}^N \mathds{1}\{g_j^0 = g\}\right) \sqrt{\widetilde{Q}_g(\theta,\alpha)}}
\end{align*}
where I have used the identity \eqref{id:alg}. Since the denominators are all $O_p(1)$ for each $g = 1,\dots, G$, I focus on the numerator and note that
\[
	\left(\frac{1}{N}\sum_{j=1}^N \mathds{1}\{\widehat{g}_j(\theta,\alpha) = g\}\right)^2 \times \frac{1}{T \sum_{j=1}^N \mathds{1}\{\widehat{g}_j(\theta,\alpha) = g\}} = \frac{1}{N}\sum_{j=1}^N \mathds{1}\{\widehat{g}_j(\theta,\alpha) = g\} \times \frac{1}{NT}
\]
so, for all $g = 1,\dots,G$,
\begin{gather*}
\left(\frac{1}{N}\sum_{j=1}^N \mathds{1}\{\widehat{g}_j(\theta,\alpha) = g\}\right)^2\widehat{Q}_g(\theta,\alpha) - \left(\frac{1}{N}\sum_{j=1}^N \mathds{1}\{g_j^0 = g\}\right)^2 \widetilde{Q}_g(\theta,\alpha)\\
=\frac{1}{N}\sum_{i=1}^N\left[ \left(\frac{1}{N}\sum_{j=1}^N \mathds{1}\{\widehat{g}_j(\theta,\alpha) = g\}\right)\mathds{1}\{\widehat{g}_i(\theta,\alpha) = g\} - \left(\frac{1}{N}\sum_{j=1}^N \mathds{1}\{g_j^0 = g\}\right)\mathds{1}\{g_i^0 = g\} \right]\frac{1}{T}\sum_{t=1}^T (y_{it} - x_{it}'\theta - \alpha_{gt})^2\\
\leq \frac{1}{N}\sum_{i=1}^N\underbrace{\left\vert \left(\frac{1}{N}\sum_{j=1}^N \mathds{1}\{\widehat{g}_j(\theta,\alpha) = g\}\right)\mathds{1}\{\widehat{g}_i(\theta,\alpha) = g\} - \left(\frac{1}{N}\sum_{j=1}^N \mathds{1}\{g_j^0 = g\}\right)\mathds{1}\{g_i^0 = g\} \right\vert}_{A_{gi}}\frac{1}{T}\sum_{t=1}^T (y_{it} - x_{it}'\theta - \alpha_{gt})^2.
\end{gather*}
Then
\begin{align*}
A_{gi} \leq& \mathds{1}\{\widehat{g}_i(\theta,\alpha) = g\} \mathds{1}\{g_i^0 = g\}\frac{1}{N} \sum_{j=1}^N \underbrace{\left\vert\mathds{1}\{\widehat{g}_j(\theta,\alpha) = g\} - \mathds{1}\{g_j^0 = g\} \right\vert}_{B_g}\\
&+ \mathds{1}\{\widehat{g}_i(\theta,\alpha) \neq g\} \mathds{1}\{g_i^0 = g\} \frac{1}{N}\sum_{j=1}^N \mathds{1}\{g_j^0 = g\}\\
&+ \mathds{1}\{\widehat{g}_i(\theta,\alpha) = g\} \mathds{1}\{g_i^0 \neq g\}  \frac{1}{N}\sum_{j=1}^N \mathds{1}\{\widehat{g}_j(\theta,\alpha) = g\}.
\end{align*}
Note that 
\begin{itemize}
	\item $\mathds{1}\{\widehat{g}_i(\theta,\alpha) \neq g\} \mathds{1}\{g_i^0 = g\} = \mathds{1}\{\widehat{g}_i(\theta,\alpha) \neq g_i^0\} \mathds{1}\{g_i^0 = g\} \leq \mathds{1}\{\widehat{g}_i(\theta,\alpha) \neq g_i^0\} $
	\item $ \mathds{1}\{\widehat{g}_i(\theta,\alpha) = g\} \mathds{1}\{g_i^0 = g\}=  \mathds{1}\{\widehat{g}_i(\theta,\alpha) = g_i^0\} \mathds{1}\{g_i^0 = g\} \leq 1$
	\item $B_g = \begin{cases} 0 \>\>\> \text{if} \>\>\> (\widehat{g}_j(\theta,\alpha) = g) \text{ and } g_{j}^0 = g \text{ or }  (\widehat{g}_j(\theta,\alpha) \neq g \text{ and } g_{j}^0 \neq g)\\ 1  \>\>\> \text{if} \>\>\>  (\widehat{g}_j(\theta,\alpha) = g \text{ and } g_{j}^0 \neq g) \text{ or }  (\widehat{g}_j(\theta,\alpha) \neq g \text{ and } g_{j}^0 = g)\end{cases}$
which gives us\[B_g = \mathds{1}\{\widehat{g}_j(\theta,\alpha) = g\}\mathds{1}\{g_j^0 \neq g\} + \mathds{1}\{\widehat{g}_j(\theta,\alpha) \neq g_j^0\}\mathds{1}\{g_j^0 = g\} \leq \mathds{1}\{\widehat{g}_j(\theta,\alpha) = g\}\mathds{1}\{g_j^0 \neq g\} + \mathds{1}\{\widehat{g}_j(\theta,\alpha) \neq g_j^0\}\]
\end{itemize}
where I deliberately find bounds that leverage Lemma \ref{lemma:missassign}. Then,
\begin{align*}
A_{gi} \leq& \frac{1}{N} \sum_{j=1}^N\mathds{1}\{\widehat{g}_j(\theta,\alpha) = g\}\mathds{1}\{g_j^0 \neq g\} + \frac{1}{N} \sum_{j=1}^N\mathds{1}\{\widehat{g}_j(\theta,\alpha) \neq g_j^0\}\\
&+ \mathds{1}\{\widehat{g}_i(\theta,\alpha) \neq g_i^0\} \times 1\\
&+ \mathds{1}\{\widehat{g}_i(\theta,\alpha) = g\} \mathds{1}\{g_i^0 \neq g\} \times 1\\
\leq& \frac{1}{N} \sum_{j=1}^N\mathds{1}\{\widehat{g}_j(\theta,\alpha) = g\}\mathds{1}\{g_j^0 \neq g\} + \sup_{(\theta,\alpha) \in \mathcal{N}_\eta}\frac{1}{N} \sum_{j=1}^N\mathds{1}\{\widehat{g}_j(\theta,\alpha) \neq g_j^0\}\\
&+ \mathds{1}\{\widehat{g}_i(\theta,\alpha) \neq g_i^0\} + \mathds{1}\{\widehat{g}_i(\theta,\alpha) = g\} \mathds{1}\{g_i^0 \neq g\}.
\end{align*}
Now, by our assumptions, I have for any $(\theta,\alpha)\in\Theta \times \mathcal{A}^{GT}$, as $T\to \infty$
\[
	\frac{1}{T} \sum_{t=1}^T (y_{it} - x_{it}'\theta - \alpha_{gt})^2 = O_p(1)
\]
for any $g=1,\dots,G$. Hence,
\begin{align*}
&\frac{1}{N}\sum_{i=1}^NA_{gi}\frac{1}{T}\sum_{t=1}^T (y_{it} - x_{it}'\theta - \alpha_{gt})^2 \\
\leq& O_p(1)\Biggr[\frac{1}{N} \sum_{j=1}^N\mathds{1}\{\widehat{g}_j(\theta,\alpha) = g\}\mathds{1}\{g_j^0 \neq g\} + \sup_{(\theta,\alpha) \in \mathcal{N}_\eta}\frac{1}{N} \sum_{j=1}^N\mathds{1}\{\widehat{g}_j(\theta,\alpha) \neq g_j^0\}\\
&+ \frac{1}{N} \sum_{i=1}^N \mathds{1}\{\widehat{g}_i(\theta,\alpha) \neq g_i^0\} + \frac{1}{N} \sum_{i=1}^N \mathds{1}\{\widehat{g}_i(\theta,\alpha) = g\} \mathds{1}\{g_i^0 \neq g\} \Biggr].
\end{align*}

Then, by Lemma \ref{lemma:missassign},
\begin{align*}
\left\vert \widehat{Q}(\theta,\alpha) -\widetilde{Q}(\theta,\alpha) \right\vert \leq& O_p(1) \Biggr[\sum_{g=1}^G \frac{1}{N} \sum_{j=1}^N\mathds{1}\{\widehat{g}_j(\theta,\alpha) = g\}\mathds{1}\{g_j^0 \neq g\} + G \sup_{(\theta,\alpha) \in \mathcal{N}_\eta}\frac{1}{N} \sum_{j=1}^N\mathds{1}\{\widehat{g}_j(\theta,\alpha) \neq g_j^0\} \\
&+ \sum_{g=1}^G  \frac{1}{N} \sum_{i=1}^N \mathds{1}\{\widehat{g}_i(\theta,\alpha) \neq g_i^0\}  + \sum_{g=1}^G\frac{1}{N} \sum_{i=1}^N \mathds{1}\{\widehat{g}_i(\theta,\alpha) = g\} \mathds{1}\{g_i^0 \neq g\} \Biggr]\\
& = o_p(T^{-\delta}).
\end{align*}

\end{proof}

\begin{proof}[Proof of the Proposition.]

For any $\eta > 0$, the consistency of $(\widehat{\theta},\widehat{\alpha})$ and $(\widetilde{\theta},\widetilde{\alpha})$ imply as $N,T \to \infty$,
\begin{gather}
	\mathbb{P}\left(\left(\widehat{\theta},\widehat{\alpha}\right) \notin \mathcal{N}_\eta\right) \to 0,\\
	\mathbb{P}\left(\left(\widetilde{\theta},\widetilde{\alpha}\right) \notin \mathcal{N}_\eta\right) \to 0.
\end{gather}
Then, for any $\varepsilon >0$,
\[
	\mathbb{P} \left[ \left\vert\widehat{Q}\left(\widehat{\theta},\widehat{\alpha}\right) - \widetilde{Q}\left(\widehat{\theta},\widehat{\alpha}\right) \right\vert > \varepsilon T^{-\delta} \right] \leq \mathbb{P}\left(\left(\widehat{\theta},\widehat{\alpha}\right) \notin \mathcal{N}_\eta\right) + \mathbb{P} \left[ \sup_{(\theta, \alpha) \in \mathcal{N}_\eta} \left\vert \widehat{Q}(\theta,\alpha) -\widetilde{Q}(\theta,\alpha) \right\vert > \varepsilon T^{-\delta} \right] = o(1)
\]
hence
\[
\widehat{Q}\left(\widehat{\theta},\widehat{\alpha}\right) - \widetilde{Q}\left(\widehat{\theta},\widehat{\alpha}\right)  = o_p(T^{-\delta})
\]
and similarly for 
\[
\widehat{Q}\left(\widetilde{\theta},\widetilde{\alpha}\right) - \widetilde{Q}\left(\widetilde{\theta},\widetilde{\alpha}\right)  = o_p(T^{-\delta}).
\]

Then, using these and the minimizer definition of the feasible and infeasible WC estimators,
\[
	0 \leq \widetilde{Q}\left(\widehat{\theta},\widehat{\alpha} \right) - \widetilde{Q}\left(\widetilde{\theta},\widetilde{\alpha} \right) 
= \widetilde{Q}\left(\widehat{\theta},\widehat{\alpha} \right) - \widehat{Q}\left(\widetilde{\theta},\widetilde{\alpha} \right) + o_p(T^{-\delta}) \leq o_p(T^{-\delta}).
\]
Therefore,
\[
	\widetilde{Q}\left(\widehat{\theta},\widehat{\alpha} \right) - \widetilde{Q}\left(\widetilde{\theta},\widetilde{\alpha} \right)  = o_p(T^{-\delta}).
\]
To show consistency of the parameter estimators, denote $\widetilde{u}_{it}^2 = y_{it} - x_{it}'\widetilde{\theta} - \widetilde{\alpha}_{g_i^0t}$ and consider the following
\begin{gather*}
	\widetilde{Q}\left(\widehat{\theta},\widehat{\alpha} \right) - \widetilde{Q}\left(\widetilde{\theta},\widetilde{\alpha} \right)=
\sum_{g=1}^G \left(\frac{1}{N}\sum_{j=1}^N \mathds{1}\{g_j^0 = g\} \right) \left[ \sqrt{\widetilde{Q}_g\left(\widehat{\theta},\widehat{\alpha} \right)} - \sqrt{\widetilde{Q}_g\left(\widetilde{\theta},\widetilde{\alpha} \right)} \right]\\
= \sum_{g=1}^G \left(\frac{1}{N}\sum_{j=1}^N \mathds{1}\{g_j^0 = g\} \right) \frac{ \widetilde{Q}_g\left(\widehat{\theta},\widehat{\alpha} \right) - \widetilde{Q}_g\left(\widetilde{\theta},\widetilde{\alpha}\right) }{ \sqrt{\widetilde{Q}_g\left(\widehat{\theta},\widehat{\alpha} \right)} + \sqrt{\widetilde{Q}_g\left(\widetilde{\theta},\widetilde{\alpha} \right)}}\\
= \sum_{g=1}^G \left(\frac{1}{N}\sum_{j=1}^N \mathds{1}\{g_j^0 = g\} \right) \frac{ \widetilde{Q}_g\left(\widehat{\theta},\widehat{\alpha} \right) - \widetilde{Q}_g\left(\widetilde{\theta},\widetilde{\alpha}\right) }{ 2\sqrt{\widetilde{Q}_g\left(\widetilde{\theta},\widetilde{\alpha} \right)} + o_p(T^{-\delta})}\\
= \sum_{g=1}^G \frac{ \frac{1}{NT} \sum_{i=1}^N\mathds{1}\{g_i^0 = g\} \sum_{t=1}^T \left( y_{it} - x_{it}'\widehat{\theta} - \widehat{\alpha}_{gt}\right)^2 -  \left( y_{it} - x_{it}'\widetilde{\theta} - \widetilde{\alpha}_{gt}\right)^2}{ 2\sqrt{\widetilde{Q}_g\left(\widetilde{\theta},\widetilde{\alpha} \right)} + o_p(T^{-\delta})}\\
= \sum_{g=1}^G \frac{ \frac{1}{NT} \sum_{i=1}^N \mathds{1}\{g_i^0 = g\}\sum_{t=1}^T \left( \widetilde{u}_{it} + x_{it}'\left(\widetilde{\theta} - \widehat{\theta}\right) + \left(\widetilde{\alpha}_{gt} - \widehat{\alpha}_{gt} \right)\right)^2 -  \widetilde{u}_{it}^2}{ 2\sqrt{\widetilde{Q}_g\left(\widetilde{\theta},\widetilde{\alpha} \right)} } + o_p(T^{-\delta})\\
= \sum_{g=1}^G \frac{ \frac{1}{NT} \sum_{i=1}^N \mathds{1}\{g_i^0 = g\}\sum_{t=1}^T \left( x_{it}'\left(\widetilde{\theta} - \widehat{\theta}\right) + \left(\widetilde{\alpha}_{gt} - \widehat{\alpha}_{gt} \right)\right)^2 +  2\widetilde{u}_{it}\left( x_{it}'\left(\widetilde{\theta} - \widehat{\theta}\right) + \left(\widetilde{\alpha}_{gt} - \widehat{\alpha}_{gt} \right)\right)}{ 2\sqrt{\widetilde{Q}_g\left(\widetilde{\theta},\widetilde{\alpha} \right)} } + o_p(T^{-\delta})\\
= \sum_{g=1}^G \frac{ \frac{1}{NT} \sum_{i=1}^N \mathds{1}\{g_i^0 = g\}\sum_{t=1}^T \left( x_{it}'\left(\widetilde{\theta} - \widehat{\theta}\right) + \left(\widetilde{\alpha}_{gt} - \widehat{\alpha}_{gt} \right)\right)^2}{ 2\sqrt{\widetilde{Q}_g\left(\widetilde{\theta},\widetilde{\alpha} \right)} }\\
 +  \sum_{g=1}^G \frac{1}{ \sqrt{\widetilde{Q}_g\left(\widetilde{\theta},\widetilde{\alpha} \right)} }\frac{1}{NT} \sum_{i=1}^N \mathds{1}\{g_i^0 = g\}\sum_{t=1}^T\widetilde{u}_{it}\left( x_{it}'\left(\widetilde{\theta} - \widehat{\theta}\right) + \left(\widetilde{\alpha}_{gt} - \widehat{\alpha}_{gt} \right)\right) + o_p(T^{-\delta}).
\end{gather*}
The second term is zero and to see that, recall the first-order conditions for $\gamma = \gamma^0$ fixed:
\begin{gather*}
0 = \sum_{g=1}^G \frac{1}{\sqrt{\widetilde{Q}_g(\widetilde{\theta}, \widetilde{\alpha})}} \frac{1}{NT} \sum_{i=1}^N \sum_{t=1}^T \mathds{1}\{g_i^0 = g\}x_{it}\left(y_{it} - x_{it}'\widetilde{\theta} - \widetilde{\alpha}_{gt}\right)\\
0 = \frac{1}{N} \sum_{i=1}^N \mathds{1}\{g_i^0 = g\}\left (y_{it} - x_{it}'\widetilde{\theta} - \widetilde{\alpha}_{gt}\right), \>\>\> \text{ for } \>\> t = 1,\dots,T
\end{gather*}
so that the second term is equal to
\begin{gather*}
	\sum_{g=1}^G \frac{1}{ \sqrt{\widetilde{Q}_g\left(\widetilde{\theta},\widetilde{\alpha} \right)} }\left[\frac{1}{NT} \sum_{i=1}^N \mathds{1}\{g_i^0 = g\}\sum_{t=1}^T\widetilde{u}_{it} x_{it}'\left(\widetilde{\theta} - \widehat{\theta}\right) + \frac{1}{T} \sum_{t=1}^T \left(\widetilde{\alpha}_{gt} - \widehat{\alpha}_{gt} \right)\frac{1}{N}\sum_{i=1}^N \mathds{1}\{g_i^0 = g\}\widetilde{u}_{it}\right]\\
=\left(\widetilde{\theta} - \widehat{\theta}\right)'\left[ \sum_{g=1}^G \frac{1}{ \sqrt{\widetilde{Q}_g\left(\widetilde{\theta},\widetilde{\alpha} \right)} }\frac{1}{NT} \sum_{i=1}^N \sum_{t=1}^T \mathds{1}\{g_i^0 = g\}x_{it}\widetilde{u}_{it}\right] + 0\\
=0
\end{gather*}
by using the first-order conditions of $\widetilde{Q}$. 

Therefore,
\begin{gather*}
\widetilde{Q}\left(\widehat{\theta},\widehat{\alpha} \right) - \widetilde{Q}\left(\widetilde{\theta},\widetilde{\alpha} \right)\\
= \sum_{g=1}^G  \frac{1}{{ 2\sqrt{\widetilde{Q}_g\left(\widetilde{\theta},\widetilde{\alpha} \right)} }}\frac{1}{NT} \sum_{i=1}^N \sum_{t=1}^T \mathds{1}\{g_i^0 = g\} \left( x_{it}'\left(\widetilde{\theta} - \widehat{\theta}\right) + \left(\widetilde{\alpha}_{gt} - \widehat{\alpha}_{gt} \right)\right)^2 + o_p(T^{-\delta})\\
\geq \widetilde{C} \sum_{g=1}^G \frac{1}{NT} \sum_{i=1}^N \sum_{t=1}^T \mathds{1}\{g_i^0 = g\} \left( x_{it}'\left(\widetilde{\theta} - \widehat{\theta}\right) + \left(\widetilde{\alpha}_{gt} - \widehat{\alpha}_{gt} \right)\right)^2 + o_p(T^{-\delta})\\
\geq \widetilde{C} \frac{1}{NT} \sum_{i=1}^N \sum_{t=1}^T  \left( x_{it}'\left(\widetilde{\theta} - \widehat{\theta}\right) + \left(\widetilde{\alpha}_{g_i^0 t} - \widehat{\alpha}_{g_i^0 t} \right)\right)^2 + o_p(T^{-\delta})\\
\geq \widetilde{C} \frac{1}{NT} \sum_{i=1}^N \sum_{t=1}^T  \left( x_{it}'\left(\widetilde{\theta} - \widehat{\theta}\right) + \overline{x}_{g_i^0 t}'\left(\widetilde{\theta} - \widehat{\theta}\right)\right)^2 + o_p(T^{-\delta})\\
=\widetilde{C} \left(\widetilde{\theta} - \widehat{\theta}\right)' \left[\frac{1}{NT} \sum_{i=1}^N \sum_{t=1}^T  \left( x_{it}- \overline{x}_{g_i^0 t}\right)\left( x_{it}- \overline{x}_{g_i^0 t}\right)' \right]\left(\widetilde{\theta} - \widehat{\theta}\right) + o_p(T^{-\delta})\\
\geq \widetilde{C}\widehat{\rho} \norm{\widetilde{\theta} - \widehat{\theta}}^2 + o_p(1) + o_p(T^{-\delta})
\end{gather*}
where $\max_{g} \sqrt{\widetilde{Q}_g(\widetilde{\theta}, \widetilde{\alpha})} < 1/\widetilde{C}$. Therefore, 
\[
\widehat{\theta} = \widetilde{\theta} + o_p(T^{-\delta}).
\]

To show $\widehat{\alpha} \to_p \alpha$, I use the fact that $\widehat{\theta} = \widetilde{\theta} + o_p(T^{-\delta})$:
\begin{gather*}
\widetilde{Q}\left(\widehat{\theta},\widehat{\alpha} \right) - \widetilde{Q}\left(\widetilde{\theta},\widetilde{\alpha} \right)\\
= \sum_{g=1}^G  \frac{1}{{ 2\sqrt{\widetilde{Q}_g\left(\widetilde{\theta},\widetilde{\alpha} \right)} }}\frac{1}{NT} \sum_{i=1}^N \sum_{t=1}^T \mathds{1}\{g_i^0 = g\} \left( x_{it}'\left(\widetilde{\theta} - \widehat{\theta}\right) + \left(\widetilde{\alpha}_{gt} - \widehat{\alpha}_{gt} \right)\right)^2 + o_p(T^{-\delta})\\
\geq O_p(1)\times\sum_{g=1}^G \frac{1}{N} \sum_{i=1}^N\mathds{1}\{g_i^0 = g\} \frac{1}{T}\sum_{t=1}^T  \left(\widetilde{\alpha}_{gt} - \widehat{\alpha}_{gt} \right)^2 + o_p(T^{-\delta})\\
= O_p(1) \frac{1}{NT} \sum_{i=1}^N \sum_{t=1}^T  \left(\widetilde{\alpha}_{g_i^0 t} - \widehat{\alpha}_{g_i^0 t} \right)^2 + o_p(T^{-\delta})
\end{gather*}
so $\frac{1}{T}\norm{\widetilde{\alpha}_{g} - \widehat{\alpha}_{g}}^2 = o_p(T^{-\delta})$ for any $g= 1,\dots,G$ therefore
\[
	\norm{\widetilde{\alpha}_{g} - \widehat{\alpha}_{g}}^2 = o_p(T^{1-\delta}).
\] 
This holds for any $\delta > 0$ so I have the result. 

To show consistency of group assignments, by a union bound I have
\[
	\mathbb{P}\left(\sup_{i = 1,\dots,N} \left\vert \widehat{g}_i\left(\widehat{\theta},\widehat{\alpha} \right) - g_i^0 \right\vert > 0 \right) \leq 
\mathbb{P}\left( \left(\widehat{\theta},\widehat{\alpha}\right) \notin \mathcal{N}_\eta\right) + N \sup_{i=1,\dots,N} \mathbb{P}\left( \left(\widehat{\theta},\widehat{\alpha}\right) \in \mathcal{N}_\eta, \widehat{g}_i\left(\widehat{\theta},\widehat{\alpha}\right) \neq g_i^0 \right).
\]
Now, taking $\eta$ satisfying \eqref{cond:eta} gives us $\mathbb{P}\left( \left(\widehat{\theta},\widehat{\alpha}\right) \notin \mathcal{N}_\eta\right) = o(1)$. Then, from the proof of Lemma \ref{lemma:missassign} I know 
\[
	\sup_{(\theta,\alpha) \in \mathcal{N}_\eta}\mathds{1}\{\widehat{g}(\theta,\alpha) \neq g_i^0\} \leq \sum_{g=1}^G \widetilde{Z}_{ig}
\]
where $\widetilde{Z}_{ig}$ is given by \eqref{zsquiggle}. Then
\begin{align*}
N \sup_{i=1,\dots,N} \mathbb{P}\left( \left(\widehat{\theta},\widehat{\alpha}\right) \in \mathcal{N}_\eta, \widehat{g}_i\left(\widehat{\theta},\widehat{\alpha}\right) \neq g_i^0 \right) 
&= N \sup_{i=1,\dots,N} \Ex{\mathds{1}\{\left(\widehat{\theta},\widehat{\alpha}\right) \in \mathcal{N}_\eta\}\mathds{1}\{\widehat{g}(\theta,\alpha) \neq g_i^0\}}\\
&\leq N \sup_{i=1,\dots,N}\Ex{\mathds{1}\{\left(\widehat{\theta},\widehat{\alpha}\right) \in \mathcal{N}_\eta\} \sum_{g=1}^G \widetilde{Z}_{ig}}\\
&\leq N \sup_{i=1,\dots,N} \sum_{g=1}^G \mathbb{P}\left(\widetilde{Z}_{ig} = 1 \right)\\
&\leq  N \left [ G(G-1)\sup_{i=1,\dots,N} \mathbb{P}\left(\frac{1}{T} \sum_{t=1}^T \norm{x_{it}} \geq \widetilde{M} \right) + o(T^{-\delta})\right]\\
&= o(NT^{-\delta})
\end{align*}
for all $\delta>0$ and this completes the proof.

\end{proof}

\section{Generalized Grouped Fixed Effects Estimation}\label{app:ggfe}

 The \emph{generalized grouped fixed-effects} (GGFE) estimator for model \eqref{themodel} is the solution to the minimization problem 
\begin{equation}\label{gen:obj}
	\widehat{\delta} = (\widehat{\theta}, \widehat{\alpha}, \widehat{\gamma}) 
= \argmin_{\delta \in \Theta \times \mathcal{A}^{GT} \times \Gamma_G^N} {\it tr}[\bm{\Omega}_\delta]
\end{equation}
subject to 
\begin{gather}
	\bm{\bm{\Omega}}_{\delta} = \sum_{g=1}^G \left(\bm{\bm{\Omega}}_{\delta}^{1/2}\widehat{\bm{\bm{\Sigma}}}_{g}(\delta) \bm{\bm{\Omega}}_{\delta}^{1/2}\right)^{1/2} P_g(\delta)\label{baryvar}\\
\widehat{\bm{\bm{\bm{\Sigma}}}}_g(\delta) = \sum_{i=1}^N  \frac{\mathds{1}\{g_i = g\}}{\sum_{j=1}^N \mathds{1}\{g_j = g\} } \left(y_i  - x_i'\theta - \alpha_{g_i}\right)\cdot\left(y_i  - x_i'\theta - \alpha_{g_i}\right)' \label{groupcov} \\
P_g(\delta) = \frac{1}{N} \sum_{j=1}^N \mathds{1}\{g_j = g\}\label{emf}.
\end{gather}
where the square root of a matrix is understood as the principle square root. The objective function contains $\bm{\bm{\Omega}}_{\delta}$ as the solution to a nonlinear matrix equation \eqref{baryvar}, which can be obtained by a fixed-point algorithm, which is guaranteed to converge to the unique fixed-point when the group covariances \eqref{groupcov} are non-singular (\cite{alvarez:2016}). In addition, it resembles a geometric mean of the group covariances weighted by the empirical mass function of the groups given by \eqref{emf}.

\subsection{Computation} 

For the form of the gradient $\partial_\gamma\text{{\it tr}}[ \bm{\Omega}_{\delta} ] $ with respect to $\gamma$ in the relaxed set $\Delta^G$ considering soft assignment, see Appendix \ref{comp-grad}. With this I can construct an algorithm to solve \eqref{gen:obj}.

\begin{algorithm}\label{WC:alg}(Gradient descent for GGFE Estimation).
\begin{itemize}
\item[1.] Initialize $\theta^{(0)}$; set the $\alpha^{(0)}$ as $G$ randomly chosen $v_i = y_i - x_i\theta^{(0)}$; create initial assignments $g_i^{(0)}$ by assigning $v_i$ to the closest $\{\alpha_{g_i}^{(0)}\}$; and calculate the $\{\widehat{\bm{\Sigma}}_g^{(0)}(\delta^{(0)})\}$. Set $s = 1$. \newline

\item[2.] (Assignment Step). Calculate $\partial_{P_g^i}\text{tr}[ \bm{\Omega}_\delta ]$ for all $i$ and $g$. Assign according to
\[
	g_i^{(s+1)} \leftarrow \argmin_g \partial_{P_g^i}\text{{\it tr}}[ \bm{\Omega}_{\delta} ] \bigg\vert_{\delta = \delta^{(s)}}
\]
and collect them in $\gamma^{(s+1)} = (g_1^{(s+1)}, \dots, g_N^{(s+1)})$.\\
\item[3.] (Update Step). Update $\alpha^{(s+1)}$ and $\theta^{(s+1)}$ according to 
\[
	(\theta^{(s+1)},\alpha^{(s+1)}) = \argmin_{(\theta,\alpha)}\text{{\it tr}}[ \bm{\Omega}_{(\theta,\alpha, \gamma^{(s+1)})} ],
\]
then update $\{\widehat{\bm{\Sigma}}_g^{(s+1)}(\delta^{(s+1)})\}$. \newline

\item[4.] If $g_i^{(s)} = g_i^{(s + 1)}$ for all $i$, stop. Otherwise, set $s \leftarrow s + 1$ and go back to step 2.

\end{itemize}

\end{algorithm}

\section{Computation}

\subsection{Gradient of the criterion function with respect to assignments}\label{comp-grad}

Calculating the parameter values and labeling matrix that solve (\ref{wgfe:obj}) is done through a gradient descent algorithm. For given parameter values $\alpha$ and $\theta$, the existence of the gradient with respect to class assignments $P^i$ of the objective function for $i=1,\dots,N$ is shown in Appendix B and calculated in Appendix C of \cite{yang:2022} as 
\begin{equation}\label{gradL}
\nabla_{P^i}\text{tr}[\widehat{\bm{\Omega}}_\delta]
= \sum_{g=1}^G \vect({\bf I})' \cdot \bm{W}_g 
\cdot \vect\bigg((y_i - x_i\theta - \alpha_{g_i})\cdot(y_i - x_i\theta - \alpha_{g_i})' + \widehat{\bm{\Sigma} }_g \bigg)\vec{e}_g
\end{equation}
where $\vec{e}_g$ is a vector of zeros except with a one in the $g$th entry and the $\bm{W}_g$ matrices are
\begin{gather}
	\bm{W}_g = (\bm{\Omega}_\delta^{1/2} \otimes \bm{\Omega}_\delta^{1/2} )
	\times \bigg[ \sum_{h=1}^G P_h (\bm{U}_h \otimes\bm{U}_h)(\bm{D}_h^{1/2} \otimes {\bf I} + {\bf I} \otimes \bm{D}_h^{1/2} )^{-1} (\bm{D}_h^{1/2} \otimes \bm{D}_h^{1/2})(\bm{U}_h' \otimes \bm{U}_h') \bigg]^{-1} \notag\\
	\times \big[ (\bm{U}_g\otimes \bm{U}_g) (\bm{D}_g^{1/2} \otimes {\bf I} + {\bf I} \otimes \bm{D}_g^{1/2} )^{-1} (\bm{U}_g' \otimes \bm{U}_g')\big]
	(\bm{\Sigma}_\delta^{1/2} \otimes\bm{\Omega}_\delta^{1/2} )\label{weightgradL}
\end{gather}
and $\bm{U}_g$ and $\bm{D}_g$ are the corresponding orthonormal and diagonal matrices found from the eigendecompositions
$$
	\bm{\Omega}_\delta^{1/2}
\widehat{\bm{\Sigma}}_g 
\bm{\Omega}_\delta^{1/2}= \bm{U}_g \bm{D}_g \bm{U}_g' 
$$
for $g\in\{1,\dots, G\}$. Note that $\bm{D}$ and $\bm{U}$ depend on the parameter values.

In order to implement gradient descent I need to make updates at each iteration to the objective function, which is the variance of the unique solution to the nonlinear matrix equation (\ref{baryvar}). With the estimated cluster covariances given some parameter values $\alpha$ and $\theta$, I have the necessary terms to solve this equation by following the iteration scheme from \cite%
{alvarez:2016}: 
\begin{equation}\label{baryiter}
\bm{\Omega}_\delta(s+1)\leftarrow \bm{\Omega}_\delta^{-1/2}(s)\bigg(%
\sum_{g=1}^{G} P_g \bm{\Omega}_\delta^{1/2}(s)
 \widehat{\bm{\Sigma} }_g 
\bm{\Omega}_\delta^{1/2}(s) \bigg)^{2}\bm{\Omega}_\delta^{-1/2}(s)
\end{equation}%
using an arbitrary symmetric positive definite initialization matrix $\bm{\Omega}_\delta^{-1/2}(0)$, e.g. the identity matrix.

\subsection{A Variable Neighborhood Search (VNS) algorithm}\label{wgfe:vns}

\begin{algorithm}\label{alg:vns} (Variable Neighborhood Search (VNS)). 
\begin{itemize}
\item[1.] Set $iter_{\max}$ and $neigh_{\max}$ to desired values (usually both 10). 

\item[2.] Initialize $(\theta^{(0)}, \alpha^{(0)})$ and perform one assignment step to get an initial $g_{{\it init}}$. Set $g^*=g_{{\it init}} $. \newline

\item[3.] Set $n=1$. 

\item[4.] (Neighborhood jump). Relocate $n$ random individuals to $n$ other random groups and obtain a new grouping $g'$ and perform one update step to obtain $(\theta',\alpha')$. 

\item[5.] Set $(\theta^{(0)}, \alpha^{(0)}) = (\theta',\alpha')$ and apply Algorithm 3 to get a grouping $g$. 

\item[6.] (Local search). With grouping $g$ from step 5, systematically check all re-assignments of individuals $i \in \{1,\dots,N\}$ to groups $g \in \{1,\dots,G\}$, updating $g_i$ whenever the objective function decreases. Denote the resulting grouping by $g''$. 

\item[7.] If the objective function with $g''$ improves relative to the one with $g'$, then set $g^* = g''$ and go to step 3, otherwise set $n= n+1$ and go to step 8.

\item[8.] If $n \leq neigh_{\max}$, then go to step 4; otherwise go to step 9. 

\item[9.] Set $j \leftarrow j+1$. If $j > iter_{\max}$, then stop; otherwise go to step 3.
\end{itemize}
\end{algorithm}

\subsection{Initialization}

Initialization of Algorithm \ref{wgfe:algo} and \ref{alg:vns} is an important consideration and heavily influences the search for global minima using these heuristics. Currently using pooled OLS, two-way fixed effects estimates or a convex combination of them for $\theta$ make sense for initialization and seem to perform well, although no formal test of performance has been done. This is mentioned since it hasn't been proposed in the literature. Other approaches is to randomly draw these parameters from some distribution or to perform a subsample WGFE or GFE estimation to collect estimates that are then used in the full sample estimation. One could also obtain GFE estimates with the full sample and use those as initial parameters and groupings.

\end{document}